\documentclass[a4paper,12pt]{amsart}
\pdfoutput=1

\usepackage[sc]{mathpazo}
\usepackage{xy}
\usepackage[T1]{fontenc}
\usepackage{subfig}
\usepackage{graphicx}
\usepackage{bm}
\usepackage{amsmath}
\usepackage{amssymb}
\usepackage[active]{srcltx}
\usepackage{hyperref}

\usepackage{color}
\definecolor{red}{rgb}{1,0,0}

\definecolor{gre}{rgb}{0,0.7,0}

\definecolor{col}{rgb}{0,0.3,0.8}

\newtheorem{thm}{Theorem}[section]%
\newtheorem{lem}{Lemma}[section]%
\newtheorem{prop}{Proposition}[section]%
\newtheorem{cor}{Corollary}[section]%
\newtheorem{con}{Conjecture}[section]%
\theoremstyle{definition}

\theoremstyle{remark}
\theoremstyle{plain}

\numberwithin{equation}{section}

\def\CC{{\mathbb C}}

\def\QQ{{\mathbb Q}}

\def\RR{{\mathbb R}}
\def\TT{{\mathbb T}}

\def\ZZ{{\mathbb Z}}

\def\veck{{\text{\boldmath$k$}}}

\def\vecm{{\text{\boldmath$m$}}}
\def\vecM{{\text{\boldmath$M$}}}

\def\vecp{{\text{\boldmath$p$}}}

\def\vecs{{\text{\boldmath$s$}}}

\def\vect{{\text{\boldmath$t$}}}
\def\vecu{{\text{\boldmath$u$}}}
\def\vecv{{\text{\boldmath$v$}}}

\def\vecx{{\text{\boldmath$x$}}}

\def\vecX{{\text{\boldmath$X$}}}
\def\vecy{{\text{\boldmath$y$}}}
\def\vecY{{\text{\boldmath$Y$}}}

\def\vecz{{\text{\boldmath$z$}}}
\def\vecZ{{\text{\boldmath$Z$}}}

\def\vecalf{{\text{\boldmath$\alpha$}}}
\def\vecbeta{{\text{\boldmath$\beta$}}}

\def\veceta{{\text{\boldmath$\eta$}}}

\def\vecomega{{\text{\boldmath$\omega$}}}

\def\vectau{{\text{\boldmath$\tau$}}}

\def\vecphi{{\text{\boldmath$\phi$}}}

\def\vecxi{{\text{\boldmath$\xi$}}}
\def\vecXi{{\text{\boldmath$\Xi$}}}

\def\vecnull{{\text{\boldmath$0$}}}

\def\scrE{{\mathcal E}}
\def\scrF{{\mathcal F}}
\def\scrI{{\mathcal I}}
\def\scrH{{\mathcal H}}
\def\scrJ{{\mathcal J}}

\def\scrL{{\mathcal L}}

\def\scrQ{{\mathcal Q}}

\def\scrR{{\mathcal R}}
\def\scrS{{\mathcal S}}
\def\scrT{{\mathcal T}}

\def\scrW{{\mathcal W}}

\def\fH{{\mathfrak H}}

\def\Im{\operatorname{Im}}

\def\H{{\mathfrak H}}

\def\e{\mathrm{e}}
\def\i{\mathrm{i}}
\def\d{\mathrm{d}}

\def\diag{\operatorname{diag}}

\def\error{\operatorname{\scrE}}
\def\cl{{\operatorname{cl}}}

\def\C{\operatorname{C{}}}

\def\L{\operatorname{L{}}}

\def\Op{\operatorname{Op}}
\def\hOp{\widehat{\operatorname{Op}}}

\def\SL{\operatorname{SL}}

\def\Tr{\operatorname{Tr}}

\def\supp{\operatorname{supp}}

\def\HiS{{\operatorname{HS}}}

\def\GamG{\Gamma\backslash G}

\def\SLZ{\SL(2,\ZZ)}

\def\trans{\,^\mathrm{t}\!}

\def\colker{\Sigma}
\def\yy{\vecy}

\title[Quantum transport in a low-density periodic potential]{Quantum transport in a low-density periodic potential: homogenisation via homogeneous flows}

\author{Jory Griffin}
\address{Jory Griffin, Department of Mathematics and Statistics,
Queen's University,
Kingston ON, K7L 3N6, Canada}
\email{\tt j.griffin@queensu.ca}

\author{Jens Marklof}
\address{Jens Marklof, School of Mathematics, University of Bristol, Bristol BS8 1TW, U.K.} 
\email{\tt j.marklof@bristol.ac.uk}

\date{2 November 2018/20 March 2019}

\begin{document}

\begin{abstract} 
We show that the time evolution of a quantum wavepacket in a periodic potential converges in a combined high-frequency/Boltzmann-Grad limit, up to second order in the coupling constant, to terms that are compatible with the linear Boltzmann equation. This complements results of Eng and Erd\"os for low-density random potentials, where convergence to the linear Boltzmann equation is proved in all orders. We conjecture, however, that the linear Boltzmann equation fails in the periodic setting for terms of order four and higher. Our proof uses Floquet-Bloch theory, multi-variable theta series and equidistribution theorems for homogeneous flows. Compared with other scaling limits traditionally considered in homogenisation theory, the Boltzmann-Grad limit requires control of the quantum dynamics for longer times, which are inversely proportional to the total scattering cross section of the single-site potential.
\end{abstract}

\maketitle

%\tableofcontents

\section{Introduction}

The analysis of wave transport in periodic media plays an important role in explaining numerous physical phenomena, most notably in solid state physics, continuum mechanics and optics. A key challenge is the derivation of macroscopic transport equations from the underlying microscopic laws, and to thus describe effects on scales which are several orders of magnitude above the length scale given by the period of the medium. Semiclassical analysis and homogenisation theory have produced a remarkable collection of results in scaling limits where the characteristic wavelength is either much larger than the period (low-frequency homogenisation) or of the same or smaller order (high-frequency homogenisation); see for example \cite{Allaire05,Beniot,BLP,Birman,Craster10,Gerard91,Gerard97,HMC16,Markowich94,Panati03}.

In this paper we  study the limit when the diameter $2r$ of the  interaction region in each fundamental cell is significantly smaller than the period, and the wavelength $h$ is comparable to the interaction region, see Figure \ref{fig1}. 

\begin{figure}
\includegraphics[width=0.9\textwidth]{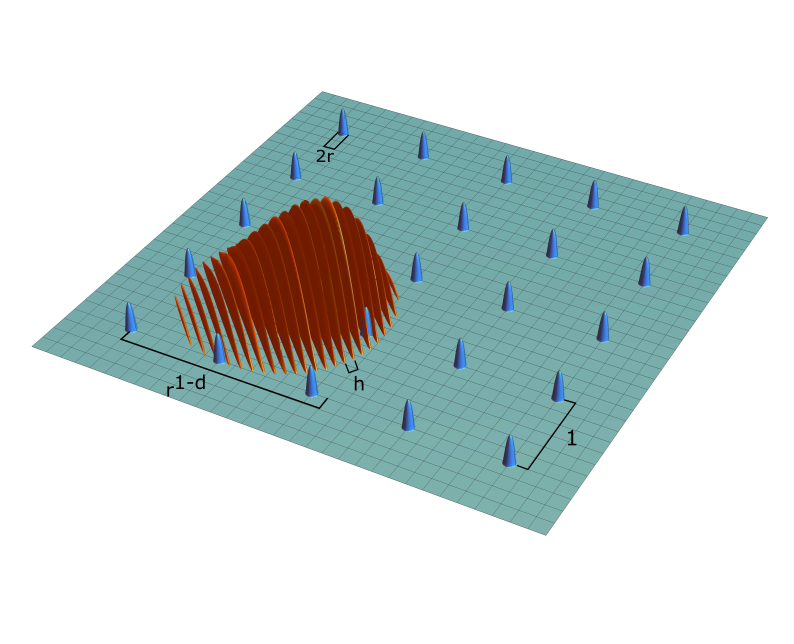}
\caption{Illustration of a wavepacket at time $t=0$ with wavelength $h$ in a $\ZZ^d$-periodic potential with interaction regions of diameter $2r$. For small $r$, the classical mean free path length in this setting is of the order $r^{1-d}$.} \label{fig1}
\end{figure}

Such a scaling, which is not traditionally discussed in high-frequency homogenisation, is motivated by the desire to understand the Boltzmann-Grad limit of particle transport in crystals. This problem is currently only understood (a) in the case of zero quasi-momentum \cite{Castella_WC,Castella_LD,Castella_Plagne_LD}, (b) in the classical limit \cite{CagliotiGolse,partIII,partI,partII,partIV}, and (c) when the medium is random rather than periodic, in both the classical \cite{Gallavotti69,Spohn78,Boldrighini83} and quantum setting \cite{Eng_Erdos} (see also \cite{Erdos_Yau,Spohn77} for the weak-coupling limit and \cite{Bal99,Bal2002,Bal2011} for related models). In the random setting---classical and quantum---the limit transport equation is proved to be the \textit{linear Boltzmann equation}, as predicted by Lorentz in 1905 \cite{Lorentz}. 

The linear Boltzmann equation for a particle density $f(t,\vecx,\vecy)$ at time $t$, where $\vecx$ denotes position and $\vecy$ momentum, is given by
\begin{equation}\label{LBeq}
\partial_t \, f(t,\vecx,\vecy) + \vecy \cdot \nabla_{\!\vecx}\, f(t,\vecx,\vecy) =  \rho(\vecx) \int_{\RR^d} \colker(\vecy,\vecy') \, [f(t,\vecx,\vecy') - f(t,\vecx,\vecy)]  \, \d \vecy' ,
\end{equation}
subject to initial data $f(0,\vecx,\vecy)=a(\vecx,\vecy)$. The collision kernel $\colker(\vecy,\vecy')$ is determinded by the single-site scatterering potential, and can be interpreted as the rate of particles with velocity $\vecy$ being scattered to velocity $\vecy'$ (or vice versa). The quantity $\rho(\vecx)$ denotes the macroscopic scatterer density at $\vecx$, which for a homogeneous medium means $\rho(\vecx)$ is constant. In the absence of scatterers $\rho(\vecx) = 0$, and the solution of \eqref{LBeq} is $f(t,\vecx,\vecy) = a(\vecx - t \vecy,\vecy)$, which is consistent with free transport. In the case of a single scatterer, classical and semiclassical scattering theory yields a linear Boltzmann equation with $\rho(\vecx)=\delta(\vecx)$ \cite{Nier95}. See also \cite{Nier96}, in particular Section 7.2 for the case when $\rho(\vecx)$ is an infinite superposition of point masses in dimension $d=1$.

The principal result of the present work establishes convergence in the Boltzmann-Grad limit for the quantum periodic setting, at least up to second order in the coupling constant. Perhaps surprisingly, and unlike the classical case \cite{Golse07}, this limit is compatible with the linear Boltzmann equation. We nevertheless conjecture that higher-order terms in the coupling constant are incompatible, and that in particular the limit process does not satisfy the linear Boltzmann equation.  A heuristic description of the full limit process will be provided elsewhere \cite{GM19}. 

A technical step in this paper is to generalise the limit theorems for multi-variable theta series, which were employed in the proof of the Berry-Tabor conjecture for the Laplacian on tori with quasi-periodic boundary conditions \cite{Marklof02,Marklof03}. Crucial ingredients in the proof of these statements are equidistribution results for homogeneous flows against unbounded test functions, which requires estimates on the escape of mass into the cusp of the relevant homogeneous space. The results in \cite{Marklof02,Marklof03} are based on Ratner's measure classification theorem and are therefore ineffective. The recent paper \cite{SP19} provides effective rate-of-convergence estimates in this context (we will not need these for the present study).

Given initial data $f_0$ in the Schwartz class $\scrS(\RR^d)$ and scaling parameter $h>0$, the quantum amplitude $f(t,\vecx)$ at time $t$ is given by the Schr\"odinger equation
\begin{equation}\label{Schrodinger}
\i \tfrac{h}{2\pi}\; \partial_t f (t,\vecx) = H_{h,\lambda} f(t,\vecx),\qquad f(0,\vecx)=f_0(\vecx),
\end{equation}
with quantum Hamiltonian
\begin{equation}
H_{h,\lambda}=H_{h,0} + \lambda \Op(V), \qquad H_{h,0}=-\frac{h^2}{8\pi^2}\; \Delta .
\end{equation}
Here $\Delta$
is the standard Laplacian in $\RR^d$, and $\Op(V)$ denotes multiplication by the $\ZZ^d$-periodic potential
\begin{equation} \label{potentialdef}
V(\vecx)=V_r(\vecx)=\sum_{\vecm\in\ZZ^d} W(r^{-1} (\vecx+\vecm)),
\end{equation}
with a fixed single-site potential $W$. We will assume from here onwards that $d \geq 2$, and that $W\in\scrS(\RR^d)$ is  real-valued. 

 We expect that our analysis can be extended to scatterer configurations where $\ZZ^d$ is replaced by an arbitrary Euclidean lattice $\scrL$ of full rank in $\RR^d$, and more generally to locally finite $\scrL$-periodic point sets. This requires, however, a substantial generalisation of the asymptotics discussed in Section 7, which are based on limit theorems for the pair correlation of general positive definite quadratic forms. The latter are currently understood, in the necessary scaling regime, only in dimension $d=2$ \cite{EskinMargulisMozes,MargulisMohammadi}.

The quantities $r,\lambda>0$ are scaling parameters which we will refer to as the scattering radius and coupling constant respectively. The operator $H_{h,\lambda}$ can be realised as the Weyl quantisation of the classical Hamiltonian $H_\lambda^\cl(\vecx,\vecy)=\frac12 \|\vecy\|^2+\lambda V(\vecx)$. The solution of \eqref{Schrodinger} can be represented as $f(t,\vecx)=U_{h,\lambda}(t)f_0(\vecx)$ with
\begin{equation}
U_{h,\lambda}(t)= \e(- H_{h,\lambda} t/h), \qquad \e(z):=\e^{2\pi\i z} .
\end{equation}
To characterise the asymptotic behaviour of the quantum dynamics, it will be convenient to use the time evolution of linear operators $A(t)$ (``quantum observables'') given by the Heisenberg evolution
\begin{equation}
A(t)=U_{h,\lambda}(t)\, A\, U_{h,\lambda}(t)^{-1} .
\end{equation}
We will use the $\L^2$ inner product 
\begin{equation}\label{innerp}
\langle a, b\rangle = \int_{\RR^d\times\RR^d} a(\vecx,\vecy)\, \overline{b(\vecx,\vecy)}\, \d \vecx \d \vecy,
\end{equation}
and the Hilbert-Schmidt inner product
\begin{equation}
\langle A, B\rangle_\HiS = \Tr A B^\dagger.
\end{equation}
As is standard in semiclassics, we will measure momentum in units of $h$, and use the rescaling $a(\vecx,\vecy) \mapsto h^{d/2} a(\vecx,h \vecy)$; the normalisation is chosen so that the $\L^2$-norm is preserved. In the classical picture of a point particle moving through an infinite field of scatterers, the {\em Boltzmann-Grad} scaling limit is one in which the radius of the scatterers is taken to zero, while space and time are simultaneously rescaled in order to ensure the mean free path length and mean free flight time remain finite. The classical mean free path length scales like $r^{1-d}$, and so we define the {\em semiclassical Boltzmann-Grad scaling} of $a\in\scrS(\RR^d\times\RR^d)$ by
\begin{equation}\label{BGscaling}
D_{r,h} a(\vecx,\vecy) =  r^{d(d-1)/2} h^{d/2} \, a( r^{d-1} \vecx, h \vecy),
\end{equation}
where again the normalisation is chosen so that $D_{r,h}$ preserves the inner product \eqref{innerp}. In order to ensure that the mean free flight time remains of constant order as $r \to 0$ we similarly rescale time by a factor of $r^{1-d}$.

We denote by $\Op(a)$ the standard Weyl quantisation of $a\in\scrS(\RR^d\times\RR^d)$:
\begin{equation}\label{def:Opa}
\Op(a) f(\vecx) = \int_{\RR^d\times\RR^d} a(\tfrac12(\vecx+\vecx'),\vecy) \, \e((\vecx-\vecx')\cdot\vecy)\, f(\vecx')\, \d \vecx' \d \vecy
\end{equation}
where $f\in\scrS(\RR^d)$. We furthermore define the corresponding scaled quantisation by $\Op_{r,h}=\Op\circ D_{r,h}$, and set $\Op_h=\Op_{1,h}$.

Throughout this paper we will consider the scaling limit where the quantum wavelength is of the same order as the scattering radius $r$, i.e. $h=h_0 r$ where $h_0$ is a fixed constant. 
By a simple scaling argument, we may assume without loss of generality that $h_0=1$.

\begin{con}\label{conj:allalpha}
There exists a family of linear operators $L(t):\L^1(\RR^d\times\RR^d)\to\L^1(\RR^d\times\RR^d)$ such that

\begin{enumerate}
\item[(i)] for all $a,b\in\scrS(\RR^d\times\RR^d)$, $A=\Op_{r,h}(a)$, $B=\Op_{r,h}(b)$, $\lambda>0$ and $t>0$,
\begin{equation}\label{eq:conj}
\lim_{h= r\to 0} \langle A(t r^{1-d}) , B \rangle_\HiS = \langle L(t) a, b\rangle ,
\end{equation}
\item[(ii)] $L(t) a(\vecx,\vecy)$ is in general {\em not} a solution of the linear Boltzmann equation.
\end{enumerate}
\end{con}

Appendix \ref{appendix0} provides an interpretation of $\langle A(t r^{1-d}) , B \rangle_\HiS$ in terms of the phase-space distribution of a solution $f(t,\vecx)$ of the Schr\"odinger equation \eqref{Schrodinger} with initial condition
\begin{equation}
f_0(\vecx) = r^{d(d-1)/2} \phi(r^{d-1} \vecx) \, \e(\vecp\cdot\vecx/h) ,
\end{equation}
for $\phi\in\scrS(\RR^d)$ and $\vecp\in\RR^d$. A schematic drawing of the initial wavepacket $f_0$ is given in Figure \ref{fig1} (shown is the positive real part of $f_0$). 

In the case of random (rather than periodic) scatterer configurations, Eng and Erd\"os \cite{Eng_Erdos} have proved convergence to a limit $L(t) a(\vecx,\vecy)$, which in fact is a solution to the linear Boltzmann equation with the standard quantum mechanical collision kernel 
\begin{align}\label{crosssection}
\colker(\vecy,\vecy') &=  8 \pi^2 \, \delta(\|\vecy\|^2-\|\vecy'\|^2) \, | T(\vecy,\vecy')|^2.
\end{align}
Here $T(\vecy,\vecy')$ is the kernel of the $T$-matrix in momentum representation.
It is related to the quantum scattering cross section by the formula (c.f.~\cite[App.~A]{Nier95})
\begin{align}
\sigma(\vecy,\vecy') = 4\pi^2 \|\vecy\|^{d-3} |T(\vecy,\vecy')|^2 .
\end{align}
%In can then be shown that for a system described by the linear Boltzmann equation whose collisions are governed by quantum mechanics, the collision kernel takes the form (c.f. \cite[Eq. A2- A3]{Nier95})
%\begin{align}
%\colker(\vecy,\vecy') &= \frac{1}{m^2} \, \sigma(\vecy,\vecy') \, \frac{1}{\|\vecy\|^{d-3}} \, \delta\left( \frac{\|\vecy\|^2}{2 m} - \frac{\|\vecy'\|^2}{2m} \right) \\
%&=  2 m \, \frac{2 \pi}{\hbar} \, |T(\vecy,\vecy')|^2 \, \frac{ 1}{(2\pi \hbar)^d}   \, \delta\left( \|\vecy\|^2 - \|\vecy'\|^2 \right).
%\end{align}
%Now we note that this is precisely the quantity given in equation \eqref{crosssection} after taking $\hbar := (2\pi)^{-1}$ and $m=1$. Furthermore, by iterating \eqref{Tmat} one obtains a formal expansion for $T$, and thus $\colker$, in terms of $\lambda$. 
%Iterating \eqref{Tmat} and then expanding $\colker(\vecy,\vecy') \sim \sum_{n=2}^\infty \colker_n(\vecy,\vecy') \, \lambda^n$ we have for the first term 
The Born approximation for the $T$-matrix yields Fermi's golden rule,
\begin{equation}
\colker_2(\vecy,\vecy')= 8 \pi^2 \,\delta(\|\vecy\|^2-\|\vecy'\|^2)  \big|\hat W\big(\vecy-\vecy'\big)\big|^2,
\end{equation}
where $\hat W$ is the Fourier transform of the single-site potential $W$. 

We will use a perturbative approach to provide evidence for Conjecture \ref{conj:allalpha}: The present paper establishes convergence up to second order in the coupling constant $\lambda$, where all terms are consistent with the linear Boltzmann equation. Based on this analysis, we develop in \cite{GM19} a heuristic model for higher order terms some of which do not match the linear Boltzmann equation; this provides support for the second assertion of Conjecture \ref{conj:allalpha}. To formulate the main theorem of the present paper, consider the formal expansion
\begin{equation}\label{LL}
L(t) \sim \sum_{n=0}^\infty L_n(t) \lambda^n,
\end{equation}
and define the linear operators $L_0$, $L_1$ and $L_2$ acting on functions in $\scrS(\RR^d\times\RR^d)$ by
\begin{equation}\label{L0}
L_0(t)a(\vecx,\vecy) = a(\vecx-t\vecy,\vecy), \qquad L_1(t)a(\vecx,\vecy) = 0,
\end{equation}
\begin{align}\label{L2}
L_2(t) a(\vecx,\vecy) 
= \int_0^t  \int_{\RR^{d}}  \colker_2(\vecy,\vecy') [ a( \vecx- s \vecy - (t-s)  \vecy', \vecy') -  a(\vecx-t\vecy, \vecy)] \d \vecy' \d s.
\end{align}
Relations \eqref{LL}--\eqref{L2} are consistent with $L(t)$ generating solutions of the linear Boltzmann equation with $\rho(\vecx)=1$.

Our main result is as follows.

\begin{thm}\label{thm:allalpha}
Let $t>0$ and $a,b\in\scrS(\RR^d\times\RR^d)$, $A=\Op_{r,h}(a)$, $B=\Op_{r,h}(b)$. Then there exist linear operators $A_0^{(r)}(t)$, $A_1^{(r)}(t)$, $A_2^{(r)}(t)$, such that for $h=r\in(0,1]$,
\begin{equation}\label{eq:118}
\langle A(t r^{1-d}) , B \rangle_\HiS = \sum_{n=0}^2 \langle A_n^{(r)}(t r^{1-d}) , B \rangle_\HiS  \,\lambda^n + \sum_{n=3}^6 O(r^{-nd/2} \lambda^n )
\end{equation}
and 
\begin{equation}\label{eq:118b}
\lim_{h= r\to 0} \langle A_n^{(r)}(t r^{1-d}) , B \rangle_\HiS = \langle L_n(t) a, b \rangle   \qquad (n=0,1,2).
\end{equation}
\end{thm}

The notation $f(x)=O(g(x))$ means ``there is a positive constant $C$ such that $|f(x)| \leq C \, |g(x)|$.''  We will also use $f(x) \ll g(x)$ synonymously, and subscript $O_\epsilon$ or $\ll_\epsilon$ to highlight the dependence of the implied constant $C=C_\epsilon$ on a parameter $\epsilon$.

The key point here is to view the first sum on the right hand side of \eqref{eq:118} as the first three terms of a formal power series expansion in $\lambda$, which according to \eqref{eq:118b} each converge to the corresponding terms of the conjectured limit \eqref{LL}. The second sum in \eqref{eq:118} provides a error estimate that allows an interpretation beyond a formal power series, but this is only of secondary interest.

We will actually prove a stronger result than Theorem \ref{thm:allalpha}. For a given quasi-momentum $\vecalf\in[0,1)^d$, consider the Bloch functions $\varphi_\vecm^\vecalf(\vecx)=\e((\vecm+\vecalf)\cdot\vecx)$, $\vecm\in\ZZ^d$,
and define the projection $\Pi_\vecalf$ acting on $f\in\scrS(\RR^d)$ by
\begin{equation}\label{Pi-def}
\Pi_\vecalf f(\vecx) = \sum_{\vecm\in\ZZ^d} \langle f, \varphi_\vecm^\vecalf\rangle \;  \varphi_\vecm^\vecalf(\vecx) ,
\end{equation}
with inner product
\begin{equation}
\langle f, g\rangle = \int_{\RR^d} f(\vecx)\, \overline{g(\vecx)}\, \d \vecx .
\end{equation}
Note that, by Poisson summation,
\begin{equation} \label{eq:Frep}
\Pi_\vecalf f(\vecx) = \sum_{\vecm\in\ZZ^d} \e(-\vecm\cdot\vecalf) f(\vecx+\vecm),
\end{equation}
and hence that by integrating over $\vecalf \in [0,1)^d$ one regains $f(\vecx)$. We will refer to $\Pi_\vecalf$ as a {\em Bloch projection} and $\vecalf$ as a {\em Bloch vector} or {\em quasi-momentum}. Instead of \eqref{eq:118} we consider now
\begin{align} \label{blochprojnorm}
\langle \Pi_\vecalf A(t r^{1-d}),B \rangle_{\HiS}.
\end{align}
As we will see, the behaviour of \eqref{blochprojnorm} in the limit $ h= r\to 0$ depends on the number theoretic properties of $\vecalf$. We  call a vector $\vecalf=(\alpha_1,\ldots,\alpha_d)\in\RR^d$ 
{\em Diophantine of type $\kappa$}, if there exists a 
constant $C>0$ such that 
\begin{equation}
\max_j \bigg|\alpha_j - \frac{m_j}{q}\bigg| > \frac{C}{q^\kappa}
\end{equation}
for all $m_1,\ldots,m_d,q\in\ZZ$, $q>0$. The smallest possible value
for $\kappa$ is $\kappa=1+\frac1d$. In this case $\vecalf$
is called {\em badly approximable}.

\begin{thm}\label{thm:dioalpha}
Suppose $\vecalf$ is Diophantine of type $\kappa < (d-1)/(d-2)$ and the components of $(1,\trans\vecalf)$ are linearly independent over $\QQ$. Let $t>0$ and $a,b\in\scrS(\RR^d\times\RR^d)$, $A=\Op_{r,h}(a)$, $B=\Op_{r,h}(b)$.  Then there exist linear operators $A_0^{(r,\vecalf)}(t)$, $A_1^{(r,\vecalf)}(t)$, $A_2^{(r,\vecalf)}(t)$, such that for $h=r\in(0,1]$,
\begin{equation}\label{eq:dioalpha001}
\langle \Pi_\vecalf A(t r^{1-d}) , B \rangle_\HiS = \sum_{n=0}^2 \langle A_n^{(r,\vecalf)}(t r^{1-d}) , B \rangle_\HiS  \,\lambda^n + \sum_{n=3}^6 O(r^{-nd/2} \lambda^n ) 
\end{equation}
and
\begin{equation} \label{thmconv}
\lim_{h=r\to 0} \langle A_n^{(r,\vecalf)}(t r^{1-d}) , B \rangle_\HiS = \langle L_n(t) a, b \rangle   \qquad (n=0,1,2) .
\end{equation}
\end{thm}

Since the set of Diophantine $\vecalf\in[0,1)^d$ has full Lebesgue measure, Theorem \ref{thm:allalpha} may be viewed as an averaged (and thus weaker) version of Theorem \ref{thm:dioalpha}. The convergence in \eqref{thmconv} is however highly non-uniform in $\vecalf$, and the derivation of Theorem \ref{thm:allalpha} from Theorem \ref{thm:dioalpha} requires non-trivial dominated convergence estimates.

In his PhD thesis \cite{Griffin-thesis}, the first author established a version of Theorem \ref{thm:dioalpha} for the small-scatter problem on the torus $\TT^d=\RR^d/\ZZ^d$ with quasi-periodic boundary conditions $f(\vecx+\vecm)=e(\vecm\cdot\vecalf) f(\vecx)$ ($\vecm\in\ZZ^d$), for observables that do not depend on position $\vecx$. This in particular complements results in \cite{Castella_WC,Castella_Plagne_LD} where $\vecalf=\vecnull$, and furthermore provides a discussion of the expansion terms leading to a failure of the linear Boltzmann equation. The key observation in \cite{Castella_WC,Castella_Plagne_LD} is that due to the large mean degeneracy of the spectrum of the Laplacian on the torus $\TT^d$, the semiclassical Boltzmann-Grad limit diverges; a different normalisation then yields a non-universal limit, which in particular is not consistent with the linear Boltzmann equation. It is interesting to note that adding a suitably chosen damping term allows one to recover the linear Boltzmann equation even in this singular case \cite{Castella_LD,Castella_LD2}. The small-scatterer problem in rectangular domains (Sinai billiards) has also been investigated in the context of quantum chaos; here the smooth potential is replaced by a disc with Dirichlet boundary conditions \cite{Berry81,Dahlqvist98}. 

This paper is organised as follows.
Sections \ref{sec:momentum} and \ref{sec:bloch} provide basic background and notation on Weyl calculus in momentum representation and Floquet-Bloch theory. Section \ref{sec:duhamel} uses the Duhamel principle to obtain a perturbation series in $\lambda$. We then apply the Boltzmann-Grad scaling in Section \ref{sec:boltzmann-grad}. The zeroth and first order terms are elementary, and are calculated in Section \ref{sec:0and1}. Terms of second order require equidistribution results for horocycles (Section \ref{sec:equidistribution}) and mean value theorems for theta functions (Section \ref{sec:theta}), which build on the papers \cite{Marklof02,Marklof03}. The second order terms are computed in Section \ref{sec:two}. The estimates of the error term in Theorem \ref{thm:dioalpha} require analogous results for higher-dimensional theta functions (Section \ref{sec:higher-order}), and are presented in Sections \ref{sec:error}. The proof of Theorem \ref{thm:dioalpha} is given at the end of Section \ref{sec:error}. Section \ref{sec:averages} concludes with the proof of Theorem \ref{thm:allalpha}.

\subsection*{Acknowledgements}
We thank Laszlo Erd\"os and Leonid Parnovski for helpful discussions,  and the anonymous referee for many valuable comments. We are grateful to the Isaac Newton Institute, Cambridge, for its support and hospitality during the programme ``Periodic and Ergodic Spectral Problems.'' The research leading to these results has received funding from the European Research Council under the European Union's Seventh Framework Programme (FP/2007-2013) / ERC Grant Agreement n. 291147.

\section{Momentum representation}\label{sec:momentum}

We have so far represented quantum wave amplitudes $f$ in the position representation. It will in fact be more convenient to work with its Fourier transform $\hat f$, which represents the wave amplitude as a function of the quantum particle's momentum. Set $\e(x)=\exp(2\pi\i x)$, and define the Fourier transform $\hat f= \scrF f$ of $f$ by
\begin{equation}\label{FT}
\hat f(\vecy)= \scrF f(\vecy) =  \int_{\RR^d} \e(-\vecy\cdot\vecx)\, f(\vecx)\,\d \vecx .
\end{equation}
The Fourier transform of a linear operator $A$ on $\L^2(\RR^d)$ is then naturally defined by
\begin{equation}
\hat A = \scrF A \scrF^{-1} .
\end{equation}
Explicitly, the corresponding Schwartz kernel satisfies
\begin{equation} \label{Skern}
\hat A(\vecy,\vecy') = \int_{\RR^{2d}} A(\vecx,\vecx') \e(-\vecx\cdot\vecy + \vecx'\cdot\vecy') \, \d \vecx\,\d \vecx' .
\end{equation}
The Schwartz kernel of the Fourier transform of $\Op(a)$ reads
\begin{equation} \begin{split}\label{hatOpa}
\hOp(a) (\vecy,\vecy') & = \int_{\RR^d} a(\vecx,\tfrac12(\vecy+\vecy')) \, \e(-\vecx\cdot(\vecy-\vecy'))\, \d \vecx \\
& = \tilde a(\vecy-\vecy',\tfrac12(\vecy+\vecy')) ,
\end{split}\end{equation}
where $\tilde a$ denotes the Fourier transform of $a$ in the first variable only, i.e. 
\begin{equation}
\tilde a(\veceta,\vecy) = \int_{\RR^d} a(\vecx,\vecy) \e(-\vecx \cdot \veceta) \, \d \vecx.
\end{equation}
The above definition extends to larger function spaces by standard arguments \cite{Folland}. Two notable special cases occur when $a$ is a function exclusively of either $\vecx$ or $\vecy$. In the first case when $a = a(\vecx)$ we have 
$\hOp(a) (\vecy,\vecy') 
%= \int_{\RR^d\times\RR^d} a(\vecx) \, \e(\vecx\cdot(\vecy-\vecy'))\, f(\vecy')\, \d \vecx \d \vecy' 
= \hat a(\vecy-\vecy')$,
and in the second case when $a = a(\vecy)$ we obtain
$\hOp(a) (\vecy,\vecy') = a(\vecy)\, \delta_\vecnull(\vecy-\vecy')$.
The choice $a=L_0(t) V$ in \eqref{hatOpa} yields for instance
\begin{equation}
\hOp(L_0(t) V)(\vecy,\vecy')=r^d \sum_{\vecm\in\ZZ^d} \hat W(r\vecm)
\e(- \tfrac12 t \vecm\cdot (\vecy +\vecy')) \, \delta_\vecm(\vecy-\vecy') ,
\end{equation}
where $\delta_\vecm$ denotes the Dirac delta mass at the point $\vecm$.

The quantizations of the Hamilton functions $H_0^\cl$ and $H_\lambda^\cl$ are denoted by $H_0=\Op H_0^\cl=-\frac{1}{8\pi^2}\Delta$ and $H_\lambda=\Op H_\lambda^\cl= H_0+\lambda \Op V$ respectively. The Schr\"odinger equation for the time evolution of the the wave amplitude $f(t,\vecx)$ can then be written (in units where Planck's constant is $1$)
\begin{equation}
\tfrac{\i}{2\pi}\partial_t f(t,\vecx)= H_\lambda f(t,\vecx),\qquad f(0,\vecx)=f_0(\vecx),
\end{equation}
which has the solution 
\begin{equation}
f(t,\vecx) = U_\lambda(t) f_0(\vecx), \qquad U_\lambda(t):= \e(- H_\lambda t) .
\end{equation}
The relation to the corresponding operators in the introduction is
\begin{equation} \label{lambdaisrescaled}
H_{h,\lambda}= h^2 H_{\lambda/h^2}, \qquad
U_{h,\lambda}(t)= U_{\lambda/h^2}(h t) .
\end{equation}
It will be more convenient to work with $U_\lambda(t)$ in what follows, and then later appeal to \eqref{lambdaisrescaled}.

Since $H_0^\cl$ is a quadratic polynomial,  we have the exact Egorov property,
\begin{equation}\label{egorov}
U_0(t) \Op(a) U_0(-t) = \Op (L_0(t) a) .
\end{equation}
In momentum representation the kernel of the operator $\hat H_0$ takes the form
\begin{equation}
\hat H_0(\vecy,\vecy')=\tfrac12 \|\vecy\|^2 \delta_\vecnull(\vecy-\vecy')
\end{equation}
and thus also
\begin{equation}
\hat U_0 (t)(\vecy,\vecy')=\e(-\tfrac12 \, t\|\vecy\|^2) \delta_\vecnull(\vecy-\vecy') .
\end{equation}

\section{Bloch projections}\label{sec:bloch}

As is standard in the study of periodic potentials, we use the fact that any solution to our Schr\"odinger equation can be decomposed into quasiperiodic functions parametrised by quasimomentum $\vecalf \in \TT^d=\RR^d/\ZZ^d$ (Floquet-Bloch decomposition). For $f\in\scrS(\RR^d)$ the function $\psi(\vecx)=\Pi_\vecalf f(\vecx)$ satisfies, for every $\veck\in\ZZ^d$, 
\begin{equation}\label{quasip}
\psi(\vecx+\veck) = \e(\veck\cdot\vecalf) \psi(\vecx) .
\end{equation}
We denote by $\scrH_\vecalf$ the Hilbert space of functions that satisfy the quasiperiodicity condition \eqref{quasip} and have finite $\L^2$-norm with respect to the inner product
\begin{equation}
\langle \psi, \varphi  \rangle_\vecalf = \int_{\TT^d} \psi(\vecx)\, \overline{\varphi(\vecx)}\, \d \vecx .
\end{equation}
We define the corresponding Hilbert-Schmidt product for linear operators from $\L^2(\RR^d)$ to $\scrH_\vecalf$ by

\begin{equation}
\langle  A,  B  \rangle_{\HiS,\vecalf} = \Tr  A B^\dagger = \int_{\TT^d} \bigg(\int_{\RR^d} A(\vecx,\vecx') \overline{B(\vecx,\vecx')} \, \d \vecx' \bigg)\, \d \vecx.
\end{equation}

\begin{lem}\label{three-one}
If $f,g\in\scrS(\RR^d)$, then $\Pi_\vecalf f,\Pi_\vecalf g\in\scrH_\vecalf\cap\C^\infty(\RR^d)$ and
\begin{equation}
\langle \Pi_\vecalf f, g \rangle=\langle f,  \Pi_\vecalf g \rangle=\langle \Pi_\vecalf f,  \Pi_\vecalf g \rangle_\vecalf 
= \sum_{\vecm\in\ZZ^d} \hat f(\vecm+\vecalf) \overline{\hat g(\vecm+\vecalf)}.
\end{equation}
\end{lem}

\begin{proof}
We have by \eqref{eq:Frep}
\begin{equation}\label{zerhs}
\langle \Pi_\vecalf f,  \Pi_\vecalf g \rangle_\vecalf
= \sum_{\vecm\in\ZZ^d} \e(\vecm\cdot\vecalf) \int_{\TT^d} (\Pi_\vecalf f)(\vecx)\, \overline{g(\vecx+\vecm)}\, \d \vecx .
\end{equation}
Using the invariance \eqref{quasip} of $\Pi_\vecalf f$, we see that the summation and integration can be combined to an integral over $\RR^d$ which equals $\langle \Pi_\vecalf f, g \rangle$. The final identity follows directly from the definition \eqref{Pi-def}, which yields
\begin{equation}
\langle \Pi_\vecalf f, g \rangle = \sum_{\vecm\in\ZZ^d} \langle f, \varphi_\vecm^\vecalf\rangle \;  \langle \varphi_\vecm^\vecalf, g \rangle = \sum_{\vecm\in\ZZ^d} \hat f(\vecm+\vecalf) \overline{\hat g(\vecm+\vecalf)}.
\end{equation}
\end{proof}

Note that for the Fourier transform,
\begin{equation}\label{FTse}
\hat\Pi_\vecalf f(\vecy) = \sum_{\vecm\in\ZZ^d} f(\vecm+\vecalf)\,\delta_{\vecm+\vecalf}(\vecy).
\end{equation}

\begin{lem} \label{lem:projectionHSnorm}
If $A,B$ have Schwartz kernel in $\scrS(\RR^d\times\RR^d)$, then $\Pi_\vecalf A$, $\Pi_\vecalf B$ are linear operators $\L^2(\RR^d)\to\scrH_\vecalf$, and
\begin{equation}
\begin{split}
\langle \Pi_\vecalf A, B \rangle_\HiS & =\langle A,  \Pi_\vecalf B \rangle_\HiS =\langle \Pi_\vecalf A,  \Pi_\vecalf B \rangle_{\HiS,\vecalf} \\
& = \sum_{\vecm\in\ZZ^d} \int_{\RR^d} \hat A(\vecm+\vecalf,\vecy) \overline{\hat B(\vecm+\vecalf,\vecy)} \, \d \vecy.
\end{split}
\end{equation}
\end{lem}

\begin{proof}
This is analogous to the proof of Lemma \ref{three-one}.
By \eqref{Pi-def}, we have
\begin{equation} 
[\Pi_\vecalf B] (\vecx,\vecx') = \sum_{\vecm\in\ZZ^d} \e(-\vecm\cdot\vecalf) B(\vecx+\vecm, \vecx'),
\end{equation}
and so 
\begin{equation} 
\begin{split}
\langle \Pi_\vecalf A , \Pi_\vecalf B \rangle_{\HiS,\vecalf} 
& = \sum_{\vecm\in\ZZ^d} \e(\vecm\cdot\vecalf) \int_{\TT^d} \bigg( \int_{\RR^d} [\Pi_\vecalf A](\vecx, \vecx') \,  \overline{B(\vecx+\vecm, \vecx')} \d\vecx' \bigg) \d\vecx \\
& = \int_{\RR^d} \bigg( \int_{\RR^d} [\Pi_\vecalf A](\vecx, \vecx') \,  \overline{B(\vecx, \vecx')} \d\vecx' \bigg) \d\vecx 
= \langle \Pi_\vecalf A, B \rangle_\HiS,
\end{split}
\end{equation}
where we have used the identity $[\Pi_\vecalf A](\vecx+\vecm, \vecx')=\e(\vecm\cdot\vecalf) [\Pi_\vecalf A](\vecx, \vecx')$, cf.~\eqref{quasip}.
The proof of $\langle A,  \Pi_\vecalf B \rangle_\HiS =\langle \Pi_\vecalf A,  \Pi_\vecalf B \rangle_{\HiS,\vecalf}$ is analogous. Finally, in view of \eqref{Skern} and \eqref{FTse} we have that
\begin{align}
[\widehat{\Pi_\vecalf A }](\vecy,\vecy') &=  \sum_{\vecm \in \ZZ^d} \delta_{\vecm + \vecalf}(\vecy)\hat A(\vecm + \vecalf, \vecy'),
\end{align}
which yields
\begin{equation} \begin{split}
\langle \Pi_\vecalf A, B \rangle_{\HiS} &= \int_{\RR^{2d}} \sum_{\vecm \in\ZZ^d} \delta_{\vecm+\vecalf} (\vecy) \, \hat A(\vecm+\vecalf, \vecy') \overline{\hat B(\vecy, \vecy')}  \, \d \vecy  \, \d \vecy' \\
&= \, \sum_{\vecm \in\ZZ^d} \int_{\RR^{d}}  \hat A(\vecm+\vecalf, \vecy) \overline{\hat B(\vecm+\vecalf, \vecy)} \, \d \vecy .
\end{split} \end{equation}

\end{proof}

We denote by $\Delta^\vecalf$ the standard Laplacian acting on $\scrH_\vecalf$, and set 
\begin{equation}
H^\vecalf_\lambda=H^\vecalf_0 + \lambda \Op(V),  \qquad U^\vecalf_\lambda(t)= \e(- H^\vecalf_\lambda t)  .
\end{equation}

\begin{lem}\label{lem:comm}
For $f\in\scrS(\RR^d)$,
\begin{equation} \label{lem:comm:statement}
\Pi_\vecalf U_\lambda(t) f = U^\vecalf_\lambda(t) \Pi_\vecalf f .
\end{equation}
\end{lem}

\begin{proof}
 We have the commutation relations 
 \begin{equation}
\Pi_\vecalf H_0 = H_0^\vecalf \Pi_\vecalf, \qquad \Pi_\vecalf \Op(V) = \Op(V) \Pi_\vecalf .
\end{equation}
Consider the time derivative of the left hand side of \eqref{lem:comm:statement},
\begin{equation}
\begin{split}
\partial_t \Pi_\vecalf U_\lambda(t) f  &= -2 \pi \i \Pi_\vecalf (H_0 + \lambda \Op(V)) U_\lambda(t) f \\
& = -2 \pi \i (H_0^\vecalf  + \lambda \Op(V)) \Pi_\vecalf  U_\lambda(t) f.
\end{split}
\end{equation}
Thus the left hand side of  \eqref{lem:comm:statement} is the unique solution to
\begin{align}
\partial_t g(t,\vecy) =  -2 \pi \i \, H_\lambda^\vecalf \, g(t,\vecy)
\end{align}
with initial condition $g(0,\vecy) := \Pi_\vecalf f(\vecy)$. The right hand side of \eqref{lem:comm:statement} solves the same PDE, and the proof is complete.
\end{proof}

\section{Duhamel's principle}\label{sec:duhamel}

Duhamel's principle provides an explicit expansion of the solution in terms of the coupling constant $\lambda$. By truncating the expansion at order $2$, we will be left with theta functions that, in a certain scaling limit, can be treated with the tools of homogeneous dynamics. The explicit error terms can be handled separately.
Our first aim is to work out the time evolution of un-scaled observables, \begin{equation}U_\lambda(t) \Op(a) U_\lambda(-t),\end{equation}
perturbatively in $\lambda$. We first study the problem in the interaction picture, i.e., consider 
\begin{equation}
U_\lambda(t) U_0(-t) \Op(a) U_0(t) U_\lambda(-t).
\end{equation}
Note that in view of the Egorov property \eqref{egorov} this is equivalent to the original problem upon replacing $a$ by $L_0(t) a$. We define the operators $K(t)$ and $R(t)$ for $t \in \RR$ by
\begin{equation}\label{def:KR}
K(t)=U_0(t) \Op(V) U_0(-t) \quad \text{ and } \quad R(t)=U_\lambda(t) U_0(-t)  .
\end{equation}
Furthermore, for $\vecs = (s_1,\dots,s_n)$ and $\ell\leq n$ we denote by $K_{\ell,n}(\vecs)$ the product 
\begin{equation}
K_{\ell,n}(\vecs)=
K(s_\ell) \cdots K(s_n)  .
\end{equation}
Then
\begin{multline}\label{propeq}
\langle \Pi_\vecalf U_\lambda(t) U_0(-t) \Op(a) U_0(t) U_\lambda(-t),\Op(b)\rangle_\HiS \\
=
\langle \Pi_\vecalf R(t) \Op(a) R(t)^{-1},\Op(b)\rangle_\HiS.
\end{multline}
Duhamel's principle asserts that
\begin{equation} \label{Duhamel}
R(t)= I - 2\pi\i \lambda \int_0^t R(s) K(s) \, \d s,
\end{equation}
and iterating this expression $N$ times yields 
\begin{equation} \label{duhamelexpansion}
R(t) = \sum_{n=0}^N  \lambda^n R_n(t) 
+  \lambda^{N+1} R_{N+1,\error}(t) ,
\end{equation}
where $R_0(t)=I$ and
\begin{equation}
R_n(t)= (- 2\pi\i)^n \int_{0<s_1<\ldots<s_n<t}  K_{1,n}(\vecs)  \, \d s  \qquad (n\geq 1).
\end{equation}
The error term is similarly given by
\begin{equation}
R_{N+1,\error}(t) = (-2\pi\i)^{N+1} \, \int_{0<s_1<\ldots<s_{N+1}<t} R(s_1) K_{1,N+1}(\vecs) \, \d \vecs.
\end{equation}

The inverse of $R(t)$ can be calculated by taking Hermitian conjugate. It is given by
\begin{equation} \label{duhamelexpansionpast}
R(t)^{-1} = \sum_{n=0}^N  \lambda^n R^-_n(t) 
+  \lambda^{N+1} R^-_{N+1,\error}(t) ,
\end{equation}
where $R^-_0(t)=I$,
\begin{equation}
R^-_n(t)= (2\pi\i)^n\, \int_{0<s_n<\ldots<s_1<t}  K_{1,n}(\vecs)  \, \d \vecs  \qquad (n\geq 1),
\end{equation}
and the error term is
\begin{equation}
R^-_{N+1,\error}(t) = (2\pi\i)^{N+1} \, \int_{0<s_{N+1}<\ldots<s_1<t}  K_{1,N+1}(\vecs) R(s_{N+1})^{-1} \, \d \vecs .
\end{equation}
We have also used the fact that $\Op(V)=\Op(V)^\dagger$ (since $V$ is real-valued) and thus $K(t)=K(t)^\dagger$. Our methods will permit explicit calculation of the terms in this expansion up to order $2$, and so specializing to the case $N=2$ the expansion takes the following form
\begin{equation}\label{propeq2}
\langle \Pi_\vecalf U_\lambda(t) U_0(-t) \Op(a) U_0(t) U_\lambda(-t),\Op(b)\rangle_\HiS
= \sum_{n=0}^6 \lambda^n Q_n(t,a,b) 
\end{equation}
with the main terms $Q_0$ to $Q_2$ given by
\begin{equation} \label{propeq2terms}
\begin{split}
Q_0(t,a,b) &= \langle \Pi_\vecalf  \Op(a) ,\Op(b)\rangle_\HiS \\ 
Q_1(t,a,b) &= \langle \Pi_\vecalf R_1(t) \Op(a) ,\Op(b)\rangle_\HiS \\
& \qquad + \langle \Pi_\vecalf \Op(a) R^-_1(t)  ,\Op(b)\rangle_\HiS \\ 
Q_2(t,a,b) &= \langle \Pi_\vecalf R_2(t) \Op(a) ,\Op(b)\rangle_\HiS \\
& \qquad +  \langle \Pi_\vecalf R_1(t) \Op(a) R^-_1(t)  ,\Op(b)\rangle_\HiS \\ 
& \qquad + \langle \Pi_\vecalf \Op(a) R^-_2(t)  ,\Op(b)\rangle_\HiS .
\end{split}
\end{equation}
The error terms $Q_3$ through $Q_6$ read
\begin{equation} \label{propeq2errorterms}
\begin{split}
Q_3(t,a,b) &=  \langle \Pi_\vecalf R_{3,\error}(t) \Op(a) ,\Op(b)\rangle_\HiS  \\
& \qquad +  \langle \Pi_\vecalf R_2(t) \Op(a) R^-_1(t)  ,\Op(b)\rangle_\HiS \\ 
&\qquad +  \langle \Pi_\vecalf R_1(t) \Op(a) R^-_2(t)  ,\Op(b)\rangle_\HiS  \\
& \qquad + \langle \Pi_\vecalf \Op(a) R^-_{3,\error}(t)  ,\Op(b)\rangle_\HiS \\ 
Q_4(t,a,b) &= \langle \Pi_\vecalf R_{3,\error}(t) \Op(a) R^-_1(t)  ,\Op(b)\rangle_\HiS  \\
& \qquad +  \langle \Pi_\vecalf R_2(t) \Op(a) R^-_2(t)  ,\Op(b)\rangle_\HiS \\ 
& \qquad + \langle \Pi_\vecalf  R_1(t) \Op(a)  R^-_{3,\error}(t)  ,\Op(b)\rangle_\HiS \\ 
Q_5(t,a,b) &= \langle \Pi_\vecalf R_{3,\error}(t) \Op(a) R^-_2(t)  ,\Op(b)\rangle_\HiS  \\ 
& \qquad +  \langle \Pi_\vecalf R_2(t) \Op(a)  R^-_{3,\error}(t)  ,\Op(b)\rangle_\HiS \\ 
Q_6(t,a,b) &= \langle \Pi_\vecalf R_{3,\error}(t) \Op(a)R^-_{3,\error}(t) ,\Op(b)\rangle_\HiS. 
\end{split}
\end{equation}
We will treat these error terms in the following way. First of all, Lemma \ref{lem:errortermsCS} shows that all of the $Q_j$ can be bounded above by quantities which are independent of $U_\lambda(t)$, and depend only on the free evolution $U_0(t)$. Then after rescaling, the resulting quantities, which we  denote $\scrJ_{\ell,n}$, can be treated with similar techniques to those used in the computation of the limit of the second order terms. 

Define
\begin{equation}\label{errortermthetafunc0}
\scrJ_{\ell,n}(t,a)= (2\pi)^{n} \int_{\substack{0<s_1<\ldots<s_{\ell}<t\\ 0<s_{n}<\ldots<s_{\ell+1}<t}} \|  \Pi_\vecalf K_{1,\ell}(\vecs) \Op(a)  K_{\ell+1,n}(\vecs)  \|_{\HiS,\vecalf} \, \d \vecs .
\end{equation}

\begin{lem} \label{lem:errortermsCS}
For $a,b\in\scrS(\RR^d)$, 
\begin{equation}
\begin{split}
\big| \langle \Pi_\vecalf  R_\ell(t) \Op(a) R^-_{n-\ell}(t), \Op(b) \rangle_\HiS \big| &\leq   \scrJ_{\ell,n}(t,a) \; \|  \Pi_\vecalf \Op(b) \|_{\HiS,\vecalf}, \\
\big| \langle \Pi_\vecalf  R_\ell(t) \Op(a) R^-_{n-\ell,\error}(t), \Op(b) \rangle_\HiS \big| &\leq   \scrJ_{\ell,n}(t,a) \; \|  \Pi_\vecalf \Op(b) \|_{\HiS,\vecalf},\\
\big| \langle \Pi_\vecalf R_{\ell,\error}(t) \Op(a) R^-_{n-\ell} (t) , \Op(b) \rangle_\HiS \big| &\leq   \scrJ_{\ell,n}(t,a) \;  \|  \Pi_\vecalf \Op(b) \|_{\HiS,\vecalf},\\
\big| \langle \Pi_\vecalf R_{\ell,\error}(t) \Op(a) R^-_{n-\ell,\error}(t) , \Op(b) \rangle_\HiS \big| &\leq   \scrJ_{\ell,n}(t,a) \; \|  \Pi_\vecalf \Op(b) \|_{\HiS,\vecalf}.\\
\end{split}
\end{equation}
\end{lem}
\begin{proof}
For the first bound, note that by Lemma \ref{lem:projectionHSnorm} and direct computation we have
\begin{equation}
\begin{split}
&| \langle \Pi_\vecalf  R_\ell(t) \Op(a) R^-_{n-\ell}(t), \Op(b) \rangle_\HiS| \\
&= | \langle \Pi_\vecalf  R_\ell(t) \Op(a) R^-_{n-\ell}(t), \Pi_\vecalf \Op(b) \rangle_{\HiS,\vecalf} | \\ 
& = (2\pi)^n  \left|\int_{\substack{0<s_1<\cdots<s_\ell<t \\ 0<s_n<\cdots<s_{\ell+1}<t}}  \langle \Pi_\vecalf  K_{1,\ell}(\vecs) \Op(a) K_{\ell+1,n}(\vecs), \Pi_\vecalf \Op(b) \rangle_{\HiS,\vecalf} \d\vecs \right| \\
& \leq (2\pi)^n \int_{\substack{0<s_1<\cdots<s_\ell<t \\ 0<s_n<\cdots<s_{\ell+1}<t}}   \left| \langle \Pi_\vecalf  K_{1,\ell}(\vecs) \Op(a) K_{\ell+1,n}(\vecs), \Pi_\vecalf \Op(b) \rangle_{\HiS,\vecalf}  \right| \d\vecs.
\end{split}
\end{equation}
The bound then follows by an application of the Cauchy-Schwarz inequality. For the second bound we similarly have that
\begin{equation}
\begin{split}
& \langle \Pi_\vecalf  R_\ell(t) \Op(a) R^-_{n-\ell,\error}(t), \Op(b) \rangle_\HiS \\
& \leq (2\pi)^n  \int_{\substack{0<s_1<\cdots<s_\ell<t \\ 0<s_n<\cdots<s_{\ell+1}<t}} \left| \langle \Pi_\vecalf  K_{1,\ell}(\vecs) \Op(a) K_{\ell+1,n}(\vecs) R(s_{n})^{-1}, \Pi_\vecalf \Op(b) \rangle_{\HiS,\vecalf} \right| \d\vecs.
\end{split}
\end{equation}
The result then follows by applying Cauchy-Schwarz and using that $R(s_n)$ is unitary. For the third bound we have 
\begin{equation}
\begin{split}
& \langle \Pi_\vecalf  R_{\ell,\error}(t) \Op(a) R^-_{n-\ell}(t), \Op(b) \rangle_\HiS \\
& \leq (2\pi)^n  \int_{\substack{0<s_1<\cdots<s_\ell<t \\ 0<s_n<\cdots<s_{\ell+1}<t}} \left| \langle \Pi_\vecalf R(s_1) K_{1,\ell}(\vecs) \Op(a) K_{\ell+1,n}(\vecs) , \Pi_\vecalf \Op(b) \rangle_{\HiS,\vecalf} \right| \d\vecs.
\end{split}
\end{equation}
This time the bound follows by first applying Lemma \ref{lem:comm}, then the Cauchy-Schwarz inequality and finally using the unitarity of $R(s)$. The last bound follows by combining the arguments for bounds two and three.
\end{proof}

Let us introduce the shorthand
\begin{equation}
\scrT_{\ell,n}(\vecy) 
=\begin{cases} \prod_{j=\ell}^{n} \e(-\tfrac12 \, (s_{j+1}-s_j) \|\vecy-\vecm_j \|^2) \hat W(r(\vecm_{j+1}-\vecm_j)) & (l \leq n) \\
1 & (l> n). 
\end{cases}
\end{equation}

\begin{lem}
The kernel of $\hat K_{\ell,n}(\vecs)=\scrF K_{\ell,n}(\vecs) \scrF^{-1}$ is explicitly given by
\begin{multline} \label{Kker}
[\hat K_{\ell,n}(\vecs) ] (\vecy,\vecy') \\
= r^{(n-\ell+1)d} \hspace{-0.2cm} \sum_{\vecm_\ell,\ldots,\vecm_n\in\ZZ^d} \e(-\tfrac12 \, s_\ell\|\vecy\|^2) \hat W(r\vecm_\ell) 
\scrT_{\ell,n-1}(\vecy) \e(\tfrac12 \, s_n \|\vecy-\vecm_n \|^2) \delta_{\vecm_n}(\vecy-\vecy') . 
\end{multline}
\end{lem}
\begin{proof}
We have that
\begin{multline}
\hat K_{\ell,n}(\vecs) \, f (\vecy) = \hat K(s_\ell) \cdots \hat K(s_n) f (\vecy)\\
= \scrF U_0(s_\ell) \Op(V) U_0(s_{\ell+1} - s_\ell) \cdots U_0(s_n-s_{n-1}) \Op(V) U_0(-s_n) \scrF^{-1} f(\vecy),
\end{multline}
and
\begin{align}
\scrF U_0(s) \Op(V) \scrF^{-1}f (\vecy) &=r^d \e(-\tfrac12 s \|\vecy\|^2)\sum_{\vecm \in \ZZ^d} \hat W(r \vecm) f(\vecy-\vecm).
\end{align}
By iterating we thus see
\begin{equation} \begin{split}
\hat K_{\ell,n}(\vecs) \, f (\vecy) &= \hat K(s_\ell) \cdots \hat K(s_n) f (\vecy)\\
&=  r^{(n-\ell+1) d} \e(-\tfrac12 s_\ell \|\vecy\|^2)\sum_{\vecm_\ell, \dots, \vecm_n \in \ZZ^d} \hat W(r \vecm_\ell) \\
& \times  \e(-\tfrac12 (s_{\ell+1} - s_\ell) \|\vecy-\vecm_\ell\|^2)  \hat W(r \vecm_{\ell+1}) \\
& \times  \e(-\tfrac12 (s_{\ell+2} - s_{\ell+1}) \|\vecy-\vecm_\ell-\vecm_{\ell+1}\|^2)  \hat W(r \vecm_{\ell+2}) \\
& \cdots \times  \e(-\tfrac12 (s_{n} - s_{n-1}) \|\vecy - \vecm_\ell - \cdots - \vecm_{n-1}\|^2)\hat W(r \vecm_n) \\
& \times \e(\tfrac12 s_n \|\vecy - \vecm_\ell - \cdots - \vecm_{n}\|^2) \,  f(\vecy - \vecm_\ell - \cdots - \vecm_{n}).
\end{split} \end{equation}
We then make the variable substitutions $\vecm_j = \tilde \vecm_j - \sum_{i= \ell}^{j-1} \vecm_i$ for $j = \ell+1, \dots, n$. Note that this gives 
$\vecy - \vecm_\ell - \cdots - \vecm_{j} = \vecy - \tilde \vecm_j$
and also
$\vecm_{j} =  \tilde \vecm_{j} - \tilde \vecm_{j-1}.$
Inserting these new variables, dropping the tildes, and using the definition of $\scrT_{\ell,n}$ yields the result.
\end{proof}

\section{The Boltzmann-Grad limit}\label{sec:boltzmann-grad}

Recall the semiclassical Boltzmann-Grad scaling \eqref{BGscaling} given by
\begin{equation}
D_{r,h} a(\vecx,\vecy) =  r^{d(d-1)/2} h^{d/2} \, a( r^{d-1} \vecx, h \vecy).
\end{equation}
Performing the Fourier transform in the $\vecx$ variable yields the expression
\begin{equation}
\tilde D_{r,h} \tilde a(\veceta,\vecy) = 
\widetilde{(D_{r,h} a)}(\veceta,\vecy)  =r^{-d(d-1)/2} h^{d/2}\, \tilde a( r^{1-d} \veceta, h \vecy),
\end{equation}
and thus after quantizing the rescaled observables we see
\begin{align}
\widehat\Op(D_{r,h} a)(\vecy,\vecy') &=  r^{-d(d-1)/2} h^{d/2}\, \tilde a( r^{1-d} (\vecy-\vecy'),  \tfrac{h}{2}(\vecy+\vecy')). 
\end{align}
Note that after this rescaling we have the relation
\begin{equation}
\begin{split}
D_{r,h}L_0(t) a(\vecx,\vecy) & = r^{d(d-1)/2} h^{d/2} \, L_0(t) a( r^{d-1} \vecx, h \vecy) \\
& = r^{d(d-1)/2} h^{d/2} \, a( r^{d-1} \vecx - t h \vecy, h \vecy) \\
& =  D_{r,h} a ( \vecx - t h r^{1-d} \vecy, \vecy) \\
& = L_0(t h r^{1-d}) D_{r,h} a( \vecx, \vecy)
\end{split}
\end{equation}
and so the Egorov property \eqref{egorov} becomes
\begin{equation} \label{egorovrescaled}
U_0(t h r^{1-d}) \Op( D_{r,h} a) U_0(-t h r^{1-d}) = \Op (D_{r,h} L_0(t)a). 
\end{equation}

Given a linear operator $A$ on $\L^2(\RR^d)$ with Schwartz kernel in $\scrS(\RR^d\times\RR^d)$, we define the partial trace 
\begin{equation}
\Tr_\vecalf A % = \Tr(\Pi_\vecalf A \,\Pi_\vecalf) 
= \sum_{\vecm\in\ZZ^d} \hat A(\vecm+\vecalf,\vecm+\vecalf) ,
\end{equation}
and note that in view of Lemma \ref{lem:projectionHSnorm}
$ \langle\Pi_\vecalf A, B \rangle_\HiS = \Tr_\vecalf A B^\dagger $.
Let us furthermore define $\scrI_{\ell,n}$, implicitly dependent on $r$ and $h$, by
\begin{equation} \label{def:scrI}
\scrI_{\ell,n}(\vecs) = 
\begin{cases}
\Tr_\vecalf [\Op(D_{r,h} a) \Op(D_{r,h} b)] & (\ell= n=0)\\
\Tr_\vecalf [K_{1,\ell}(\vecs) \Op(D_{r,h} a) K_{\ell+1,n}(\vecs) \Op(D_{r,h} b)]  & (1\leq \ell <n)\\
\Tr_\vecalf [K_{1,n}(\vecs) \Op(D_{r,h} a) \Op(D_{r,h} b)] & (0<\ell = n)\\
\Tr_\vecalf [\Op(D_{r,h} a) K_{1,n}(\vecs) \Op(D_{r,h} b)]  & (\ell=0<n) .
\end{cases}
\end{equation}
In view of equation \eqref{propeq2terms}, we have for $n=0,1,2$ 
\begin{align}\label{theexpansion}
Q_n(t, D_{r,h} a, D_{r,h} \overline b) &= (2\pi\i)^n \sum_{\ell=0}^n  (-1)^\ell \, \int_{ \substack{ 0<s_1<\cdots<s_\ell< t   \\ 0<s_n<\cdots<s_{\ell+1}< t}} \scrI_{\ell,n}(\vecs) \d \vecs.
\end{align}
(We work with $\overline b$ rather than $b$ to simplify the notation in the calculations that follow.)
In other words, the $\scrI_{\ell,n}$ are precisely the expressions that appear in the expansion of
\begin{equation}
\langle \Pi_\vecalf U_\lambda( t) U_0(- t) \Op(D_{r,h}a) U_0( t) U_\lambda(-  t),\Op(D_{r,h}\overline b)\rangle_{\HiS},
\end{equation}
cf.~\eqref{propeq2}.

Let us write down the $\scrI_{\ell,n}$ explicitly. For $1\leq \ell <n$, we show in the Appendix \ref{appendix1} that one has
\begin{equation} \label{Is}
\begin{split}
\scrI_{\ell,n}(\vecs) &=r^{nd} h^d \int_{\RR^d}  \sum_{\vecm_1,\dots,\vecm_{n}} \\
& \times \e(-\tfrac12 \, s_1\|\vecm_n+\vecalf\|^2) \hat W(r(\vecm_n-\vecm_1)) 
\scrT_{1,\ell-1}^-(\vecalf) \e(\tfrac12 \, s_\ell \|\vecm_\ell+\vecalf \|^2) \\
& \times \tilde a (-\veceta, h (\vecm_\ell+\vecalf+\tfrac12 r^{d-1}\veceta)) \,  \e(-\tfrac12 \, s_{\ell+1}\|\vecm_\ell+\vecalf + r^{d-1} \veceta\|^2) \\
& \times\hat W(r(\vecm_{\ell}-\vecm_{\ell+1})) 
\scrT_{\ell+1,n-1}^-(\vecalf + r^{d-1} \veceta) \e(\tfrac12 \, s_n \|\vecm_n+\vecalf + r^{d-1} \veceta \|^2) \\
& \times \tilde b( \veceta, h (\vecm_n+\vecalf + \tfrac12 r^{d-1} \veceta)) \, \d \veceta+ O(r^\infty)
\end{split}
\end{equation}
with the definition
\begin{equation}
\scrT^-_{\ell,n}(\vecy) 
=
\begin{cases}
\prod_{j=\ell}^{n} \e(\tfrac12 \, (s_j-s_{j+1}) \|\vecy+\vecm_j \|^2) \hat W(r(\vecm_j-\vecm_{j+1})) & (\ell\leq n)
\\
1 &  (\ell>n) .
\end{cases}
\end{equation}
The symbol $O(r^\infty)$ is a shorthand for ``$O_\beta(r^\beta)$ for any $\beta\geq 1$.''
It follows more immediately from the definition of $K_{\ell,n}$ that for $\ell=n$,
\begin{equation}\label{Isn}
\begin{split}
\scrI_{n,n}(\vecs) &= r^{nd} h^d \int_{\RR^d}  \sum_{\vecm_1,\ldots,\vecm_n\in\ZZ^d}  \e(-\tfrac12 \, s_1\|\vecm_n+\vecalf\|^2) \hat W(r(\vecm_n-\vecm_1)) \\
 & \times \scrT^-_{1,n-1}(\vecalf) \e(\tfrac12 \, s_n \|\vecm_n+\vecalf \|^2) \tilde a(-\veceta,h(\vecm_n+\vecalf+\tfrac12 r^{d-1}\veceta)) \\
 &\times \tilde b(\veceta,h(\vecm_n+\vecalf+\tfrac12 r^{d-1}\veceta)) \, \d \veceta+O(r^\infty),
 \end{split}
\end{equation}
and for $\ell = 0$,
\begin{equation} \label{Is0}
\begin{split}
\scrI_{0,n}(\vecs)& = r^{nd} h^d \int_{\RR^d} \sum_{\vecm_1,\ldots,\vecm_n\in\ZZ^d}   \tilde a(-\veceta ,h(\vecm_n+\vecalf+\tfrac12r^{d-1}\veceta) )  \\
&  \times \e(-\tfrac12 \, s_1\|\vecm_n+\vecalf+r^{d-1}\veceta \|^2) \,\hat W(r(\vecm_n-\vecm_1)) \scrT^-_{1,n-1}(\vecalf+r^{d-1}\veceta) \\
 & \times  \e(\tfrac12 \, s_n \|\vecm_n +\vecalf+r^{d-1}\veceta \|^2)  \tilde b(\veceta, h(\vecm_n+\vecalf+\tfrac12r^{d-1}\veceta)) \, \d \veceta +O(r^\infty).
\end{split}
\end{equation}

\section{Orders zero and one}\label{sec:0and1}

The asymptotics for zeroth and first order terms follows from the Poisson summation formula.

\begin{lem} \label{scrI0}
\begin{align}
\scrI_{0,0} = \int_{\RR^{d} \times \RR^d} a(\vecx, \vecy)  b( \vecx,  \vecy)\, \d \vecx \d \vecy+ O(h^\infty).  
\end{align}
\end{lem}
\begin{proof}
We have (by Lemma \ref{lem:projectionHSnorm})
\begin{align}
\scrI_{0,0} &= h^{d} \sum_{\vecm \in \ZZ^d} \int_{\RR^d} \tilde a(- \veceta,  h(\vecm+\vecalf+ \tfrac12 r^{d-1}\veceta)) \tilde b( \veceta,  h(\vecm+\vecalf + \tfrac12r^{d-1}\veceta)) \, \d \veceta .
\end{align}
Since $\tilde a$ and $\tilde b$ are Schwartz class, applying Poisson summation in $\vecm$ gives
\begin{equation} \begin{split}
\scrI_{0,0} &=  \int_{\RR^d\times\RR^d} \tilde a(- \veceta,  \vecy) \tilde b( \veceta,  \vecy) \, \d \veceta \, \d \vecy+ O(h^\infty) \\
 &= \int_{\RR^{d} \times \RR^d}  a(\vecx, \vecy)  b( \vecx,  \vecy) \,\d \vecx\, \d \vecy + O(h^\infty) .
\end{split} \end{equation}
\end{proof}
Recall that the mean free flight time is of the order of $r^{1-d}$, and that according to \eqref{lambdaisrescaled} we should consider time in units of $h$. This suggests the rescaling $t \to hr^{1-d} t$, and thus, by the Egorov property \eqref{egorovrescaled}, we obtain for the propagated symbol
\begin{equation} \begin{split}
& \Tr_\vecalf [U_0(t h r^{1-d}) \Op (D_{r,h} a) U_0(-t h r^{1-d}) \Op (D_{r,h} b) ]  \\ 
& = \Tr_\vecalf [\Op (D_{r,h} L_0(t)a) \Op(D_{r,h} b) ] \\
&= \int_{\RR^d\times\RR^d} (L_0(t)a)( \vecx, \vecy ) b( \vecx, \vecy) \d \vecx\, \d \vecy +O(h^\infty)  \\
& = \int_{\RR^d\times\RR^d} a( \vecx-t\vecy, \vecy ) b( \vecx, \vecy) \d \vecx\, \d \vecy +O(h^\infty) ,
\end{split} \end{equation}
uniformly for all $t$ in a fixed compact interval. It is worth noting that this is precisely the answer one would expect: at order zero the potential does not appear, which means the solution simply displays free evolution. We see this is true by virtue of the fact that the initial density has simply been translated in position space for time $t$ with momentum $\vecy$. 

\begin{lem} \label{scrI1}
\begin{align}
\scrI_{0,1}(s_1) - \scrI_{1,1}(s_1) & =  O(r^ d h^\infty+r^{\infty}).
\end{align}
\end{lem}
\begin{proof}  By \eqref{Isn},
\begin{equation} \begin{split}
\scrI_{1,1}(s_1) & = r^{d} h^{d} \hat W(\vecnull) \sum_{\vecm\in\ZZ^d} \int_{\RR^d} 
 \tilde a(-\veceta,h(\vecm+\vecalf +\tfrac12 r^{d-1}\veceta)) \\
&\times \tilde b(\veceta, h(\vecm+\vecalf +\tfrac12r^{d-1}\veceta)) \, d\veceta +O(r^\infty) \\
& = r^{d} \hat W(\vecnull) \int_{\RR^d \times \RR^d} \tilde a( -\veceta, \vecy) \tilde b( \veceta, \vecy)\, \d\veceta\, \d \vecy + O(r^ d h^\infty+r^{\infty}),
\end{split} \end{equation}
again by Poisson summation. 
%Integrating over $s_1$ from $0$ to $h r^{1-d} t$ yields
%\begin{align}
%\int_0^{h r^{1-d}t} \scrI_{1,1}(s_1) \d s_1 &
%= h r t \hat W(\vecnull) \int_{\RR^d\times\RR^d} \tilde a( -\veceta, \vecy) \tilde b( \veceta, \vecy) \d\veceta \d \vecy + O(r h^\infty+h r^{\infty})\\
%&
%= h r t \hat W(\vecnull) \int_{\RR^d\times\RR^d} a(\vecx, \vecy) b( \vecx, \vecy) \d\vecx \d \vecy + O(r h^\infty+h r^{\infty}) .
%\end{align}
Similarly, using \eqref{Is0},
\begin{equation} \begin{split}
\scrI_{0,1}(s_1)
& = r^{nd} h^d \hat W(\vecnull) \sum_{\vecm_1\in\ZZ^d}   \int_{\RR^d}  \tilde a(-\veceta ,h(\vecm_1+\vecalf+\tfrac12 r^{d-1}\veceta) )\\
& \times \tilde b(\veceta,h(\vecm_1+\vecalf+\tfrac12 r^{d-1}\veceta)) \, d\veceta +O(r^\infty) \\
& = r^{d} \hat W(\vecnull) \int_{\RR^d \times \RR^d} \tilde a( -\veceta, \vecy) \tilde b( \veceta, \vecy) \, \d\veceta\, \d \vecy + O(r^ d h^\infty+r^{\infty}) .
\end{split} \end{equation}
\end{proof}

 Indeed, in the expansion \eqref{theexpansion} the terms $\scrI_{1,1}(s_1)$ and $\scrI_{0,1}(s_1)$ appear with opposite sign and therefore cancel up to an error $O(r^ d h^\infty+r^{\infty})$. The total error term after integrating over $s_1$ is thus obtained by multiplying this by the integration range of size $h r^{1-d} t$.

\section{Equidistribution of horocycles}\label{sec:equidistribution}

At second order we will use the fact that the $\scrI_{\ell,n}$ can be written as functions on some non-compact, finite volume manifold. Specifically, consider the semi-direct product group
$G=\SL(2,\RR)\ltimes\RR^{2d}$ with multiplication law
\begin{equation}
(M,\vecxi)(M',\vecxi') =(MM',\vecxi+M\vecxi') ,
\end{equation}
where $M,M'\in\SL(2,\RR)$ and $\vecxi,\vecxi'\in\RR^d\times\RR^d$;
the action of $\SL(2,\RR)$ on $\RR^d\times\RR^d$ is defined canonically as
\begin{equation}
M\vecxi = \begin{pmatrix}
a\vecx +b\vecy \\
c\vecx +d\vecy \\
\end{pmatrix}, \quad 
M=\begin{pmatrix}
a&b\\
c&d
\end{pmatrix},\quad 
\vecxi= \begin{pmatrix}
\vecx\\
\vecy
\end{pmatrix},
\end{equation}
where $\vecx,\vecy\in\RR^d$.
A convenient parametrization of $\SL(2,\RR)$ can be obtained
by means of the Iwasawa decomposition
\begin{equation}
M = n_-(u) \, \Phi^{-\log v} \, \scrR(\phi) 
\end{equation}
with 
\begin{equation}
n_-(u) = \begin{pmatrix} 1 & u \\ 0 & 1 \end{pmatrix}, \quad \Phi^t = \begin{pmatrix} e^{-t/2} & 0 \\ 0 & e^{t/2} \end{pmatrix}, \quad \scrR(\phi) = \begin{pmatrix} \cos \phi & -\sin \phi \\ \sin \phi & \cos \phi \end{pmatrix}.
\end{equation}

This decomposition is unique for $\tau=u+\i v\in\H$, $\phi\in[0,2\pi)$, where $\H$ denotes the upper half plane $\H=\{ \tau \in\CC: \Im\tau>0\}$. We will use the notation $M=(\tau,\phi)$ and $(M,\vecxi)=(\tau,\phi,\vecxi)$ interchangeably. With this, we have for instance
$n_-(u) \Phi^{-2\log r} = (u+\i r^2,0 )$
and
\begin{equation}
\left(1,\begin{pmatrix} \vecnull \\ \yy \end{pmatrix} \right) n_-(u) \Phi^{-2\log r} = \left(u+\i r^2,0,\begin{pmatrix} \vecnull \\ \yy \end{pmatrix} \right).
\end{equation}

Throughout this section, let $\Gamma$ be a subgroup of $\SL(2,\ZZ)\ltimes(\tfrac12\ZZ)^{2d}$ of finite index. The Haar measure on $G$ induces a $G$-invariant measure on $\GamG$, which will be denoted by $\mu$. Since $\Gamma$ is a lattice in $G$, we have (by definition) $0<\mu(\GamG)<\infty$.

\begin{prop}\label{equim}
Fix $\yy \in \RR^d \backslash \QQ^d$ so that the components of $(1,\trans\yy)$ linearly independent over $\QQ$. Let $w:\RR\to\RR$ piecewise continuous with compact support. Let $F: \Gamma \backslash G \times \RR \to \RR$ be bounded continuous, and $F_r$ be a sequence of continuous, uniformly bounded functions $ \Gamma \backslash G \times\RR \to \RR$ such that $F_r\to F_0$ uniformly on compacta as $r\to 0$. Then, for $\sigma \geq 0$, we have
\begin{multline}
 \lim_{r \to 0}  r^\sigma \int_{\RR} F_r((u+\i r^2,0,(\begin{smallmatrix} \vecnull \\ \yy \end{smallmatrix}) ), r^\sigma u)\, w(r^\sigma u)\, \d u \\
 = \frac{1}{\mu(\GamG)}\int_{\Gamma \backslash G} \int_{\RR} F_0(g,u) \, w(u) \,\d u \,\d \mu(g).
\end{multline}
\end{prop}
\begin{proof}
The proof of Theorem 5.1 in \cite{Marklof02} 
tells us that for $F : \Gamma \backslash G \to \RR$ bounded continuous, we have
\begin{equation} \label{step1}
\lim_{r\to 0} r^\sigma \int_\RR F((u+\i r^2,0,(\begin{smallmatrix} \vecnull \\ \yy \end{smallmatrix})  )) \, w(r^\sigma u)\, \d u = \frac{1}{\mu(\Gamma\backslash G)} \int_{\Gamma\backslash G} F \d \mu \int_\RR w(u) \d u.
\end{equation}
The claim now follows from the same argument as \cite[Theorem 5.3]{partI}.
\end{proof}

We define the subgroup $\Gamma_\infty$ by
\begin{align}
\Gamma_\infty = \left\{ \begin{pmatrix} 1 & m \\ 0 & 1 \end{pmatrix} : m \in \ZZ \right\} \subset \SLZ
\end{align}
and for $\gamma = \left(\begin{smallmatrix} a & b \\ c & d \end{smallmatrix}\right)$ use the notation
\begin{align}
v_\gamma := \Im(\gamma \tau) = \frac{v}{|c\tau+d|^2}, \quad \vecy_\gamma := c \vecx + d \vecy.
\end{align}
Then, with $\chi_R$ the characteristic function of $[R,\infty)$ we define the characteristic function $X_R : \fH \to \RR_{\geq 0}$ by
\begin{align}
X_R(\tau) = \sum_{\gamma \in (\Gamma_\infty \cup - \Gamma_\infty) \backslash \SLZ} \chi_R(v_\gamma).
\end{align}
Note that by construction $X_R$ is $\SLZ$-invariant.
For $f:\RR\to\RR_{\geq 0}$ of rapid decay at $\pm\infty$ and $\beta\in\RR$, the function $\Psi_{R,f}^\beta : G \to \RR_{\geq 0}$ is defined by
\begin{align}\label{def:PsiRf}
\Psi_{R,f}^\beta(\tau,\vecxi) = \sum_{\gamma \in \Gamma_\infty \backslash \SLZ} \sum_{\vecm \in \ZZ^d} f( (\yy_\gamma + \vecm) v_\gamma^{1/2}) \, v_\gamma^{\beta d/2} \, \chi_R(v_\gamma),
\end{align}
and for convenience when $\beta=1$ we  write $\Psi_{R,f} :=\Psi_{R,f}^1$.
The function $\Psi_{R,f}^\beta$ is left-invariant under $\SLZ\ltimes(\tfrac12\ZZ)^{2d}$. Both $X_R$ and $\Psi_{R,f}^\beta$ can thus be viewed as functions on $G$ and, since $\Gamma$ is a finite-index subgroup of $\SLZ\ltimes(\tfrac12\ZZ)^{2d}$, are also left $\Gamma$-invariant.

\begin{prop}{\cite[Proposition 6.4]{Marklof02}} \label{diophprop}
Let $\yy$ be Diophantine of type $\kappa$, $w:\RR\to\RR$ piecewise continuous with compact support, and $0<\epsilon <1$ and $0<\epsilon'<1/(\kappa-1)$. Then, for every $R\geq 1$,
\begin{multline}
\limsup_{r \to 0} r^{d-2} \int_{|u|>r^{2-\epsilon}} \Psi_{R,f}\left(u+\i r^2, \begin{pmatrix} \vecnull \\ \yy \end{pmatrix}\right) \, w(r^{d-2} u)\, \d u \\
\ll_{\epsilon,\epsilon'} R^{-(1/(\kappa-1)-d+2)/2} + R^{-\epsilon'/2}.
\end{multline}
\end{prop}

Note that the term $R^{-\epsilon'/2}$ is only relevant for $d=2$. The expression vanishes as $R \to \infty$ if $\kappa < (d-1)/(d-2)$. The following generalization to $\beta<1$ holds. Note the range of integration is now over all $u \in \RR$.

\begin{prop} \label{diophprop234}
Let $0\leq \beta<1$, $\yy$ be Diophantine of type $\kappa$, $w:\RR\to\RR$ piecewise continuous with compact support. Then, for every $R\geq 1$,
\begin{multline}
\limsup_{r \to 0} r^{d-2} \int_{\RR} \Psi_{R,f}^\beta \left(u+\i r^2, \begin{pmatrix} \vecnull \\ \yy \end{pmatrix} \right) \, w(r^{d-2} u)\, \d u \\
\ll R^{-(1/(\kappa-1)-\beta d+2)/2} + R^{(\beta-1)d/2}.
\end{multline}
\end{prop}

The right hand side vanishes as $R \to \infty$ if and only if \begin{equation}\kappa < \begin{cases} \infty & \beta \leq 2/d  \\ (\beta d-1)/(\beta d-2) & \beta > 2/d \end{cases}.\end{equation} In practice, we  want both Propositions \ref{diophprop} and \ref{diophprop234} to hold simultaneously. We do this by taking $\kappa < (d-1)/(d-2)$ and use the fact that for $2/d \leq \beta < 1$ we have $(\beta d- 1)/(\beta d -2) > (d-1)/(d-2)$.

\begin{proof}
Writing $\tau = u+\i v$ and $v = r^{2}$ we have the explicit representation
\begin{multline}\label{expliPsi}
\Psi_{R,f}^\beta \left(\tau,\begin{pmatrix} \vecnull \\ \yy \end{pmatrix} \right)
= 2 \sum_{\vecm \in \ZZ^d} f\left( \vecm \frac{v^{1/2}}{|\tau|} \right) \frac{v^{\beta d/2}}{|\tau|^{\beta d}} \chi_R \left( \frac{v}{|\tau|^2} \right) \\
+ 2 \sum_{\substack{(c,d) \in \ZZ^2 \\ gcd (c,d) = 1 \\ c>0, d\neq 0}} \sum_{\vecm \in \ZZ^d} f\left( (d\yy+\vecm) \frac{v^{1/2}}{|c\tau+d|}\right) \frac{v^{\beta d/2} }{|c\tau+d|^{\beta d}} \chi_R \left( \frac{v}{|c\tau+d|^2} \right).
\end{multline}
For the first term we make the substitution $u = vt$ in the integral, which yields
\begin{multline}
2v^{d/2-1}\int_{\RR} \, w(v^{d/2-1}u) \sum_{\vecm \in \ZZ^d} f\left( \vecm \frac{v^{1/2}}{|\tau|} \right) \frac{v^{\beta d/2}}{|\tau|^{\beta d}} \chi_R \left( \frac{v}{|\tau|^2} \right) \d u \\
= 2 v^{(1-\beta)d/2} \int_{\RR} \frac{w(v^{d/2}t)}{(1+t^2)^{\beta d/2}} \sum_{\vecm \in \ZZ^d} f\left(  \frac{\vecm}{v^{1/2}(1+t^2)^{1/2}} \right) 
 \chi_R \left( \frac{1}{v(1+t^2)} \right) \d t .
\end{multline}
Under the assumption that $0< \beta < 1$ we have
\begin{equation}\label{eq:this2}
\frac{v^{(1-\beta)d/2}}{(1+t^2)^{\beta d/2}} \, \chi_R \left( \frac{1}{v(1+t^2)} \right)  \leq  \frac{R^{(\beta-1)d/2}}{(1+t^2)^{d/2}} \,   \chi_R \left( \frac{1}{v(1+t^2)} \right) 
\end{equation}
and thus obtain the bound
\begin{multline}
\limsup_{v\to 0} 2v^{d/2-1}\int_{\RR} w(v^{d/2-1}u) \sum_{\vecm \in \ZZ^d} f\left( \vecm \frac{v^{1/2}}{|\tau|} \right) \frac{v^{\beta d/2}}{|\tau|^{\beta d}} \chi_R \left( \frac{v}{|\tau|^2} \right) \d u  \\ 
 \leq 2 R^{(\beta-1)d/2} \, w(0) \, f(\vecnull) \, \int_{\RR} \frac{\d t}{(1+t^2)^{d/2}} + O(R^{-\infty}).
\end{multline}
For the second term, using 
\begin{align}
& \frac{v^{\beta d/2}}{|c\tau+d|^{\beta d}} \chi_R \left( \frac{v}{|c\tau+d|^2} \right)  \leq \frac{v^{ d/2}}{|c\tau+d|^{ d}} R^{(\beta-1)d/2} \chi_R \left( \frac{v}{|c\tau+d|^2} \right),
\end{align}
we see that
\begin{multline}
\sum_{\substack{(c,d) \in \ZZ^2 \\ gcd (c,d) = 1 \\ c>0, d\neq 0}} \sum_{\vecm \in \ZZ^d} v^{d/2-1} \int_\RR f\left( (d \yy + \vecm) \frac{v^{1/2}}{|c\tau+d|} \right) \times \\[-0.6cm] \times
\frac{v^{\beta d/2}}{|c\tau+d|^{\beta d}} \chi_R \left( \frac{v}{|c\tau+d|^2} \right) w(v^{d/2-1} u) \d u.
\end{multline}
is bounded above by
\begin{multline}\label{eq:that2}
R^{(\beta-1)d/2} \sum_{\substack{(c,d) \in \ZZ^2 \\ gcd (c,d) = 1 \\ c>0, d\neq 0}} \sum_{\vecm \in \ZZ^d} v^{d/2-1} \int_\RR f\left( (d \yy + \vecm) \frac{v^{1/2}}{|c\tau+d|} \right) \times \\[-0.6cm] \times
\frac{v^{ d/2}}{|c\tau+d|^{d}} \chi_R \left( \frac{v}{|c\tau+d|^2} \right) w(v^{d/2-1} u) \d u.
\end{multline}
This reduces the problem to the same calculation as in the proof of Proposition \ref{diophprop}, which yields that \eqref{eq:that2} is bounded above by
\begin{equation}
R^{(\beta-1)d/2}  (R^{-(1/(\kappa-1)-d+2)/2} + 1) = R^{-(1/(\kappa-1)-\beta d+2)/2} + R^{((\beta-1)d)/2} .
\end{equation}
\end{proof}

Fix a compact interval $A\subset\RR$. We say {\em $F:\GamG\times \RR\to\CC$ is dominated by $\Psi_{R,f}$ on $\GamG\times A$} if there are positive constants $L,R_0$ such that $|F((\tau,\phi,\vecxi),u') |X_R(\tau) \leq L(1 + \Psi_{R,f}(\tau,\phi,\vecxi))$ for all  $(\tau,\phi,\vecxi)\in G$, $u'\in A$ and $R\geq R_0$. A sequence of functions $F_r:\GamG\times \RR\to\CC$ is {\em uniformly dominated} if $L,R_0$ are independent of $r$.

\begin{prop}\label{propr}
Assume $\yy$ is Diophantine of type $\kappa < (d-1)/(d-2)$ with the components of $(1,\trans\yy)$ linearly independent over $\QQ$. Let $w:\RR\to\RR$ be piecewise continuous with compact support. Let $F_0: \Gamma \backslash G \times \RR \to \RR$ be continuous and dominated by $\Psi_{R,f}$ on $\GamG\times\supp w$. Let $F_r$ be a sequence of continuous functions $ \Gamma \backslash G \times\RR \to \RR$ uniformly dominated by $\Psi_{R,f}$ on $\GamG\times\supp w$,  such that $F_r\to F_0$ uniformly on compacta as $r\to 0$. Then for any $0 < \epsilon < 2$ we have
\begin{multline}
\lim_{r \to 0} r^{d-2} \int_{|u|>r^{2-\epsilon}} F_r((u+\i r^2,0, (\begin{smallmatrix} \vecnull \\ \yy \end{smallmatrix}) ), r^{d-2} u) \,w(r^{d-2} u)\, \d u \\
= \frac{1}{\mu(\Gamma \backslash G)} \int_\RR \int_{\Gamma \backslash G} F_0(g,u) \,w(u)\,\d \mu(g)\, \d u.
\end{multline}
\end{prop}

\begin{proof}
(This follows the proof of \cite[Theorem 6.8/Corollary 6.10]{Marklof02}.)
We may assume without loss of generality that $F_r$ and $w$ are real-valued and non-negative. Set 
\begin{align}
J_{r,R}((\tau,\phi,\vecxi) ,u') & = F_r((\tau,\phi,\vecxi) ,u')(1-X_R(\tau)).
\end{align}
Then $J_{r,R}$ is bounded and thus
\begin{multline}
 \int_{|u|>r^{2-\epsilon}} J_{r,R}((u+\i r^2,0, (\begin{smallmatrix} \vecnull \\ \yy \end{smallmatrix}) ), r^{d-2} u)\, w(r^{d-2}u) \, \d u \\
= \int_\RR J_{r,R}((u+\i r^2,0, (\begin{smallmatrix} \vecnull \\ \yy \end{smallmatrix}) ), r^{d-2} u)\, w(r^{d-2}u) \, \d u + O(r^{2-\epsilon}) .
\end{multline}
By Proposition \ref{equim}, which (by a standard probabilistic argument) extends to functions such as $J_{r,R}$ whose points of discontinuity are contained in a set of $\mu$-measure zero (alternatively simply smooth the characteristic function $\chi_R$ to make $J_{r,R}$ continuous),
\begin{multline} 
 \lim_{r \to 0} r^{d-2}\int_\RR J_{r,R}((u+\i r^2,0, (\begin{smallmatrix} \vecnull \\ \yy \end{smallmatrix}) ), r^{d-2} u)\, w(r^{d-2}u) \, \d u \\
= \frac{1}{\mu(\Gamma \backslash G)} \int_\RR  \int_{\Gamma \backslash G} J_{0,R}(g,u') w(u') \d \mu(g) \d u'.
\end{multline}

Furthermore, $F_0((\tau,\phi,\vecxi) ,u') X_R(\tau) \leq L X_R(\tau) +L \Psi_{R,f}(\tau,\vecxi)$ for large $R$, and hence
\begin{multline}
\int_\RR \int_{\Gamma\backslash G} F_0((\tau,\phi,\vecxi) ,u') X_R(\tau) w(u') \d \mu \d u'  \\ \leq \int_\RR  w(u') du'\int_{\Gamma\backslash G} (L X_R + L\Psi_{R,f}) \d \mu \ll R^{-1},
\end{multline}
cf.~\cite[\S 6.2]{Marklof02}.
Combining this with the result for $J_{0,R}$ yields
\begin{align}
\int_\RR  \int_{\Gamma\backslash G} J_{0,R}(g,u') w(u') \d \mu(g) \d u' = \int_\RR  \int_{\Gamma\backslash G} F_0(g,u') w(u') \d \mu(g) \d u' + O(R^{-1}).
\end{align}
In summary, we have shown thus far that
\begin{equation} \begin{split}
&\liminf_{r \to 0} r^{d-2}\int_{|u|>r^{2-\epsilon}} F_r((u+\i r^2,0,(\begin{smallmatrix} \vecnull \\ \yy \end{smallmatrix}) ), r^{d-2} u)\, w(r^{d-2}u) \, \d u \\
& \geq \lim_{r \to 0} r^{d-2}\int_{|u|>r^{2-\epsilon}} J_{r,R}((u+\i r^2,0, (\begin{smallmatrix} \vecnull \\ \yy \end{smallmatrix}) ), r^{d-2} u)\, w(r^{d-2}u) \, \d u \\
&=\frac{1}{\mu(\Gamma\backslash G)}\int_\RR  \int_{\Gamma\backslash G} F_0(g,u') \,w(u')\,\d \mu(g) \d u' + O(R^{-1}),
\end{split} \end{equation}
for every $R\geq R_0$.  For the upper bound we use that
\begin{align}
F_r((\tau,\phi,\vecxi),u') \leq F_r((\tau,\phi,\vecxi),u')(1-X_R(\tau)) + L X_R(\tau) + L \Psi_{R,f}(\tau,\vecxi).
\end{align}
We proceed as above for the first two terms, and apply Proposition \ref{diophprop} to the third to obtain 
\begin{equation} \begin{split}
&\limsup_{r \to 0} r^{d-2}\int_{|u|>r^{2-\epsilon}} F_r((u+\i r^2,0, (\begin{smallmatrix} \vecnull \\ \yy \end{smallmatrix}) ), r^{d-2} u)\, w(r^{d-2}u) \, \d u \\
& \leq \frac{1}{\mu(\Gamma \backslash G)} \int_\RR  \int_{\Gamma \backslash G} F_0(g,u') \d \mu(g) \d u' + O(R^{-(1/(\kappa-1) - d+2)/2} + R^{-\epsilon'/2}),
\end{split} \end{equation}
for every $R\geq R_0$.
\end{proof}

\section{Mean value theorems for theta functions}\label{sec:theta}

For $f\in\scrS(\RR^d\times\RR^d)$ and $\phi\in\RR$, define $f_\phi$ by 

\begin{align} \label{eightone}
f_\phi(\vecy_1,\vecy_2) = \begin{cases} f(\vecy_1,\vecy_2) & (\phi = 0\bmod 2\pi) \\ 
f(-\vecy_1,-\vecy_2) & (\phi = \pi\bmod 2\pi) \\ 
\int_{\RR^{2d}} G_\phi(\vecy_1,\vecy_2,\vecx_1,\vecx_2) \, f(\vecx_1,\vecx_2) \, \d \vecx_1 \d \vecx_2 & (\phi \neq 0\bmod \pi),
\end{cases}
\end{align}
where
\begin{multline}
G_\phi(\vecy_1,\vecy_2,\vecx_1,\vecx_2) \\
= |\sin \phi|^{-d} \, \e \left( \frac{\tfrac{1}{2}(\|\vecy_1\|^2+\|\vecx_1\|^2-\|\vecy_2\|^2-\|\vecx_2\|^2) \cos \phi - \vecy_1\cdot\vecx_1+\vecy_2\cdot\vecx_2}{\sin \phi}\right).
\end{multline}

\begin{lem} \label{lem:schwa}
If $f \in \scrS(\RR^{d}\times \RR^d)$ then $f_\phi \in \scrS(\RR^d\times\RR^d)$.
\end{lem}
\begin{proof} If $\phi = 0\bmod\pi$ then the result is immediate. For fixed $\phi \neq 0 \bmod\pi$, define
\begin{equation}\label{def:ggg}
g(\vecx_1,\vecx_2) =  \e \left( \frac{\tfrac{1}{2}(\|\vecx_1\|^2-\|\vecx_2\|^2) \cos \phi}{\sin \phi}\right) \, f(\vecx_1,\vecx_2).
\end{equation}
and its Fourier transform
\begin{equation}
I(\vecy_1,\vecy_2) = |\sin \phi|^{-d} \, \int_{\RR^{2d}} g(\vecx_1,\vecx_2)\, \e \left( \frac{- \vecy_1\cdot\vecx_1+\vecy_2\cdot\vecx_2}{\sin \phi}\right) \d \vecx_1 \d \vecx_2.
\end{equation}
Note that
\begin{equation}
f_\phi(\vecy_1,\vecy_2) = \e \left( \frac{\tfrac{1}{2}(\|\vecy_1\|^2-\|\vecy_2\|^2) \cos \phi}{\sin \phi}\right) I(\vecy_1,\vecy_2).
\end{equation}
Now $f \in \scrS(\RR^{d}\times \RR^d)$ implies $g \in \scrS(\RR^{d}\times \RR^d)$ (since all derivatives of the exponential factor in \eqref{def:ggg} grow at most polynomially), which implies $I \in \scrS(\RR^{d}\times \RR^d)$ (since the Fourier transform preserves Schwartz class; use integration by parts), which in turn implies $f_\phi \in \scrS(\RR^{d}\times \RR^d)$ (by the first argument).
\end{proof}

The following lemma provides rapid decay that is uniform in $\phi$.

\begin{lem}\label{decay101}
If $f\in\scrS(\RR^d\times\RR^d)$, then for all multi-indices $\vecbeta_1,\vecbeta_2 \in \ZZ_{\geq 0}^d$ and for every $T>1$ 
\begin{equation}
\sup_{\vecy_1,\vecy_2,\phi} (1+\|\vecy_1\|)^T (1+\|\vecy_2\|)^T | \partial_{\vecy_1}^{\vecbeta_1} \, \partial_{\vecy_2}^{\vecbeta_2}  \, f_\phi(\vecy_1,\vecy_2) | <\infty .
\end{equation}
\end{lem}

\begin{proof}
The proof of Lemma \ref{lem:schwa} shows that 
\begin{equation}\label{estim88}
\sup_{\vecy_1,\vecy_2,\phi\in I} (1+\|\vecy_1\|)^T (1+\|\vecy_2\|)^T | \partial_{\vecy_1}^{\vecbeta_1} \, \partial_{\vecy_2}^{\vecbeta_2}  \, f_\phi(\vecy_1,\vecy_2) | <\infty 
\end{equation}
for any closed interval $I$ not containing $\phi=0\bmod\pi$. As in the proof of \cite[Lemma 4.3]{Marklof03}, we represent $f_{\phi+\pi/2}=\int_{\RR^{2d}} G_\phi(\vecy_1,\vecy_2,\vecx_1,\vecx_2) \, f_{\pi/2}(\vecx_1,\vecx_2) \, \d \vecx_1 \d \vecx_2$ using the Fourier transform $f_{\pi/2}$ of $f$. Since $f_{\pi/2}\in\scrS(\RR^d\times\RR^d)$, we see that \eqref{estim88} holds for any closed interval not containing $\phi=\frac{\pi}{2}\bmod\pi$. Both cases taken together, this shows that \eqref{estim88} holds in fact for every closed interval $I$, and so in particular for $I=[0,2\pi]$. This proves the claim in view of the $2\pi$-periodicity of $f_\phi$.
\end{proof}

We define the theta function $\Theta_f:G\mapsto \CC$ by
\begin{multline} \label{def:theta}
\Theta_f\bigg(u+\i v, \phi, \begin{pmatrix}  \vecx \\ \vecy \end{pmatrix}  \bigg) 
= v^{d/2} \sum_{\vecm_1,\vecm_2\in\ZZ^d} f_\phi(v^{1/2} (\vecm_1-\vecy),v^{1/2} (\vecm_2-\vecy)) \\
\times \e(\tfrac12 \, u (\|\vecm_1-\vecy\|^2- \| \vecm_2-\vecy\|^2) + \vecx \cdot (\vecm_1-\vecm_2) ) .
\end{multline}
Since $f_\phi\in\scrS(\RR^d\times\RR^d)$ we have that $\Theta_f\in\C^\infty(G)$. Let 
\begin{equation}
\Gamma=\left\{ 
\left( \begin{pmatrix}
a & b \\
c & d
\end{pmatrix}
, 
\begin{pmatrix}
ab \vecs \\
cd \vecs
\end{pmatrix}
+\vecm \right)
: \; \begin{pmatrix}
a & b \\
c & d
\end{pmatrix} \in\SL(2,\ZZ),\; \vecm\in\ZZ^{2d}
\right\} \subset G,
\end{equation}
with
$\vecs=(\tfrac12,\tfrac12,\ldots,\tfrac12)\in\RR^d$. 
Then $\Gamma$ is of finite index in $\SL(2,\ZZ)\ltimes(\tfrac12\ZZ)^{2d}$, and $\Theta_f$ is left $\Gamma$ invariant; cf.~\cite[Prop.~4.9]{Marklof03}. That is, $\Theta_f\in\C^\infty(\GamG)$.

\begin{prop} \label{prop:thetaleadingorder} Let $f \in \scrS(\RR^d\times\RR^d)$. Then
\begin{align}
\Theta_f(u+\i v,\phi,\vecxi) = v^{d/2} \sum_{\vecm\in\ZZ^d} f_\phi( (\vecm-\vecy) \, v^{1/2}, (\vecm -\vecy) \, v^{1/2}) + O(v^{-\infty})
\end{align}
uniformly for all $(u+\i v,\phi,\vecxi) \in G$ with $v > 1/2$. 
\end{prop}

\begin{proof}
See \cite[Prop.~4.10]{Marklof03}.
\end{proof}

\begin{cor}
Let $f \in \scrS(\RR^d \times \RR^d)$, then for all $T>1$ we have that $\Theta_f$ is dominated by $\Psi_{R,\bar f}$ for \begin{equation}\label{def:fbar}
\bar f (\vecx) = (1+\| \vecx \|)^{-{2T}}.\end{equation} 
\end{cor}
\begin{proof}
This follows from Proposition \ref{prop:thetaleadingorder} and Lemma \ref{decay101} (with $\vecbeta_1=\vecbeta_2=\vecnull)$.
\end{proof}

\begin{prop} \label{maincorollary}
Assume $\yy$ is Diophantine of type $\kappa < (d-1)/(d-2)$ with the components of $(1,\trans\yy)$ linearly independent over $\QQ$. Let $w:\RR\to\RR$ piecewise continuous, continuous at 0, with compact support.  Then
\begin{multline}
\lim_{r \to 0} r^{d-2} \int_{\RR} \Theta_f \left( u+\i r^2,0, \left( \begin{matrix} \vecnull \\ \yy \end{matrix} \right) \right) \, w(r^{d-2} u) \, \d u \\ 
=2\,w(0)\,\int_{\RR^d\times\RR^d} f(\vecy_1,\vecy_2) \, \delta(\|\vecy_1\|^2-\|\vecy_2\|^2)\, \d \vecy_1 \d \vecy_2\\
+ \int_{\RR^d} f(\vecy_1,\vecy_1)\, \d \vecy_1 \int_\RR w(u) \, \d u.
\end{multline}
\end{prop}

\begin{proof}
Fix $0 < \epsilon < 1$, and split the integration over $u$ into the regions $|u| < r^{2-\epsilon}$ and $|u| > r^{2-\epsilon}$. In the first region, the proof of \cite[Lemma 7.3]{Marklof02} shows that
\begin{equation} 
\begin{split} 
&r^{d-2} \int_{|u|<r^{2-\epsilon}} \Theta_f \left( u+\i r^2,0, \left( \begin{matrix} \vecnull \\ \yy \end{matrix} \right) \right) \, w(r^{d-2} \, u) \, \d u \\[0.3cm]
& = r^{-2}  \int_{|u|<r^{2-\epsilon}}   \bigg( \int_{\RR^d\times\RR^d} f(\vecy_1,\vecy_2) \\
& \qquad \times \e(\tfrac{1}{2}(\|\vecy_1\|^2-\|\vecy_2\|^2) r^{-2} u ) \,  \d \vecy_1  \d \vecy_2 \bigg) w(r^{d-2} u) \, \d u +o(1)\\[0.3cm]
& = 2 \, w(0) \, \int_{\RR^d\times\RR^d} f(\vecy_1,\vecy_2) \, \delta(\|\vecy_1\|^2-\|\vecy_2\|^2 ) \, \d \vecy_1  \d \vecy_2 +o(1).
\end{split} 
\end{equation}
Since $\Theta_f$ is dominated by $\Psi_{R,f}$, for the region $|u|>r^{2-\epsilon}$ we can apply Proposition \ref{propr} and note that the limit can be written as
\begin{align}
& \frac{1}{ \mu(\GamG)} \int_{\GamG} \Theta_f \, \d \mu  \, \int_\RR w(u) \, \d u 
= \int_{\RR^d}  f(\vecy_1, \vecy_1) \, \d\vecy_1  \, \int_\RR w(u) \, \d u,
\end{align}
cf.~\cite[Lemma 7.2]{Marklof02}. 
\end{proof}

We will now deal with $f$ that depend continuously on additional parameters $u\in\RR$, $\veceta\in\RR^d$. We denote by $\tilde \scrS$ the class of functions $f\in\C(\RR^d\times\RR^d\times\RR\times\RR^d)$ with the property that for every multi-index $\vecbeta_1,\vecbeta_2 \in \ZZ_{\geq 0 }^d$ the derivative $\partial_{\vecy_1}^{\vecbeta_1} \, \partial_{\vecy_2}^{\vecbeta_2} \, f(\vecy_1,\vecy_2,u,\veceta)$ (a) exists, (b) is continuous (in all variables), and (c) is rapidly decaying, i.e., 
\begin{equation}
\sup_{\vecy_1,\vecy_2,u,\veceta} (1+\|\vecy_1\|)^T (1+\|\vecy_2\|)^T (1+|u|)^T (1+\|\veceta\|)^T \big| \partial_{\vecy_1}^{\vecbeta_1} \, \partial_{\vecy_2}^{\vecbeta_2}  \, f(\vecy_1,\vecy_2,u,\veceta) \big| <\infty
\end{equation}
for every $T> 1$. For $f \in \tilde\scrS$ we define $f_\phi\in\tilde\scrS$ in analogy with \eqref{eightone} by
\begin{align} \label{eightone2367}
f_\phi(\vecy_1,\vecy_2,u,\veceta) = \begin{cases} f(\vecy_1,\vecy_2,u,\veceta) & (\phi = 0\bmod 2\pi) \\ 
f(-\vecy_1,-\vecy_2,u,\veceta) & (\phi = \pi\bmod 2\pi) \\ 
\int_{\RR^{2d}} G_\phi(\vecy_1,\vecy_2,\vecx_1,\vecx_2) \, f(\vecx_1,\vecx_2,u,\veceta) \, \d \vecx_1 \d \vecx_2 & (\phi \neq 0\bmod \pi).
\end{cases}
\end{align}
The fact that $f_\phi\in\tilde\scrS$ follows from the same argument as in Lemma \ref{lem:schwa}.
We also have the following.

\begin{lem}\label{decay101b}
If $f\in\tilde\scrS$, then for all multi-indices $\vecbeta_1,\vecbeta_2 \in \ZZ_{\geq 0}^d$ and every $T>1$ 
\begin{multline} \label{decay101b:eq}
\sup_{\vecy_1,\vecy_2,u,\veceta,\phi}  (1+\|\vecy_1\|)^T (1+\|\vecy_2\|)^T (1+|u|)^T \\ \times
(1+\|\veceta\|)^T \big| \partial_{\vecy_1}^{\vecbeta_1} \partial_{\vecy_2}^{\vecbeta_2}f_\phi(\vecy_1,\vecy_2,u,\veceta) \big| <\infty .
\end{multline}
\end{lem}

\begin{proof} 
This is analogous to the proof of Lemma \ref{decay101}.
%The proof is analogous to \cite[Lemma 4.3]{Marklof03}. First, if $\phi$ is bounded away from $0 \mod \pi$ then the estimate follows from the definition of $f_\phi$ and proceeding as in the proof of Lemma \ref{lem:schwa}. Also, if $f \in \tilde \scrS$, then so is the Fourier transform $$\hat f(\vecz_1,\vecz_2,u,\veceta) := \int_{\RR^{2d}} f(\vecy_1,\vecy_2,u,\veceta) \, \e(-\vecy_1\cdot\vecz_1+\vecy_2\cdot\vecz_2) \, \d \vecy_1 \d \vecy_2,$$ and by an identical argument one obtains the analogous bound with $\hat f$ in place of $f$, again for $\phi$ bounded away from zero. The result follows by noticing that $\hat f_{\phi} = f_{\phi+\pi/2}$, and hence one of these two bounds is always valid.
\end{proof}

Given $f\in\tilde \scrS$, we define the theta function
\begin{equation}\label{eq:Theta}
\Theta_f(g,u,\veceta)=\Theta_{f(\,\cdot\,,u,\veceta)}(g),
\end{equation}
with $\Theta_{f(\,\cdot\,,u,\veceta)}$ as defined in \eqref{def:theta} (with $u,\veceta$ fixed).
In view of Lemma \ref{decay101b}, we have $\Theta_f\in\C(\GamG\times\RR\times\RR^d)$. We further define
\begin{equation}\label{def:Fr}
F_r(g,u) = \int_{\RR^d} \Theta_f\bigg(g \bigg(1,\begin{pmatrix} \vecnull \\ \frac12 r^d \veceta \end{pmatrix}\bigg), u, \veceta \bigg) \d \veceta.
\end{equation}
\begin{prop} \label{prop:Frleadingorder} Let $f \in \tilde\scrS$. Then
\begin{multline} \label{eq:Frleadingorder}
F_r(u+\i v,\phi,\vecxi ,u') 
=v^{d/2} \sum_{\vecm \in \ZZ^d} \int_{\RR^d} f_\phi(v^{1/2}(\vecm-\vecy),v^{1/2}(\vecm-\vecy),u',\veceta) \, \d \veceta \\
+ O(r^d) + O(v^{-\infty}).
\end{multline}
uniformly for all $(u+\i v,\phi,\vecxi) \in G$, $u' \in \RR$, with $v > 1/2$ and $r <1$. 
\end{prop}

\begin{proof}
Note that 
\begin{align}
\bigg( u+i v,\phi, \vecxi \bigg) \bigg(1,\begin{pmatrix} \vecnull \\ \frac12 r^d \veceta \end{pmatrix}\bigg) = \Bigg(u+iv,\phi,  
\begin{pmatrix}\vecx + \vecx_{\tau,\phi,\veceta}\\ \vecy + \vecy_{\tau,\phi,\veceta} 
\end{pmatrix} \bigg)
\end{align}
where
\begin{equation}
\begin{split}
\vecx_{\tau,\phi,\veceta}&= -\tfrac12 v^{1/2} r^d \veceta \sin \phi + \tfrac12 u v^{-1/2} r^d \veceta \cos \phi \\ 
\vecy_{\tau,\phi,\veceta} &= \tfrac12 v^{-1/2} r^d \veceta \cos \phi .
\end{split}
\end{equation}
We thus have
\begin{multline}
F_r(u+\i v,\phi,\vecxi,u') \\
= \int_{\RR^d} v^{d/2} \sum_{\vecm_1,\vecm_2 \in \ZZ^d} f_\phi(v^{1/2}(\vecm_1 -\vecy -\vecy_{\tau,\phi,\veceta}) , v^{1/2}(\vecm_2 -\vecy-\vecy_{\tau,\phi,\veceta} ),u',\veceta)  \\
\times \e(\tfrac12 u ( \|\vecm_1 - \vecy - \vecy_{\tau,\phi,\veceta} \|^2-\|\vecm_2-\vecy - \vecy_{\tau,\phi,\veceta}\|^2) ) \\ \times \e((\vecx+\vecx_{\tau,\phi,\veceta})\cdot(\vecm_1-\vecm_2)) \, \d \veceta.
\end{multline}
Choose $\vecm\in\ZZ^d$ such that $\vecm \in [-\tfrac12,\tfrac12)^d+\vecy+\vecy_{\tau,\phi,\veceta}$. Then, for any $T\geq 1$ and for all $\vecm_1 \neq \vecm$,
\begin{multline}
f_\phi\big(v^{1/2}(\vecm_1 - \vecy - \vecy_{\tau,\phi,\veceta}) , v^{1/2}(\vecm_2 - \vecy-\vecy_{\tau,\phi,\veceta} ),u',\veceta\big) \\ = O_T\big(v^{-T} (1+\|\vecm_1\|^{-2T}) (1+\|\vecm_2\|^{-2T}) (1+\|\veceta\|^{-2T})\big).
\end{multline}
The same is true for $\vecm_2 \neq \vecm$. Therefore
\begin{multline}
F_r(u+\i v,\phi,\vecxi,u') \\
= v^{d/2} \sum_{\vecm\in\ZZ^d} \int_{\RR^d}  f_\phi(v^{1/2}(\vecm-\vecy-\vecy_{\tau,\phi,\veceta}),v^{1/2}(\vecm-\vecy-\vecy_{\tau,\phi,\veceta}),u',\veceta) \, \d \veceta \\ + O(v^{-\infty}).
\end{multline}
The result follows from applying Taylor's theorem and using Lemma \ref{decay101b} to conclude that the error term is small uniformly in $u'$ and $\phi$.
\end{proof} 

\begin{lem} \label{lem:Fruniformbound}
Fix $T>d$, then
\begin{enumerate}
\item 
The sequence $(F_r)_r$ of continuous functions $\GamG\times\RR\to\CC$ is uniformly dominated by $\Psi_{R,\overline f}$ where $\overline f(\vecy) = (1+\|\vecy\|)^{-2T}$.
\item $F_r\to F_0$ uniformly on compacta. 
\end{enumerate}
\end{lem}

\begin{proof}
 The set of $(u+\i v,\phi,\vecxi) \in G$ with $v > 1/2$ contains a fundamental domain of $\Gamma$ in $G$. Therefore, by Proposition \ref{prop:Frleadingorder} we have for $r<1$ that,
\begin{equation} \begin{split}
F_r(u+\i v,\phi,\vecxi ,u')  &\ll 1 + v^{d/2} \sum_{\vecm\in\ZZ^d} \int_{\RR^d} f_\phi( (\vecm-\vecy) \, v^{1/2}, (\vecm -\vecy) \, v^{1/2},u',\veceta) \, \d \veceta \\
& \ll 1 + v^{d/2} \sum_{\vecm\in\ZZ^d} \bar f( (\vecm-\vecy) v^{1/2} ) \int_{\RR^d}  (1+\|\veceta\|)^{-T} \, \d \veceta \\
& \ll 1+ \Psi_{R,\bar f}(\tau,\vecxi).
\end{split} \end{equation}
The first result is thus proved. The second result follows from the continuity of $\Theta_f$ and Lemma \ref{decay101b}.
\end{proof}

\begin{prop} \label{maincorollary2}
Let $f \in \tilde \scrS$, and assume $\yy$ is Diophantine of type $\kappa < (d-1)/(d-2)$ with the components of $(1,\trans\yy)$ linearly independent over $\QQ$. Let $w:\RR\to\RR$ piecewise continuous, continuous at 0, with compact support.  Then
\begin{multline}
\lim_{r \to 0} r^{d-2} \int_{\RR} F_r \left( \bigg( u+\i r^2,0, \begin{pmatrix} \vecnull \\ \yy \end{pmatrix} \bigg), r^{d-2} u \right) \, w(r^{d-2} u) \, \d u \\ 
\hspace{1cm} =2\,w(0)\,\int_{(\RR^d)^3} f(\vecy_1,\vecy_2,0,\veceta) \, \delta(\|\vecy_1\|^2-\|\vecy_2\|^2)\, \d \vecy_1\, \d \vecy_2\, \d \veceta \\
\hspace{5cm} + \int_{\RR^d\times\RR\times\RR^d} f(\vecy_1,\vecy_1,u,\veceta)\, w(u) \, \d \vecy_1\, \d u\, \d\veceta.
\end{multline}
\end{prop}

\begin{proof}
This is analogous to the proof of Proposition \ref{maincorollary}.
\end{proof}

\section{Order two}\label{sec:two}

In this section we show how the terms at order $\lambda^2$ can be written as averages over theta functions of the form \eqref{def:Fr}. We assume throughout this section that $\vecalf$ is Diophantine of type $\kappa < (d-1)/(d-2)$ with the components of $(1,\trans\vecalf)$ linearly independent over $\QQ$. 

\subsection{The cases \texorpdfstring{$\ell=2$}{l=2} and \texorpdfstring{$\ell=0$}{l=0}}

The cases $\ell =0$ and $2$ are similar and we treat them together. First, from \eqref{Isn} we have that $\scrI_{2,2}$ can be written
\begin{equation}
\begin{split}
&\scrI_{2,2}(s_1,s_2)
 = r^{2d} h^{d} \sum_{\vecm_1,\vecm_2\in\ZZ^d} \int_{\RR^d} | \hat W(r(\vecm_2-\vecm_1))|^2 
  \\
& \times  \e(\tfrac12 \, (s_2-s_1) (\|\vecm_2+\vecalf \|^2- \| \vecm_1+\vecalf \|^2) ) 
 \\
&\times \tilde a(-\veceta,h(\vecm_2+\vecalf +\tfrac12 r^{d-1}\veceta)) \tilde b(\veceta, h(\vecm_2+\vecalf +\tfrac12r^{d-1}\veceta)) \d\veceta + O(r^\infty),
\end{split}
\end{equation}
which we express as
\begin{equation}
\begin{split}
& \scrI_{2,2}(s_1,s_2) = r^{2d} h^{d} \sum_{\vecm_1,\vecm_2\in\ZZ^d} \int_{\RR^d} | \hat W(r(\vecm_2-\vecm_1))|^2 \\
&\times \e(-\tfrac12 (s_2-s_1) r^{d-1} (\vecm_2-\vecm_1)\cdot \veceta) ) 
 \\
& \times  \e(\tfrac12 \, (s_2-s_1) (\|\vecm_2+\vecalf + \tfrac12 r^{d-1}\veceta \|^2- \| \vecm_1+\vecalf + \tfrac12r^{d-1}\veceta\|^2) ) 
 \\
&\times \tilde a(-\veceta,h(\vecm_2+\vecalf +\tfrac12 r^{d-1}\veceta)) \tilde b(\veceta, h(\vecm_2+\vecalf +\tfrac12r^{d-1}\veceta)) \d\veceta + O(r^\infty) .
\end{split}
\end{equation}
In the same way we can see from \eqref{Is0} that $\scrI_{0,2}$ can be written 
\begin{equation}
\begin{split}
& \scrI_{0,2}(s_1,s_2) = r^{2d} h^{d} \sum_{\vecm_1,\vecm_2\in\ZZ^d} \int_{\RR^d}  \tilde a(-\veceta,h(\vecm_2+\vecalf +\tfrac12 r^{d-1}\veceta)) \\
&\times \e(-\tfrac12 \, s_{1} \|\vecm_2+\vecalf+r^{d-1}\veceta\|^2) \hat W(r(\vecm_2-\vecm_{1})) \\
& \times \e(-\tfrac12 \, (s_{2}-s_1) \|\vecm_1+\vecalf+r^{d-1}\veceta \|^2) \hat W(r(\vecm_1-\vecm_2))  \\
 &\times \e(\tfrac12 \, s_2 \|\vecm_2+\vecalf+r^{d-1}\veceta \|^2)
 \tilde b(\veceta, h(\vecm_2+\vecalf +\tfrac12r^{d-1}\veceta)) \d\veceta + O(r^\infty),
 \end{split}
\end{equation}
which we express as
\begin{equation}
\begin{split}
& \scrI_{0,2}(s_1,s_2) = r^{2d} h^{d} \sum_{\vecm_1,\vecm_2\in\ZZ^d} \int_{\RR^d} | \hat W(r(\vecm_2-\vecm_1))|^2 \\
& \times \e(\tfrac12 \, (s_2-s_1) r^{d-1} (\vecm_2-\vecm_1)\cdot \veceta ) 
 \\
&\times  \e(\tfrac12 \, (s_2-s_1) (\|\vecm_2+\vecalf + \tfrac12r^{d-1}\veceta\|^2- \| \vecm_1+\vecalf + \tfrac12r^{d-1}\veceta\|^2) ) 
 \\
& \times \tilde a(-\veceta,h(\vecm_2+\vecalf +\tfrac12 r^{d-1}\veceta)) \tilde b(\veceta, h(\vecm_2+\vecalf +\tfrac12r^{d-1}\veceta)) \d\veceta + O(r^\infty).
\end{split}
\end{equation}
We can then combine these two terms in the following way: First define $\scrI_{+,2}$ as
\begin{equation}
\begin{split}
& \scrI_{+,2}(s_1,s_2) = r^{2d} h^{d} \sum_{\vecm_1,\vecm_2\in\ZZ^d} \int_{\RR^d} | \hat W(r(\vecm_2-\vecm_1))|^2  \\
& \times \e(-\tfrac12 \, |s_2-s_1| r^{d-1} (\vecm_2-\vecm_1)\cdot \veceta ) 
 \\
& \times  \e(\tfrac12 \, (s_2-s_1) (\|\vecm_2+\vecalf + \tfrac12r^{d-1}\veceta\|^2- \| \vecm_1+\vecalf + \tfrac12r^{d-1}\veceta\|^2) ) 
 \\
& \times \tilde a(-\veceta,h(\vecm_2+\vecalf +\tfrac12 r^{d-1}\veceta)) \tilde b(\veceta, h(\vecm_2+\vecalf +\tfrac12r^{d-1}\veceta)) \d\veceta 
\end{split}
\end{equation}
and note that
\begin{equation}
\scrI_{+,2}(s_1,s_2)  =
\begin{cases}
\scrI_{2,2}(s_1,s_2) +O(r^\infty) & \text{if $s_1\leq s_2$}\\
\scrI_{0,2}(s_1,s_2) + O(r^\infty) & \text{if $s_1\geq s_2$.}
\end{cases}
\end{equation}
Therefore, after inserting the integration over $s_1$ and $s_2$ we obtain
\begin{multline}
\int_{0<s_1<s_2<h r^{1-d} t} \scrI_{2,2}(s_1,s_2) \d s_1\, \d s_2 +
\int_{0<s_2<s_1<h r^{1-d} t} \scrI_{0,2}(s_1,s_2) \d s_1\, \d s_2 \\
=
\int_0^{h r^{1-d} t} \int_0^{h r^{1-d} t} \scrI_{+,2}(s_1,s_2) \d s_1\, \d s_2 + O(r^\infty) .
\end{multline}
Note that we measure time in units of $h r^{1-d}$ as in the treatment of the zeroth order term.

\begin{lem}\label{lem:light}
Let $\scrI_{+,2}$ be defined as above and set $h=r$. Then,
\begin{multline}\label{light}
\int_0^{h r^{1-d} t} \int_0^{h r^{1-d} t} \scrI_{+,2}(s_1,s_2)\d s_1\, \d s_2 \\
= r^{d+2}  \int_{- r^{2-d} t}^{ r^{2-d} t}  F_r \left( \bigg( u+\i r^2,0, \begin{pmatrix} \vecnull \\ - \vecalf \end{pmatrix} \bigg), r^{d-2} u \right)  \, \d u ,
\end{multline}
with $F_r$ as defined in \eqref{def:Fr}, with the choice
\begin{multline} \label{littlef}
f(\vecy_1,\vecy_2,u,\veceta) =  \e(\tfrac12 (u+|u|) \, (\vecy_2-\vecy_1)\cdot \veceta) 
 \; (t-|u|) \chi_{[- t,t]}(u) \\ \times  |\hat W(\vecy_2-\vecy_1)|^2 \, \tilde a( \veceta, \vecy_2) \tilde b(-\veceta, \vecy_2) .
\end{multline} 
\end{lem}

\begin{proof}
In the case $h=r$ the left hand side of \eqref{light} reads  (after the variable substitution $\veceta\mapsto -\veceta$)
\begin{equation}
\begin{split}
&   r^{3d} \int_0^{ r^{2-d} t} \int_0^{ r^{2-d} t} \int_{\RR^d} \\
&\times \sum_{\vecm_1,\vecm_2\in\ZZ^d} | \hat W(r(\vecm_2-\vecm_1))|^2 
 \e(\tfrac12 \, |s_2-s_1| r^{d-1} (\vecm_2-\vecm_1)\cdot \veceta) 
 \\
& \times  \e(\tfrac12 \, (s_2-s_1)  (\|\vecm_2+\vecalf - \tfrac12r^{d-1}\veceta\|^2- \| \vecm_1+\vecalf - \tfrac12r^{d-1}\veceta\|^2) ) 
 \\
&\times \tilde a(\veceta,r(\vecm_2+\vecalf -\tfrac12 r^{d-1}\veceta)) \tilde b(-\veceta, r(\vecm_2+\vecalf -\tfrac12r^{d-1}\veceta)) \d\veceta  \, \d s_1 \d s_2 .
\end{split}
\end{equation}
We then use the relation
\begin{align}
\int_0^t  \int_0^t  f(s_2-s_1) \d s_1 \d s_2 =  \int_{-t}^t  (t-|u|) f(u) \, \d u
\end{align}
to re-write the above as 
\begin{equation}
\begin{split}
&  r^{2d+2}  \int_{- r^{2-d} t}^{ r^{2-d} t}  \int_{\RR^d} \sum_{\vecm_1,\vecm_2\in\ZZ^d} | \hat W(r(\vecm_2-\vecm_1))|^2 
 \e(\tfrac12 \, |u| r^{d-1} (\vecm_2-\vecm_1)\cdot \veceta ) 
 \\
& \times (t- r^{d-2}|u|) \; \e(\tfrac12 \, u (\|\vecm_1+\vecalf - \tfrac12r^{d-1}\veceta\|^2- \| \vecm_2+\vecalf - \tfrac12r^{d-1}\veceta\|^2) ) 
 \\
&\times \tilde a(\veceta, r(\vecm_2+\vecalf -\tfrac12 r^{d-1}\veceta )) \tilde b(-\veceta,  r(\vecm_2+\vecalf -\tfrac12r^{d-1}\veceta)) \d\veceta \,  \d u \\
& = r^{d+2}  \int_{- r^{2-d} t}^{ r^{2-d} t}\int_{\RR^d} \Theta_f\bigg( \bigg(u+\i r^2, 0, \begin{pmatrix}  \tfrac12 u r^{d-1}\veceta \\ -\vecalf + \tfrac12 r^{d-1}\veceta \end{pmatrix} \bigg), r^{d-2} u,\veceta \bigg) \, \d \veceta \, \d u,
\end{split}
\end{equation}
with $f$ as in \eqref{littlef}.
The result then follows from the fact that
\begin{equation}\label{calc001}
\left(u+\i r^2, 0, \begin{pmatrix}  \tfrac12 u r^{d-1}\veceta \\ -\vecalf + \tfrac12 r^{d-1}\veceta \end{pmatrix} \right)
= \left(u+\i r^2, 0, \begin{pmatrix}  \vecnull \\ -\vecalf  \end{pmatrix} \right)
 \left(\i, 0, \begin{pmatrix} \vecnull \\ \tfrac12 r^d\veceta \end{pmatrix} \right) .
\end{equation}
\end{proof}

Note that in view of \eqref{lambdaisrescaled} we should consider the rescaling of the coupling constant $\lambda\to \lambda h^{-2}$, or equivalently of the potential itself $W \to h^{-2} W$. At second order the potential appears as $|\hat W|^2$, and so we must rescale our terms by a factor of $h^{-4}$.
\begin{prop} Let $\scrI_{+,2}$ be defined as above. Then 
\begin{equation}
\begin{split}
\label{light1}
\lim_{ h=r\to 0} & \,  h^{-4}  \int_0^{h r^{1-d} t} \int_0^{h r^{1-d} t} \scrI_{+,2}(s_1,s_2)\d s_1\, \d s_2 \\
&= 2t  \int_{(\RR^d)^3}  | \hat{W}(\vecy_2-\vecy_1)|^2 \, a(\vecx,\vecy_2) b(\vecx,\vecy_2) \, \delta(\|\vecy_1\|^2-\|\vecy_2\|^2) \, \d \vecx \, \d \vecy_1 \, \d \vecy_2 \,  \\
& \hspace{4cm} + t^2  \, |\hat{W}(\vecnull)|^2 \, \int_{\RR^d\times\RR^d}  a(\vecx,\vecy) \, b(\vecx,\vecy) \, \d \vecx \, \d \vecy .
\end{split}
\end{equation}
\end{prop}

\begin{proof}
By Proposition \ref{maincorollary2} and Lemma \ref{lem:light} we have that the limit in \eqref{light1} is given by
\begin{multline}
 2 \, \, \int_{(\RR^d)^3} f(\vecy_1,\vecy_2,0,\veceta) \, \delta(\|\vecy_1\|^2-\|\vecy_2\|^2) \, \d \vecy_1\d \vecy_2 \, \d \veceta \, \\
+ \int_{- t}^{ t} \,\int_{\RR^d\times\RR^d} f(\vecy,\vecy,u,\veceta)  \d \vecy \, \d \veceta \, \d u .
\end{multline}
We have for the first term
\begin{multline}
2\int_{(\RR^d)^3} f(\vecy_1,\vecy_2,0,\veceta) \, \delta(\|\vecy_1\|^2-\|\vecy_2\|^2) \,  \d \vecy_1  \d \vecy_2 \, \d \veceta \,\\
= 2t\int_{(\RR^d)^3} | \hat{W}(\vecy_2-\vecy_1)|^2 \, \tilde a(\veceta,\vecy_2) \tilde b(-\veceta, \vecy_2) \, \delta(\|\vecy_1\|^2-\|\vecy_2\|^2) \, \d \vecy_1 \d \vecy_2 \, \d \veceta \,  \\
= 2t  \int_{(\RR^d)^3}  | \hat{W}(\vecy_2-\vecy_1)|^2 \, a(\vecx,\vecy_2) b(\vecx,\vecy_2) \, \delta(\|\vecy_1\|^2-\|\vecy_2\|^2) \, \d \vecx \, \d \vecy_1  \d \vecy_2
\end{multline}
Similarly for the second term we obtain
\begin{equation}
\begin{split}
&\int_{- t}^{ t} \int_{\RR^d\times\RR^d}  f(\vecy,\vecy,u,\veceta) \, \d \vecy \, \d \veceta \,  \,\d u\\
&= \int_{- t}^{ t} (t-|u|) \int_{\RR^d\times\RR^d}  |\hat{W}(\vecnull)|^2 \, \tilde a(\veceta,\vecy) \tilde b(-\veceta,\vecy)  \,  \d \vecy \, \d \veceta \,  \,\d u\\
&= t^2 \, |\hat{W}(\vecnull)|^2 \, \int_{\RR^d\times\RR^d}  a(\vecx,\vecy) \, b(\vecx,\vecy) \, \d \vecx \, \d \vecy.
\end{split}
\end{equation}
\end{proof}

\subsection{The case \texorpdfstring{$\ell=1$}{l=1}} 

\begin{lem}\label{lem:light2}
For $h= r$,
\begin{multline}\label{light22}
\int_0^{h r^{1-d} t} \int_0^{h r^{1-d} t} \scrI_{1,2}(s_1,s_2) \, \d s_1 \d s_2 \\
= r^{d+2}   \int_{- r^{2-d} t}^{ r^{2-d} t}  F_r \left( \bigg( u+\i r^2,0, \begin{pmatrix} \vecnull \\ -\vecalf \end{pmatrix} \bigg), r^{d-2} u \right) \,  \d u +O(r^\infty),
\end{multline}
with $F_r$ as defined in \eqref{def:Fr}, where
\begin{multline}\label{littlef2}
f(\vecy_1,\vecy_2,u,\veceta)  =  \frac12 \left(\int_{|u|}^{2t - |u|} \e( \tfrac12 (u-u') \veceta \cdot (\vecy_2-\vecy_1)) \d u' \right)   \chi_{[- t,t]}(u) \\
\times | \hat W( \vecy_1-\vecy_2)|^2 \, \tilde a(\veceta, \vecy_1) \, \tilde b(-\veceta,\vecy_2) .
\end{multline}
\end{lem}

\begin{proof}
As before, we start from Eq.~\eqref{Is}. For $\scrI_{1,2}$ this yields the explicit formula
\begin{equation}
\begin{split}
\scrI_{1,2}(s_1,s_2) & = r^{2d} h^{d} \sum_{\vecm_1,\vecm_2\in\ZZ^d} \int_{\RR^d} \e(-\tfrac12 \, s_1\|\vecm_2+\vecalf\|^2) \hat W(r(\vecm_2-\vecm_1)) \\
 & \times \e(\tfrac12 \, s_1 \|\vecm_1+\vecalf \|^2) \tilde a(-\veceta,h(\vecm_1+\vecalf +\tfrac12 r^{d-1}\veceta)) \\
& \times \e(-\tfrac12 \, s_2 \|\vecm_1+\vecalf+r^{d-1}\veceta\|^2) \hat W(r(\vecm_1-\vecm_2)) \\
 & \times \e(\tfrac12 \, s_2 \|\vecm_2+\vecalf+r^{d-1}\veceta \|^2)
 \tilde b(\veceta, h(\vecm_2+\vecalf +\tfrac12r^{d-1}\veceta)) d\veceta +O(r^\infty).
\end{split}
\end{equation}
We then note that we can write
\begin{multline}
s_1 \| \vecm_1 + \vecalf \|^2 - s_2 \|\vecm_1 + \vecalf + r^{d-1} \veceta\|^2 \\
= (s_1-s_2) \|\vecm_1+\vecalf + \tfrac12 r^{d-1} \veceta\|^2  - (s_1+s_2) r^{d-1} \veceta \cdot( \vecm_1+\vecalf) \\
 - \tfrac14 s_1 r^{2d-2} \|\veceta\|^2 - \tfrac34 s_2 r^{2d-2} \|\veceta\|^2.
\end{multline}
and similarly
\begin{multline}
-s_1 \| \vecm_2 + \vecalf \|^2 + s_2 \|\vecm_2 + \vecalf + r^{d-1} \veceta\|^2 \\
= (s_2-s_1) \|\vecm_2+\vecalf + \tfrac12 r^{d-1} \veceta\|^2  + (s_1+s_2) r^{d-1} \veceta \cdot( \vecm_2+\vecalf) \\
 + \tfrac14 s_1 r^{2d-2} \|\veceta\|^2 + \tfrac34 s_2 r^{2d-2} \|\veceta\|^2.
\end{multline}
We then insert these expressions into the exponential and make the variable substitutions $s_1-s_2 = u_1$, $s_1+s_2 = u_2$, and $\veceta\mapsto-\veceta$, to obtain
\begin{equation}
\begin{split}
& \tfrac12 r^{d+2} h^{d} \int_{-h r^{1-d} t}^{h r^{1-d} t}  \bigg( \int_{ r^{d-2} |u_1|}^{2 hr^{-1} t- r^{d-2}|u_1|} \sum_{\vecm_1,\vecm_2\in\ZZ^d} \int_{\RR^d} |\hat W(r(\vecm_2-\vecm_1))|^2 \\ 
& \times  \e( - \tfrac12 u_2  \,r\, \veceta\cdot(\vecm_2-\vecm_1)) \, \e(\tfrac12 \, u_1 (\|\vecm_1+\vecalf - \tfrac12 r^{d-1} \veceta \|^2-\|\vecm_2+\vecalf - \tfrac12 r^{d-1}\veceta \|^2)) \\
& \times \tilde a(\veceta,h(\vecm_1+\vecalf -\tfrac12 r^{d-1}\veceta)) \tilde b(-\veceta, h(\vecm_2+\vecalf -\tfrac12r^{d-1}\veceta)) \,  \d u_2 \bigg) \, \d u_1\\
& = r^{d+2}  \int_{- r^{2-d} t}^{ r^{2-d} t}\int_{\RR^d} \Theta_f\bigg( \bigg(u_1+\i r^2, 0, \begin{pmatrix}  \tfrac12 u_1 r^{d-1}\veceta \\ -\vecalf + \tfrac12 r^{d-1}\veceta \end{pmatrix} \bigg), r^{d-2} u_1,\veceta \bigg) \, \d \veceta \, \d u_1,
\end{split}
\end{equation}
with $f$ as in \eqref{littlef2}. The statement follows from \eqref{calc001}.
\end{proof}

\begin{prop}
\begin{multline}
\label{light2}
\lim_{ h=r\to 0} \,  h^{-4} \int_0^{h r^{1-d} t} \int_0^{h r^{1-d} t} \scrI_{1,2}(s_1,s_2)\, \d s_1\d s_2 \\
= 2 \int_0^{t}\int_{(\RR^d)^3}  | \hat{W}(\vecy_2-\vecy_1)|^2 \, \delta(\|\vecy_1\|^2-\|\vecy_2\|^2)  \\ 
\hspace{2cm}\times\,a(\vecx - s (\vecy_2-\vecy_1),\vecy_1)  \, b(\vecx,\vecy_2)  \, \d \vecy_1\, \d \vecy_2 \, \d \vecx \,  \d s \\
+  t^2  \, |\hat{W}(\vecnull)|^2 \, \int_{\RR^d\times\RR^d}  a(\vecx,\vecy) \, b(\vecx,\vecy) \, \d \vecx \, \d \vecy .
\end{multline}
\end{prop}

\begin{proof}
By Proposition \ref{maincorollary2} and Lemma \ref{lem:light2} we have that the limit in \eqref{light2} is the sum of two terms. The first one can be written
\begin{equation} \begin{split} 
&2\int_{(\RR^d)^3}  f(\vecy_1,\vecy_2,0,\veceta) \, \delta(\|\vecy_1\|^2-\|\vecy_2\|^2) \, \d \vecy_1 \, \d \vecy_2 \, \d \veceta \\
&=  2\int_{(\RR^d)^3}  | \hat{W}(\vecy_2-\vecy_1)|^2 \, \tilde a(\veceta,\vecy_1) \tilde b(-\veceta,\vecy_2) \, \delta(\|\vecy_1\|^2-\|\vecy_2\|^2) \, \\ 
&  \hspace{3cm} \times \left(\int_0^{t} \e( -u'\veceta \cdot (\vecy_2-\vecy_1)) \d u' \right) \,  \d \vecy_1 \, \d \vecy_2 \, \d \veceta \\
&=   2\int_0^{t}  \int_{(\RR^d)^3} | \hat{W}(\vecy_2-\vecy_1)|^2 \, \\ 
& \times\,  a(\vecx - s (\vecy_2-\vecy_1),\vecy_1) \, b(\vecx,\vecy_2)  \delta(\|\vecy_1\|^2-\|\vecy_2\|^2)  \, \d \vecy_1\, \d \vecy_2\,  \d \vecx \, \d s.
\end{split} \end{equation}
The second term takes the form
\begin{equation}
\begin{split}
&\frac12 \int_{- t}^{ t} \int_{\RR^d\times\RR^d}  f(\vecy,\vecy,u,\veceta) \, \d \vecy \, \d \veceta \,  \d u'   \d u \\
&= \, \int_{- t}^{ t} (t-|u|) \int_{\RR^d\times\RR^d}  |\hat{W}(\vecnull)|^2 \, \tilde a(\veceta,\vecy) \tilde b(-\veceta,\vecy) \, \d \vecy \, \d \veceta  \,  \d u \\
&= t^2 \, |\hat{W}(\vecnull)|^2 \, \int_{\RR^d\times\RR^d}  a(\vecx,\vecy) \, b(\vecx,\vecy) \, \d \vecx \, \d \vecy .
\end{split}
\end{equation}
\end{proof}

Thus, combining $\scrI_{j,2}$ for $j=0,1,2$ yields the following limiting expression for the second order terms.

\begin{cor} \label{cor:secondorderterms}
\begin{multline}
\label{light333}
\lim_{ h= r\to 0}\, h^{-4} \Big[ -\int_0^{hr^{1-d} t}\int_0^{s_2} \scrI_{2,2}(s_1,s_2) \, \d s_1 \d s_2  \\
\hspace{2cm}+ \int_0^{h r^{1-d} t} \int_0^{h r^{1-d} t} \scrI_{1,2}(s_1,s_2) \, \d s_1 \d s_2 \\
 \hspace{6cm} - \int_0^{hr^{1-d} t} \int_{s_2}^{hr^{1-d}t} \scrI_{0,2}(s_1,s_2) \, \d s_1 \d s_2   \Big]  \\
\hspace{-3cm}= 2\int_0^{t} \int_{(\RR^d)^3}  | \hat{W}(\vecy_2-\vecy_1)|^2 \delta(\|\vecy_1\|^2-\|\vecy_2\|^2) \\ 
 \times\,  \big[a(\vecx - s (\vecy_2-\vecy_1),\vecy_1) -a(\vecx,\vecy_2)\big]\, b(\vecx,\vecy_2) \, \d \vecy_1 \d \vecy_2  \d \vecx \, \d s .
\end{multline}
\end{cor}

Now replacing $a$ by the time-evolved symbol $L_0(t) a$ yields, in place of \eqref{light333},
\begin{multline}
2 \int_0^{t}\int_{(\RR^d)^3}  | \hat{W}(\vecy_2-\vecy_1)|^2 \, \delta(\|\vecy_1\|^2-\|\vecy_2\|^2) \\
\times \big[a(\vecx - (t-s) \vecy_1  - s \vecy_2,\vecy_1)  - a(\vecx-t \vecy_2,\vecy_2) \big]\, b(\vecx,\vecy_2) \, \d \vecx  \,  \d \vecy_1  \d \vecy_2 \,  \d s .
\end{multline}

\section{Higher-order theta functions}\label{sec:higher-order}

In order to prove bounds on the error terms \eqref{propeq2errorterms} in the Duhamel expansion we will need to define higher-order theta functions, that is generalisations of the theta function given in \eqref{def:theta} that live on the product space $(\Gamma \backslash G)^k$. Specifically, for $f\in\scrS(\RR^{d\times k}\times \RR^{d\times k})$, we denote by $\Theta_f^{(k)} : (\Gamma \backslash G)^{k} \to \CC$ the theta function
\begin{multline}
\Theta^{(k)}_f(\vectau,\vecphi,\vecXi) = \det(\vecv)^{d/2} \sum_{\vecM,\vecM'\in\ZZ^{d\times k}} f_{\vecphi}((\vecM-\vecY)\vecv^{1/2},(\vecM'-\vecY)\vecv^{1/2}) \\
 \times \e(\Tr[\tfrac12 \trans(\vecM-\vecY)(\vecM-\vecY) \vecu - \tfrac12 \trans(\vecM'-\vecY)(\vecM'-\vecY) \vecu + \trans(\vecM-\vecM')\vecX  ]) ,
\end{multline}
or more explicitly, 
\begin{multline}
\Theta^{(k)}_f(\vectau,\vecphi,\vecXi) = \sum_{\substack{\vecm_1,\ldots,\vecm_k\in\ZZ^{d}\\ \vecm_1',\ldots,\vecm_k'\in\ZZ^{d}}} \\ f_{\vecphi}(v_1^{1/2}(\vecm_1-\vecy_1),\ldots,v_k^{1/2}(\vecm_k-\vecy_k), 
v_1^{1/2}(\vecm_1'-\vecy_1),\ldots,v_k^{1/2}(\vecm_k'-\vecy_k)) \\
\times \prod_{j=1}^{k} v_j^{d/2} \e(\tfrac12 u_j (\|\vecm_j-\vecy_j\|^2-\|\vecm_j'-\vecy_j\|^2)+\vecx_j\cdot (\vecm_j-\vecm_j') ),
\end{multline}
where we use the natural notation
\begin{equation}
\begin{split}
\vectau &= \vecu+\i \vecv , \quad \vecu=\diag(u_1,\ldots,u_k), \quad u_j\in\RR, \\
\vecv & =\diag(v_1,\ldots,v_k), \quad v_j\in\RR_{>0} ,\quad
\vecphi = ( \phi_1 , \ldots , \phi_k) \in [0,2\pi)^k, \\
\vecXi &= ( \vecxi_1, \cdots, \vecxi_k) = \bigg( \begin{pmatrix}  \vecx_1 \\ \vecy_1 \end{pmatrix} ,\cdots,\begin{pmatrix}  \vecx_k \\ \vecy_k \end{pmatrix} \bigg) \in \RR^{2d\times k}, \\
\vecX &= (\vecx_1,\cdots,\vecx_k) \in \RR^{d\times k}, \qquad \vecY = (\vecy_1,\cdots,\vecy_k) \in \RR^{d\times k}, \\
 \vecM & = (\vecm_1,\cdots,\vecm_k) \in \ZZ^{d\times k} 
 \end{split}
 \end{equation}
and
\begin{align}
f_\vecphi(\vecY,\vecY') = \int_{\RR^{d\times k}\times\RR^{d\times k}} G_\vecphi(\vecY,\vecY',\vecZ,\vecZ') \, f(\vecZ,\vecZ') \, \d \vecZ \d \vecZ'
\end{align}
with
\begin{multline}
G_\vecphi(\vecY,\vecY',\vecZ,\vecZ') \\
= \prod_{j=1}^k |\sin \phi_j|^{-d} \e \left( \frac{\tfrac{1}{2}(\|\vecy_j\|^2+\|\vecz_j\|^2-\|\vecy_j'\|^2-\|\vecz_j'\|^2) \cos \phi_j - \vecy_j\cdot\vecz_j+\vecy_j'\cdot\vecz_j'}{\sin \phi_j}\right) .
\end{multline}
For $\phi = 0 \mod \pi$ we define $f_\phi$ by generalising \eqref{eightone2367} in the analogous way.
In the special case where $f = \prod_{j=1}^k f_j$ with $f_j \in \scrS(\RR^{d} \times \RR^d)$ the function $\Theta_f^{(k)}$ becomes the the product of $k$ independent theta functions of the form \eqref{def:theta}. In a similar vein as earlier, we wish to consider a generalisation of this theta function in which the function $f$ is allowed to depend directly on $\vecu\in\RR^k$ and some new parameters $\veceta \in \RR^d$ and $\omega \in \RR$. 

We denote by $\tilde \scrS_k$ the class of functions $f\in\C(\RR^{d\times k}\times\RR^{d\times k} \times \RR^k \times \RR^d \times \RR)$ with the property that for every multi-index $\vecbeta_1,\vecbeta_2\in \ZZ_{\geq 0 }^{d\times k}$ , the derivative 
$$\partial_{\vecY_1}^{\vecbeta_1} \, \partial_{\vecY_2}^{\vecbeta_2} \, f(\vecY_1,\vecY_2,u,\veceta,\omega)$$ 
(a) exists, (b) is continuous (in all variables), and (c) is rapidly decaying, i.e., 
\begin{multline}
\sup_{\vecY_1,\vecY_2,u,\veceta,\omega} (1+\|\vecY_1\|)^T (1+\|\vecY_2\|)^T (1+|u|)^T \\
(1+\|\veceta\|)^T(1+|\omega |)^T \big| \partial_{\vecY_1}^{\vecbeta_1} \, \partial_{\vecY_2}^{\vecbeta_2} \, f(\vecY_1,\vecY_2,u,\veceta,\omega) \big| <\infty
\end{multline}
for every $T> 1$.

We then consider the test function $f = f(\vecY,\vecY',\vecu,\veceta,\omega)$ in $\tilde\scrS_k$ and set
\begin{equation}\label{eq:Thetak}
\Theta_{f}^{(k)}(g,\vecu,\veceta,\omega):=\Theta_{f(\,\cdot\, , \vecu ,\veceta,\omega)}^{(k)}(g).
\end{equation} 

We now proceed to state some results in direct analogy with Section \ref{sec:theta}.

\begin{lem}[cf.~Lemma \ref{decay101b}]\label{decay102}
If $f\in\tilde\scrS_k$, then for all multi-indices $\vecbeta_1,\vecbeta_2 \in \ZZ_{\geq 0}^{d\times k}$ and every $T>1$ 
\begin{multline}
\sup_{\vecY_1,\vecY_2,u,\veceta,\omega, \vecphi} (1+\|\vecY_1\|)^T (1+\|\vecY_2\|)^T (1+|u|)^T \\
\times (1+\|\veceta\|)^T(1+|\omega |)^T \big|\partial_{\vecY_1}^{\vecbeta_1} \partial_{\vecY_2}^{\vecbeta_2} f_\vecphi(\vecY_1,\vecY_2,\vecu,\veceta,\omega) \big| <\infty .
\end{multline}
\end{lem}
\begin{proof}
The proof is analogous to those of Lemmas \ref{decay101} and \ref{decay101b}.
\end{proof}

Now, let us use the shorthand
\begin{align*}
z_k(\veceta) := \left( \bigg(1,\begin{pmatrix} \vecnull \\ \frac12 \veceta \end{pmatrix}\bigg), \cdots, \bigg(1,\begin{pmatrix} \vecnull \\ \frac12 \veceta \end{pmatrix}\bigg) \right) \in G^k,
\end{align*}
and further define
\begin{equation}\label{def:Fr111}
F_r^{k,\beta}(g,\vecu) := \int_\RR  \left| \int_{\RR^d} \Theta_f^{(k)}\bigg(g \, z_k(r^d \veceta), \vecu, \veceta,\omega \bigg) \d \veceta \right|^\beta  \d \omega.
\end{equation}

\begin{prop}[cf.~Proposition  \ref{prop:Frleadingorder}] \label{prop:thetakleadingorder}Let $0 < \beta <1$ and $f \in \tilde{\scrS}_k$. Then,
\begin{multline}
F_r^{k,\beta}(\vecu+\i\vecv,\vecphi,\vecXi,\vecu')\\[-0.1cm]
 = \det(\vecv)^{\beta d/2} \sum_{\vecM \in \ZZ^{d\times k}}  \int_\RR \left| \int_{\RR^d} f_\vecphi( (\vecM-\vecY) \vecv^{1/2},(\vecM-\vecY) \vecv^{1/2},\vecu',\veceta,\omega)\d \veceta \right|^\beta  \d \omega \\
+  O (r^{ d}) +  \sum_{j=1}^{k} O(v_j^{-\infty})
\end{multline}
uniformly for all $(\vecu+\i\vecv,\vecphi,\vecXi) \in (\Gamma \backslash G)^{k}$, $\vecu' \in \RR^k$ with $v_j > 1/2$ for all $j$ and $r<1$.
\end{prop}
\begin{proof}
The proof is analogous to that of Proposition \ref{prop:Frleadingorder}.
\end{proof}

Recall the definitions of $\Psi_{R,f}^{\beta}$ and $\bar f$ in \eqref{def:PsiRf} and \eqref{def:fbar}.

\begin{lem} \label{lem:Frkuniformbound}
Fix $T > d$, then
\begin{enumerate}
\item 
There is a constant $C$, such that for all $r < 1$
\begin{align}
|F_r^{k,\beta}(\vectau,\vecphi,\vecXi,\vecu')| < C \prod_{j=1}^{k}(1+\Psi_{1/2,\bar f}^{\beta}(\tau_j,\vecxi_j)).
\end{align}
\item $F_r^{k, \beta} \to F_0^{k,\beta}$ uniformly on compacta. 
\end{enumerate}
\end{lem}
\begin{proof}
The proof is analogous to that of Lemma \ref{lem:Fruniformbound}  with Lemma \ref{decay102} in place of Lemma \ref{decay101b}.
\end{proof}

 In the following, we denote by $I_k$ the $k \times k$ identity matrix.

\begin{prop} \label{maincorollary22}
Let $0 < \beta <1$ and $f \in \tilde{\scrS}_k$. Assume for $j=1,\dots, k$ that $\vecy_j$ is Diophantine of type $\kappa < (d-1)/(d-2)$ with the components of $(1,\trans\vecy_j)$ linearly independent over $\QQ$. Let $w:\RR^k\to\RR$ be bounded with compact support. Then,  
\begin{equation}
 \limsup_{r \to 0}  r^{k(d-2)} \int_{\RR^k} F_r^{k,\beta}\left(\vecu+ \i r^2 I_k, \vecnull, \left(\begin{matrix} \vecnull \\ \vecY \end{matrix}\right) , r^{d-2}\vecu \right) \, w(r^{d-2} \vecu)  \, \d \vecu < \infty .
\end{equation}
\end{prop}
\begin{proof}
Applying Lemma \ref{lem:Frkuniformbound} yields
\begin{multline}
\limsup_{r \to 0}  \, r^{k(d-2)} \int_{\RR^k} F_r^{k,\beta}\left(\vecu+ \i r^2 I_k, \vecnull, \left(\begin{matrix} \vecnull \\ \vecY \end{matrix}\right) ,  r^{d-2}\vecu \right) \, w(r^{d-2} \vecu)  \, \d \vecu \\
< C \limsup_{r \to 0} r^{k(d-2)} \int_{\RR^k} \left[\prod_{j=1}^{k} \left(1+\Psi_{1/2,\bar f}^{\beta}\left(u_j + \i r^2 I_k,\left(\begin{matrix} \vecnull \\ \vecy_j \end{matrix}\right) \right) \right) \right]w(r^{d-2} \vecu)  \, \d \vecu.
\end{multline}
The function $w$ has compact support, so fix $L$ such that the cube $(-L,L)^k$ contains the support of $w$, and denote by $\chi_L$ the characteristic function of the interval $(-L,L)$. We can then bound the above expression by
\begin{align}
C \sup |w| \limsup_{r\to 0} \prod_{j=1}^{k} \left( r^{d-2} \, \int_{\RR}\left(1+\Psi_{1/2,\bar f}^{\beta}\left(u_j + \i r^2 I_k,\left(\begin{matrix} \vecnull \\ \vecy_j \end{matrix}\right) \right) \right)  \chi_L( r^{d-2} u_j) \d u_j \right).
\end{align}
The result then follows by applying Proposition \ref{diophprop234}.
\end{proof}

\section{Error terms} \label{sec:error}

In this section we prove upper bounds on the error terms \eqref{propeq2errorterms} in the semiclassical Boltzmann-Grad scaling, i.e., for $Q_n(h r^{1-d}t, D_{r,h}a, D_{r,h}b)$, where relevant cases are $n=3,4,5,6$.
Lemma \ref{lem:errortermsCS} tells us that  
\begin{equation}
| Q_n(h r^{1-d}t, D_{r,h}a, D_{r,h}b)| \leq \sum_{\ell = n-3}^3 \scrJ_{\ell,n}(h r^{1-d}t, D_{r,h}a) \; \|  \Pi_\vecalf \Op_{r,h}(b) \|_{\HiS,\vecalf}. 
\end{equation}
The term $\|  \Pi_\vecalf \Op_{r,h}(b) \|_{\HiS,\vecalf}$ has a uniform upper bound; cf.~Lemma \ref{scrI0}.
Hence the key is to estimate (recall Def.~\eqref{errortermthetafunc0}  and Lemma \ref{lem:projectionHSnorm})
\begin{multline} 
\scrJ_{\ell,n}(h r^{1-d}t, D_{r,h}a) =  (2\pi)^n\int_{\substack{0<s_1<\ldots<s_{\ell}<t h r^{1-d} \\ 0<s_{n}<\ldots<s_{\ell+1}<t h r^{1-d}}} \\ 
 \times \big(\Tr_\vecalf \big[ K_{1,\ell}(\vecs)^\dagger K_{1,\ell}(\vecs) \Op_{r,h}(a)  K_{\ell+1,n}(\vecs)K_{\ell+1,n}(\vecs)^\dagger  \Op_{r,h}(\overline{a}) \big] \big)^{1/2}  \d \vecs .
\end{multline}
A straightforward computation (see Appendix \ref{appendix2}) yields the expression
\begin{equation}
\begin{split} \label{eq:compli}
\Tr_\vecalf &[K_{1,\ell}(\vecs)^\dagger K_{1,\ell}(\vecs) \Op_{r,h}(a)K_{\ell+1,n}(\vecs) K_{\ell+1,n}(\vecs)^\dagger \Op_{r,h}(\bar a)] \\
&= r^{2nd} h^{d}  \sum_{\substack{\vecm_0,\vecm_1,\ldots,\vecm_n\in\ZZ^d\\ \vecm_1',\ldots,\vecm_n'\in\ZZ^d}} \mathbb{1}[ \vecm_n' - \vecm_n + \vecm_\ell -\vecm_\ell'= 0]  \\
& \times \int_{\RR^d} \hat W(r(\vecm_0-\vecm_1)) \scrT_{1,\ell-1}^-(\vecalf)   \e(\tfrac12 \, s_\ell (\|\vecm_\ell+\vecalf \|^2-\|\vecm_\ell'+\vecalf \|^2))     \\
& \times \hat W(r(\vecm_1'-\vecm_0)) \overline{\scrT}_{1,\ell-1}^-(\vecalf)  \tilde a(\veceta,h(\vecm_\ell +\vecalf - \tfrac12 r^{d-1} \veceta)) \\
& \times \hat W(r(\vecm_\ell-\vecm_{\ell+1})) \scrT_{\ell+1,n-1}^-(\vecalf - r^{d-1} \veceta  )    \\
& \times \e(\tfrac12 \, s_{\ell+1}(\|  \vecm_\ell' + \vecalf -r^{d-1} \veceta \|^2- \| \vecm_\ell + \vecalf - r^{d-1} \veceta \|^2))  \\
& \times \hat W(r(\vecm_{\ell+1}'-\vecm_\ell')) \overline{\scrT}_{\ell+1,n-1}^-( \vecalf -r^{d-1} \veceta )    \tilde {\overline{a}}(-\veceta,h(  \vecm_\ell'+ \vecalf - \tfrac12 r^{d-1} \veceta )) \, \d \veceta \\
& +O(r^\infty),
\end{split}
\end{equation} 
where 
\begin{equation} \begin{split}
\scrT^-_{\ell,n}(\vecy) 
&=\begin{cases} \prod_{j=\ell}^{n} \e(\tfrac12 \, (s_j-s_{j+1}) \|\vecy+\vecm_j \|^2) \hat W(r(\vecm_j-\vecm_{j+1})) & (l\leq n) \\ 1 & (l>n), \end{cases} \\
\overline{\scrT}^-_{\ell,n}(\vecy) 
&= \begin{cases} \prod_{j=\ell}^{n} \e(\tfrac12 \, (s_{j+1}-s_j) \|\vecy+\vecm_j' \|^2) \hat W(r(\vecm_{j+1}'-\vecm_j')) & (l\leq n) \\ 1 & (l>n). \end{cases}
\end{split} \end{equation} 
Let us focus on the exponential factors in \eqref{eq:compli}; they are 
\begin{equation}
\begin{split}
& \bigg( \prod_{j=1}^{\ell-1} \e(\tfrac12 \, (s_j-s_{j+1})( \|\vecm_j+\vecalf \|^2-\|\vecm_j'+\vecalf\|^2) ) \bigg) \\
& \times \e(\tfrac12 \, s_\ell ( \|\vecm_\ell+\vecalf \|^2 - \|\vecm_\ell'+\vecalf\|^2))  \\
& \times \e(\tfrac12 \, s_{\ell+1}(\|\vecm_\ell'+\vecalf -r^{d-1}\veceta\|^2-\|\vecm_\ell+\vecalf -r^{d-1}\veceta\|^2)) \\
& \times \bigg( \prod_{j=\ell+1}^{n-1} \e(\tfrac12 \, (s_j-s_{j+1}) (\|\vecm_j+\vecalf  - r^{d-1}\veceta\|^2-\|\vecm_j'+\vecalf  - r^{d-1}\veceta\|^2)) \bigg).
\end{split}
\end{equation}
We write the above as
\begin{equation} \label{errortermexponentials}
\begin{split}
& \bigg( \prod_{j=1}^{\ell-1} \e(\tfrac12 \, (s_j-s_{j+1}) (\|\vecm_j+\vecalf -\tfrac12 r^{d-1}\veceta\|^2 \\[-0.5cm]
& \hspace{4cm}-\|\vecm_j'+\vecalf -\tfrac12 r^{d-1}\veceta \|^2 + r^{d-1} (\vecm_j-\vecm_j')\cdot\veceta)) \bigg) \\
& \times \e(\tfrac12 \, (s_\ell-s_{\ell+1}) ( \|\vecm_\ell+\vecalf- \tfrac12 r^{d-1} \veceta \|^2 - \|\vecm_\ell'+\vecalf -\tfrac12 r^{d-1} \veceta\|^2 )) \\
& \times \e(\tfrac12 \,(s_\ell + s_{\ell+1})r^{d-1}\veceta \cdot (\vecm_\ell-\vecm_\ell')) \\
& \times \bigg( \prod_{j=\ell+1}^{n-1} \e(\tfrac12 \, (s_j-s_{j+1}) (\|\vecm_j+\vecalf -\tfrac12 r^{d-1}\veceta\|^2 \\[-0.5cm]
& \hspace{4cm} -\|\vecm_j'+\vecalf -\tfrac12 r^{d-1}\veceta\|^2 - r^{d-1} (\vecm_j-\vecm_j')\cdot\veceta) \bigg).
\end{split}
\end{equation}
%where we have used the identities
%\begin{equation} \begin{split}
%\| \vecm_j + \vecalf \|^2 & = \|\vecm_j + \vecalf + \tfrac 12 r^{d-1} \veceta\|^2 - r^{d-1} \veceta \cdot (\vecm_j+\vecalf) - \tfrac14 r^{2d-2} \|\veceta\|^2, \\
%\| \vecm_j + \vecalf + r^{d-1} \veceta \|^2 & = \|\vecm_j + \vecalf + \tfrac 12 r^{d-1} \veceta\|^2 + r^{d-1} \veceta \cdot (\vecm_j+\vecalf) + \tfrac34 r^{2d-2}\|\veceta\|^2.
%\end{split} \end{equation}
Note that this product of exponentials is independent of the variables $\vecm_0$, $\vecm_n$ and $\vecm_n'$ - and so the entire dependence on these variables is in the product of $\hat W$ terms. In \eqref{eq:compli} we can therefore separately evaluate the threefold sum
\begin{multline}
\sum_{\substack{\vecm_0,\vecm_n,\vecm_n' \\ \vecm_n' - \vecm_n + \vecm_\ell -\vecm_\ell'= 0}} 
\hat W(r(\vecm_0-\vecm_1))  \hat W(r(\vecm_1'-\vecm_0)) \\
\times \hat W(r(\vecm_{n-1} - \vecm_n)) \, \hat W(r(\vecm_n'-\vecm_{n-1}')) 
\end{multline}
which is equal to
\begin{multline}
\sum_{\vecm_0,\vecm_n}\hat W(r(\vecm_0-\vecm_1)) \hat W(r(\vecm_1'-\vecm_0))  \hat W(r(\vecm_{n-1} - \vecm_n)) \\
 \times\hat W(r(\vecm_n+\vecm_\ell'-\vecm_\ell-\vecm_{n-1}')).
\end{multline}
Applying the Poisson summation formula to the sums over $\vecm_0$ and $\vecm_n$ yields
\begin{multline} \label{errortermpoissum0}
r^{-2d} \sum_{\veck_0,\veck_n} \iint_{\RR^{2d}}  \, \hat W(\vecy_0-r\vecm_1) \hat W(r\vecm_1' - \vecy_0) \hat W(r\vecm_{n-1} - \vecy_n) \\
\times \, \hat W(\vecy_n+r(\vecm_\ell'-\vecm_\ell-\vecm_{n-1}'))  \e(r^{-1} \veck_0 \cdot\vecy_0+r^{-1} \veck_n \cdot\vecy_n)\, \d \vecy_0 \d \vecy_n .
\end{multline}
Since $W\in\scrS(\RR^d)$, we have for any $T_1,T_2\geq 1$ that \eqref{errortermpoissum0} equals
\begin{multline} \label{errortermpoissum01}
r^{-2d} \scrW(r(\vecm_1'-\vecm_1))\scrW(r(\vecm_\ell'-\vecm_\ell+\vecm_{n-1}-\vecm_{n-1}')) \\ 
+  O_T\big( r^{T_1} (1+ r \| \vecm_1-\vecm_1' \|)^{-T_2}  (1+ r \| \vecm_\ell'-\vecm_\ell+\vecm_{n-1}-\vecm_{n-1}' \|)^{-T_2} \big),
\end{multline}
with
\begin{equation}
\scrW(\vect) = \int_{\RR^{d}}  \hat W(\vect- \vecy)  \hat W(\vecy) \, \d \vecy .
\end{equation}
The error term in \eqref{errortermpoissum01}, after applying the remaining $\vecm_j$-sums, yields therefore a total contribution of order $O(r^\infty)$ for $h=r\in(0,1]$. In order to write \eqref{eq:compli} as a higher order theta function, we change variable by the linear map $A:\RR^n\to\RR^n$, $\vecs\mapsto \vecomega=A\vecs$, given by
\begin{equation}
\omega_j = s_j-s_{j+1} \quad (j=1,\ldots,n-1) , \qquad \omega_n= s_\ell+s_{\ell+1} .
\end{equation}
The corresponding determinant equals 2, and hence $A$ is invertible. Let 
\begin{equation}
\scrQ=\{ \vecs\in\RR^n \mid 0<s_1<\ldots<s_{\ell}<1, \; 0<s_{n}<\ldots<s_{\ell+1}<1\} .
\end{equation}
Then, for $h=r$ and $\vecomega=(\vecu,\omega)\in\RR^{n-1}\times\RR$,
\begin{equation} \label{errortermthetafunc}
\begin{split}
&\scrJ_{\ell,n}(h r^{1-d}t, D_{r,h}a) \\
&=  (2\pi)^n \,  r^{nd/2} \int_{\RR^{n-1}}  \int_{\RR} \mathbb{1}(r^{d-2} (\vecu,\omega)\in A\scrQ) \\
& \times \left| \int_{\RR^d} \Theta_{f_*}^{(n-1)}\left( g_r(\vecu,\vecalf) \, z_{n-1}(r^d \veceta) ,r^{d-2} \vecu,\veceta, r^{d-2} \omega  \right) \d \veceta \right|^{1/2}    \,  \d \omega \, \d \vecu + O(r^\infty)
\end{split}
\end{equation}
with 
\begin{align*}
g_r(\vecu,\vecalf) &= \left( \vecu + \i r^2 I_k, \vecnull, \left( \begin{pmatrix} 0 \\ -\vecalf \end{pmatrix}, \cdots, \begin{pmatrix} 0 \\ -\vecalf \end{pmatrix} \right) \right) \in G^{n-1} ,
\end{align*}
 and $\Theta_{f_*}^{(k)}$ as in \eqref{eq:Thetak} with $k=n-1$ and test function
\begin{equation} \label{errorboundtestf}
 \begin{split}
f_*(\vecY,\vecY', \vecu,\veceta, \omega) &=  \scrW(\vecy_1'-\vecy_1)\scrW(\vecy_\ell'-\vecy_\ell+\vecy_{n-1}-\vecy_{n-1}') \\
& \times \left(\prod_{j=1}^{n-2} \hat W(\vecy_j-\vecy_{j+1}) \hat W(\vecy_{j+1}'-\vecy_{j}')\right) \,\tilde a (\veceta, \vecy_\ell) \, \bar{\tilde{ a}} (-\veceta, \vecy_\ell')   \\
& \times   \bigg( \prod_{j=\ell+1}^{n-1} \e( - u_j  (\vecy_j-\vecy_j')\cdot\veceta)) \bigg)  \e(\tfrac12(\omega -  u_\ell) \veceta \cdot(\vecy_\ell-\vecy_\ell'))
\end{split} \end{equation}
 where $\vecY,\vecY' \in \RR^{d\times(n-1)}$ are given by
\begin{equation}
\vecY = (\vecy_1, \cdots,\vecy_{n-1}), \quad \vecY' = (\vecy_1', \cdots,\vecy_{n-1}').
\end{equation}
In order to apply the results in Section \ref{sec:higher-order}, we however require $f_*$ to be continuous and compactly supported in $\vecu$, and rapidly decaying in $\omega$. To achieve this, note that we can find $f$ with precisely these properties by setting 
\begin{equation}
f(\vecY,\vecY',\vecu,\veceta, \omega) = (\iota(\vecu,\omega))^2 f_*(\vecY,\vecY',\vecu,\veceta, \omega)
\end{equation}
with $\iota : \RR^n \to \RR_{\geq 0}$ smooth and compactly supported such that $\iota(\vecu,\omega) > (2\pi)^n $ on the domain of integration. We then have, instead of \eqref{errortermthetafunc},
\begin{equation} \label{errortermthetafunc0001}
\begin{split}
&\scrJ_{\ell,n}(h r^{1-d}t, D_{r,h}a) \\
&\leq  r^{nd/2} \int_{\RR^{n-1}}  \int_{\RR} \mathbb{1}(r^{d-2} (\vecu,\omega)\in A\scrQ) \\
& \times \left| \int_{\RR^d} \Theta_f^{(n-1)}\left( g_r(\vecu,\vecalf) \, z_{n-1}(r^d \veceta) ,r^{d-2} \vecu,\veceta, r^{d-2} \omega  \right) \d \veceta \right|^{1/2}    \,  \d \omega \, \d \vecu  + O(r^\infty) ,
\end{split}
\end{equation}
and thus after the variable substitution $\omega \mapsto r^{2-d} \omega$,
\begin{equation} \label{errortermthetafunc0002}
\begin{split}
&\scrJ_{\ell,n}(h r^{1-d}t, D_{r,h}a) \\
&\leq  r^{nd/2} r^{2-d} \int_{\RR^{n-1}}  w(r^{d-2} \vecu) \\
& \times \int_{\RR} \left| \int_{\RR^d} \Theta_f^{(n-1)}\left( g_r(\vecu,\vecalf) \, z_{n-1}(r^d \veceta),r^{d-2} \vecu,\veceta, \omega  \right) \d \veceta \right|^{1/2}    \,  \d \omega\, \d \vecu + O(r^\infty) ,
\end{split}
\end{equation}
with
\begin{equation}
w(\vecu) = \sup_{\omega\in\RR} \mathbb{1}((\vecu,\omega)\in A\scrQ) ,
\end{equation}
which is bounded and has compact support.

\begin{lem} \label{lem:scrJbound}
Under the assumptions of Theorem \ref{thm:dioalpha}, for $h=r <1$,
\begin{equation}
\scrJ_{\ell,n}(h r^{1-d}t, D_{r,h}a )= O(r^{-nd/2+2n}).
\end{equation}
\end{lem}
\begin{proof}
For $F_r^{k,\beta}$ as in \eqref{def:Fr111}, we have
\begin{multline} \label{scrJboundstep1}
\int_{\RR} \left| \int_{\RR^d} \Theta_f^{(n-1)}\left( g \, z_{n-1}(r^d \veceta),r^{d-2} \vecu,\veceta, \omega  \right) \d \veceta \right|^{1/2}  \d \omega =  F_r^{n-1,1/2}\left( g ,r^{d-2} \vecu  \right) .
\end{multline}
Thus, applying Proposition \ref{maincorollary22} we see that the right hand side of \eqref{errortermthetafunc0002} is bounded above by a constant times
\begin{equation} 
  r^{nd/2} \times r^{2-d} \times r^{-(n-1)(d-2)}=r^{-nd/2+ 2n}.
\end{equation}

\end{proof}

\begin{proof}[Proof of Theorem \ref{thm:dioalpha}]
We recall the rescaling of $t$ and $\lambda$ in eq.~\eqref{lambdaisrescaled}.
The existence of the operators $A_n^{(r,\vecalf)}(t r^{1-d})$ follows from the Duhamel expansion in Equation \eqref{propeq2}. The error term follows from Lemma \ref{lem:errortermsCS} and Lemma \ref{lem:scrJbound}, remembering that $\lambda$ should be rescaled $\lambda \mapsto \lambda/h^2$ as in \eqref{lambdaisrescaled}. Finally, the convergence of the operators $A_n^{(r,\vecalf)}(t r^{1-d})$ in the limit $r \to 0$ is proved by combining Lemma  \ref{scrI0}, Lemma \ref{scrI1} and Corollary \ref{cor:secondorderterms}.
\end{proof} 

\section{Averages over \texorpdfstring{$\vecalf$}{alpha}}\label{sec:averages}

In this section we give the analogous results required to prove Theorem \ref{thm:allalpha}. First recall that Proposition \ref{equim} tells us that for $\yy \in \RR^d \backslash \QQ^d$ with the components of $(1,\trans\yy)$ linearly independent, and $(F_r)_{r\geq 0}$ a sequence of uniformly bounded, continuous functions we have
\begin{multline} \label{recallequim}
 \lim_{r \to 0}  r^\sigma \int_{\RR} F_r((u+\i r^2,0,(\begin{smallmatrix} \vecnull \\ \yy \end{smallmatrix})), r^\sigma u)\, w(r^\sigma u)\, \d u \\
 =  \frac{1}{\mu(\GamG)} \int_{\Gamma \backslash G} \int_{\RR} F(g,u) \, w(u) \,\d u \,\d \mu(g).
\end{multline}
Note that since the $F_r$ are uniformly bounded and continuous, and $w \in \L^1(\RR)$, the integral over $u$ is bounded uniformly in $r$ and $\yy$. Since the statement \eqref{recallequim} holds for a full measure set of $\yy \in [0,1)^d$, one can apply dominated convergence to conclude
\begin{multline}
 \lim_{r \to 0}  r^\sigma \int_{[0,1)^d} \int_{\RR} F_r((u+\i r^2,0,(\begin{smallmatrix} \vecnull \\ \yy \end{smallmatrix} )), r^\sigma u)\, w(r^\sigma u)\, \d u \, \d \yy \\
 =  \frac{1}{\mu(\GamG)} \int_{\Gamma \backslash G} \int_{\RR} F(g,u) \, w(u) \,\d u \,\d \mu(g).
\end{multline}
Thus we now just need to consider the case of unbounded test functions. It follows from \eqref{expliPsi} that
\begin{multline}
\int_{\yy\in[0,1)^d}  \Psi_{R,f}^\beta(\tau,(\begin{smallmatrix} \vecnull \\ \yy \end{smallmatrix})) \d \yy
= 2 \sum_{\vecm \in \ZZ^d} f\left( \vecm \frac{v^{1/2}}{|\tau|} \right) \frac{v^{\beta d/2}}{|\tau|^{\beta d}} \chi_R \left( \frac{v}{|\tau|^2} \right) \\
+ 2 \int_{\RR^d} f(\vecy)\d \vecy \sum_{\substack{(c,d) \in \ZZ^2 \\ gcd (c,d) = 1 \\ c>0, d\neq 0}}  \frac{v^{(\beta-1) d/2} }{|c\tau+d|^{(\beta-1) d}} \chi_R \left( \frac{v}{|c\tau+d|^2} \right).
\end{multline}
Since for $0\leq\beta\leq 1$
\begin{equation}
\frac{v^{(\beta-1) d/2} }{|c\tau+d|^{(\beta-1) d}} \leq R^{(\beta-1)d/2} ,
\end{equation}
we have
\begin{multline} \label{eq:apple}
\int_{\yy\in[0,1)^d}  \Psi_{R,f}^\beta(\tau,(\begin{smallmatrix} \vecnull \\ \yy \end{smallmatrix})) \d \yy
\leq 2 \sum_{\vecm \in \ZZ^d} f\left( \vecm \frac{v^{1/2}}{|\tau|} \right) \frac{v^{\beta d/2}}{|\tau|^{\beta d}} \chi_R \left( \frac{v}{|\tau|^2} \right) \\
+ R^{(\beta-1)d/2} X_R(\tau) \int_{\RR^d} f(\vecy)\d \vecy .
\end{multline}
This allows us to prove the following $\yy$-averaged version of Propositions \ref{diophprop} and \ref{diophprop234}.
\begin{prop} \label{avephprop}
Let $w:\RR\to\RR$ piecewise continuous with compact support, and $0<\epsilon<1$. Then, for every $R\geq 1$,
\begin{equation}
\limsup_{r \to 0} r^{d-2} \int_{|u|>r^{2-\epsilon}} \int_{[0,1)^d}  \Psi_{R,f}(u+\i r^2,(\begin{smallmatrix} \vecnull \\ \yy \end{smallmatrix}))  \, w(r^{d-2} u)\,\d \yy \, \d u \\
\ll R^{-1}.
\end{equation}
\end{prop}

\begin{proof}
When $\beta = 1$ the first term in the right hand side of \eqref{eq:apple} vanishes as $v\to 0$, see \cite[\S 6.6.1]{Marklof02}.
By the equidistribution of closed horocycles and the fact that $X_R$ is bounded and piecewise constant, we have for $R\geq 1$ that 
\begin{equation} \begin{split} 
& \lim_{r \to 0} r^{d-2} \int_{\RR} X_R(u+\i r^2)\, w(r^{d-2} u)\, \d u \\
& = \frac{3}{\pi} \int_\RR w(x) \, \d x \int_{\SL(2,\RR)\backslash\H} X_R(u+\i v) \frac{\d u \d v}{v^2} \\
& = \frac{3}{\pi} \int_\RR w(x) \, \d x  \int_R^\infty \frac{\d v}{v^2} = \frac3{\pi R} \int_\RR w(x) \, \d x .
\end{split} \end{equation}
\end{proof}
\begin{prop} \label{aveprop234}
Let $w:\RR\to\RR$ piecewise continuous with compact support, and $0\leq \beta <1$. Then, for every $R\geq 1$,
\begin{equation}
\limsup_{r \to 0} r^{d-2} \int_{\RR} \int_{[0,1)^d}  \Psi_{R,f}^\beta(u+\i r^2,(\begin{smallmatrix} \vecnull \\ \yy \end{smallmatrix}))  \, w(r^{d-2} u)\, \d \yy \, \d u \\
\ll R^{(\beta-1)d/2}  .
\end{equation}
\end{prop}

\begin{proof}
The first term in the right hand side of \eqref{eq:apple} has already been estimated in the proof of Proposition \ref{diophprop234}. For the remaining terms the statement now follows from the observation that $X_R$ is a bounded function. 
\end{proof}

\begin{proof}[Proof of Theorem \ref{thm:allalpha}]
The convergence of the operators $A_n^{(r)}(r^{1-d}t)$ follows in the cases $n=0,1$ directly from the calculations in Section \ref{sec:0and1} for fixed $\vecalf$. Using Proposition \ref{avephprop} one can prove an $\vecalf$-averaged version of Proposition \ref{maincorollary2}, and hence prove the convergence of $A_2^{(r)}(r^{1-d}t)$ as in Corollary \ref{cor:secondorderterms}, with $\vecy=-\vecalf$. All that remains is the bound on the error terms. One first proves the $\vecalf$-averaged version of Proposition \ref{maincorollary22}, with $\vecy_j=-\vecalf$, by using Proposition \ref{aveprop234}. The remaining analysis proceeds identically to Section \ref{sec:error}.
\end{proof}

\appendix

\section{} \label{appendix0}

The following proposition explains how Corollary \ref{conj:allalpha} and Theorem \ref{thm:allalpha} yield information on the phase-space distribution of the wavepacket $f^{(\vecp)}(t)=U_{h,\lambda}(t) f^{(\vecp)}_0$ with an initial wavepacket $f^{(\vecp)}_0$ of the form (cf.~Figure \ref{fig1})
\begin{equation}
f^{(\vecp)}_0(\vecx) = r^{d(d-1)/2} \phi(r^{d-1} \vecx) \, \e(\vecp\cdot\vecx/h) ,
\end{equation}
where $\phi\in\scrS(\RR^d)$ is assumed to have unit $\L^2$-norm, and $\vecp\in\RR^d$. 

We use the shorthand
\begin{equation}
A(t)=U_{h,\lambda}(t) \Op_{r,h} (a) U_{h,\lambda}(t)^{-1}, \qquad B=\Op_{r,h} (b).
\end{equation}

\begin{prop}
Let $f^{(\vecp)}_0$, $f^{(\vecp)}(t)$ as above, $w\in\scrS(\RR^d)$ and $b\in\scrS(\RR^d\times\RR^d$). Set
\begin{equation}\label{A3}
a(\vecx,\vecy)= | \phi(\vecx)|^2  \, w(\vecy) .
\end{equation}
Then
\begin{equation}\label{A4}
r^{-d(d-1)/2} h^{-d/2} \int_{\RR^d} \langle f^{(\vecp)}(t) ,   B\,  f^{(\vecp)}(t)  \rangle \,w(\vecp) \d\vecp 
= \langle A(t) , B \rangle_\HiS +O(r^{d-1}h) ,
\end{equation}
uniformly in $r,h,t>0$.
\end{prop}

(The pre-factor $r^{-d(d-1)/2} h^{-d/2}$ in \eqref{A4} compensates the $\L^2$-normalisation of $B=\Op_{r,h}(b)$ in \eqref{BGscaling}, which is not suitable in the present setting.)

\begin{proof}
Consider the linear operator $F^{(\vecp)}_{r,h}:\L^2(\RR^d)\to\L^2(\RR^d)$ with Schwartz kernel
\begin{equation}
F^{(\vecp)}_{r,h}(\vecx,\vecx') = f^{(\vecp)}_0(\vecx) \, \overline{f^{(\vecp)}_0(\vecx')} 
= r^{d(d-1)} \phi(r^{d-1} \vecx)\, \overline{\phi(r^{d-1} \vecx')} \, \e(\vecp\cdot(\vecx-\vecx')/h) .
\end{equation}
Using the Fourier transform $\widehat w$ of $w$ yields
\begin{equation}
\begin{split}
F_{r,h}(\vecx,\vecx') & = \int F^{(\vecp)}_{r,h}(\vecx,\vecx') w(\vecp) \d\vecp   \\
& = r^{d(d-1)} \phi(r^{d-1} \vecx)\, \overline{\phi(r^{d-1} \vecx')} \, \widehat w((\vecx'-\vecx)/h)
\end{split}
\end{equation}
and by Taylor's theorem we have
\begin{equation}
\phi(r^{d-1} \vecx)\,  = \phi(\tfrac12 r^{d-1}(\vecx+\vecx'))  + R_{r,h}(\vecx,\vecx'),
\end{equation}
with remainder
\begin{equation}
R_{r,h}(\vecx,\vecx') = \tfrac12 r^{d-1} \int_0^1 (\vecx-\vecx') \cdot \nabla \phi(\tfrac12 r^{d-1}((\vecx+\vecx') + s(\vecx-\vecx'))) \d s .
\end{equation}
We can express this term in the form
\begin{equation}\label{eqRRR}
R_{r,h}(\vecx,\vecx')=\tfrac12 r^{d-1} h \, S_b(\tfrac12 r^{d-1}(\vecx+\vecx'), (\vecx-\vecx')/h), \qquad b=\tfrac12 r^{d-1} h,
\end{equation}
with
\begin{equation}
S_b(\vecx,\vecy)= \int_0^1 \vecy \cdot \nabla \phi(\vecx+ s b \vecy) \d s .
\end{equation}
Now
\begin{equation}\label{A11}
\begin{split}
F_{r,h}(\vecx,\vecx') & = r^{d(d-1)} | \phi(\tfrac12 r^{d-1}(\vecx+\vecx'))|^2  \, \widehat w((\vecx'-\vecx)/h) +E_{r,h}(\vecx,\vecx') ,
\end{split}
\end{equation}
with
\begin{multline}
E_{r,h}(\vecx,\vecx') = r^{d(d-1)} \widehat w((\vecx'-\vecx)/h) 
\big\{ \phi(\tfrac12 r^{d-1}(\vecx+\vecx')) \overline{R_{r,h}(\vecx',\vecx)} \\
+ R_{r,h}(\vecx,\vecx') \overline{\phi(\tfrac12 r^{d-1}(\vecx+\vecx'))}
+  R_{r,h}(\vecx,\vecx' ) \overline{R_{r,h}(\vecx',\vecx)} \big\} .
\end{multline}
On account of \eqref{eqRRR},
\begin{equation}
E_{r,h}(\vecx,\vecx') = r^{(d+1)(d-1)} h \, W_b ( \tfrac12 r^{d-1} (\vecx+\vecx'), (\vecx-\vecx')/h),  \qquad b=\tfrac12 r^{d-1} h,
\end{equation}
with
\begin{equation}
W_b ( \vecx,\vecy) = \tfrac12 \widehat w(\vecy) 
\big\{ \phi(\vecx) \overline{S_b(\vecx,-\vecy)} 
+ S_b(\vecx,\vecy) \overline{\phi(\vecx)}
+  b \, S_b(\vecx,\vecy) \overline{S_b(\vecx,-\vecy))} \big\} .
\end{equation}
We re-write \eqref{A11} as
\begin{multline}
F_{r,h}(\vecx,\vecx') = r^{d(d-1)} h^d \int_{\RR^d} | \phi(\tfrac12 r^{d-1}(\vecx+\vecx'))|^2 \, w(h \vecy) \,\e( (\vecx-\vecx')\cdot\vecy) \, \d\vecy \\+E_{r,h}(\vecx,\vecx') ,
\end{multline}
and so, for $a$ as in \eqref{A3},
\begin{equation}
F_{r,h} = r^{d(d-1)/2} h^{d/2} \Op_{r,h}(a) + E_{r,h} .
\end{equation}
We conclude
\begin{equation}
\begin{split}
&  r^{-d(d-1)/2} h^{-d/2} \int_{\RR^d} \langle f^{(\vecp)}(t) , B\, f^{(\vecp)}(t)  \rangle \,w(\vecp) \d\vecp \\
& = r^{-d(d-1)/2} h^{-d/2} \langle U_{h,\lambda}(t) F_{r,h} U_{h,\lambda}(t)^{-1} , B \rangle_\HiS \\
& = \langle U_{h,\lambda}(t) \Op_{r,h}(a) U_{h,\lambda}(t)^{-1} , B \rangle_\HiS + O(r^{d-1}h) ,
\end{split}
\end{equation}
where the error term follows from the upper bounds
\begin{equation}
\big| \langle U_{h,\lambda}(t) E_{r,h} U_{h,\lambda}(t)^{-1} , B \rangle_\HiS \big| 
\leq  \| E_{r,h} \|_\HiS \| B \|_\HiS ,
\end{equation}
and
\begin{equation}
\begin{split}
\| E_{r,h} \|_\HiS & = r^{(d+1)(d-1)} h \bigg(\int_{\RR^d\times\RR^d} \big|W_b( \tfrac12 r^{d-1} (\vecx+\vecx'), (\vecx-\vecx')/h)\big|^2  \,\d\vecx\,\d\vecx' \bigg)^{1/2} \\
& = r^{(1+d/2)(d-1)} h^{d/2+1} \bigg(\int_{\RR^d\times\RR^d} \big|W_b( \vecx, \vecy)\big|^2  \,\d\vecx\,\d\vecy \bigg)^{1/2} \\
%& \ll r^{(1+d/2)(d-1)} h^{d/2+1} .
\end{split}
\end{equation}
with
\begin{equation}
\lim_{b\to 0} \int_{\RR^d\times\RR^d} \big|W_b( \vecx, \vecy)\big|^2  \,\d\vecx\,\d\vecy
=\int_{\RR^d\times\RR^d} \big|W_0( \vecx, \vecy)\big|^2  \,\d\vecx\,\d\vecy <\infty.
\end{equation}
\end{proof}

\section{} \label{appendix1}

In this section we compute the expression \eqref{Is} for $\scrI_{\ell,n}$. Recall that
\begin{multline}
[\hat K_{\ell,n}(\vecs) ] (\vecy,\vecy') \\
= r^{(n-\ell+1)d} \hspace{-0.2cm} \sum_{\vecm_\ell,\ldots,\vecm_n\in\ZZ^d} \e(-\tfrac12 \, s_\ell\|\vecy\|^2) \hat W(r\vecm_\ell) 
\scrT_{\ell,n-1}(\vecy) \e(\tfrac12 \, s_n \|\vecy-\vecm_n \|^2) \delta_{\vecm_n}(\vecy-\vecy').
\end{multline}
Hence we have for $1 \leq \ell \leq n-1$ that
\begin{equation} \begin{split}
\scrI_{\ell,n}(\vecs) &= \Tr_\vecalf[K_{1,\ell}(\vecs) \Op( D_{r,h} a) K_{\ell+1,n}(\vecs)\Op(D_{r,h} b)] \\
&=r^{nd} r^{-d(d-1)} h^d \int_{\RR^d} \sum_{\vecm_0,\dots,\vecm_n}  \\
& \times \e(-\tfrac12 \, s_1\|\vecm_0+\vecalf\|^2) \hat W(r\vecm_1) 
\scrT_{1,\ell-1}(\vecm_0+\vecalf) \e(\tfrac12 \, s_\ell \|\vecm_0+\vecalf-\vecm_\ell \|^2) \\
& \times \tilde a (r^{1-d}(\vecm_0+\vecalf-\vecm_\ell - \veceta), \tfrac h2 (\vecm_0+\vecalf-\vecm_\ell+\veceta)) \\
& \times \e(-\tfrac12 \, s_{\ell+1}\|\veceta \|^2) \hat W(r\vecm_{\ell+1}) 
\scrT_{\ell+1,n-1}(\veceta ) \e(\tfrac12 \, s_n \|\veceta -\vecm_n \|^2) \\
& \times \tilde b( r^{1-d} (\veceta-\vecm_n - \vecm_0-\vecalf), \tfrac h2 (\veceta-\vecm_n + \vecm_0+\vecalf)) \,  \d \veceta.
\end{split} \end{equation}
We then make the variable substitution $\veceta \to r^{d-1} \veceta + \vecm_0 +\vecalf - \vecm_\ell$ so that $\tilde a$ has first argument $-\veceta$. This leaves $\tilde{b}$ with first argument $\veceta - r^{1-d}(\vecm_n+\vecm_\ell)$, and by the rapid decay of $\tilde a$ and $\tilde b$ the leading order terms come from when $\vecm_n+\vecm_\ell=0$, and we incur an error of order $r^\infty$. We thus have
\begin{equation} \begin{split}
\scrI_{\ell,n}(\vecs) &=r^{nd} h^d \int_{\RR^d} \sum_{\vecm_0,\dots,\vecm_{n}} \, \mathbb{1}[\vecm_n+\vecm_\ell=0]  \\
& \times \e(-\tfrac12 \, s_1\|\vecm_0+\vecalf\|^2) \hat W(r\vecm_1) 
\scrT_{1,\ell-1}(\vecm_0+\vecalf) \e(\tfrac12 \, s_\ell \|\vecm_0+\vecalf-\vecm_\ell \|^2) \\
& \times \tilde a (-\veceta, h (\vecm_0+\vecalf-\vecm_\ell+\tfrac12 r^{d-1}\veceta)) \,  \e(-\tfrac12 \, s_{\ell+1}\|\vecm_0+\vecalf -\vecm_\ell + r^{d-1} \veceta\|^2) \\
& \times\hat W(r\vecm_{\ell+1}) 
\scrT_{\ell+1,n-1}(\vecm_0+\vecalf + r^{d-1} \veceta - \vecm_\ell ) \e(\tfrac12 \, s_n \|\vecm_0+\vecalf + r^{d-1} \veceta \|^2) \\
& \times \tilde b( \veceta, h (\vecm_0+\vecalf + \tfrac12 r^{d-1} \veceta)) \, \d \veceta  + O(r^\infty).
\end{split} \end{equation}
Finally, we make the substitutions $\vecm_j \to \vecm_0 - \vecm_j$ for $j=1,\dots,\ell$ followed by $\vecm_j \to \vecm_\ell - \vecm_{j}$ for $j= \ell+1, \dots, n$ to obtain
\begin{equation} \begin{split}
\scrI_{\ell,n}(\vecs) &=r^{nd} h^d \int_{\RR^d} \sum_{\vecm_0,\dots,\vecm_{n}} \, \mathbb{1}[\vecm_n=\vecm_0]  \\
& \times \e(-\tfrac12 \, s_1\|\vecm_0+\vecalf\|^2) \hat W(r(\vecm_0-\vecm_1)) 
\scrT_{1,\ell-1}^-(\vecalf) \e(\tfrac12 \, s_\ell \|\vecm_\ell+\vecalf \|^2) \\
& \times \tilde a (-\veceta, h (\vecm_\ell+\vecalf+\tfrac12 r^{d-1}\veceta)) \,  \e(-\tfrac12 \, s_{\ell+1}\|\vecm_\ell+\vecalf + r^{d-1} \veceta\|^2) \\
& \times\hat W(r(\vecm_{\ell}-\vecm_{\ell+1})) 
\scrT_{\ell+1,n-1}^-(\vecalf + r^{d-1} \veceta) \e(\tfrac12 \, s_n \|\vecm_0+\vecalf + r^{d-1} \veceta \|^2) \\
& \times \tilde b( \veceta, h (\vecm_0+\vecalf + \tfrac12 r^{d-1} \veceta)) \, \d \veceta + O(r^\infty) .
\end{split} \end{equation}
This proves \eqref{Is}.

\section{}\label{appendix2}
This section establishes relation \eqref{eq:compli}, which is needed in the analysis of $\scrJ_{\ell,n}(t,a)$. First we compute the kernel of $\hat K_{\ell,n}^\dagger=\scrF K_{\ell,n}^\dagger \scrF^{-1}$. By taking complex conjugate and switching $\vecy$ and $\vecy'$ in \eqref{Kker}, we obtain
\begin{multline}
[\hat K_{\ell,n}(\vecs)^\dagger] (\vecy,\vecy') 
= r^{(n-\ell+1)d} \hspace{-0.2cm} \sum_{\vecm_\ell',\ldots,\vecm_n' \in\ZZ^d} \e(\tfrac12 \, s_\ell\|\vecy + \vecm_n'\|^2) \hat W(-r\vecm_\ell') \\ \times
\overline{\scrT}_{\ell,n-1}(\vecy + \vecm_n') \e(-\tfrac12 \, s_n \|\vecy \|^2) \delta_{\vecm_n'}(\vecy'-\vecy),
\end{multline}
where
\begin{equation}
\overline{\scrT}_{\ell,n}(\vecy) 
=\prod_{j=\ell}^{n} \e(\tfrac12 \, (s_{j+1}-s_j) \|\vecy-\vecm_j' \|^2) \hat W(r(\vecm_{j}'-\vecm_{j+1}')) .
\end{equation}
Thus, using the formulae for the kernels of $\hat K_{\ell,n}$, $\hat K_{\ell,n}^\dagger$ and $\hOp_{r,h}$ we have that
%\begin{multline}
%[K_{\ell,n}(\vecs)^\dagger K_{\ell,n}(\vecs) ] (\vecy,\vecy') \\
%= r^{2(n-\ell+1)d} \hspace{-0.2cm} \sum_{\vecm_\ell,\ldots,\vecm_n\in\ZZ^d} \sum_{\vecm_\ell',\ldots,\vecm_n'\in\ZZ^d} \e(\tfrac12 \, s_\ell\|\vecy + \vecm_n\|^2) \hat W(-r\vecm_\ell) 
%\overline{\scrT}_{\ell,n-1}(\vecy + \vecm_n) \e(-\tfrac12 \, s_n \|\vecy \|^2) \\
%\e(-\tfrac12 \, s_\ell\|\vecy + \vecm_n \|^2) \hat W(r\vecm_\ell') 
%\scrT_{\ell,n-1}(\vecy + \vecm_n) \e(\tfrac12 \, s_n \|\vecy + \vecm_n-\vecm_n' \|^2) \delta_{\vecm_n'-\vecm_n}(\vecy-\vecy')
%\end{multline}
%and also
\begin{equation} \begin{split}
& [\hat K_{\ell,n}(\vecs)^\dagger \hat K_{\ell,n}(\vecs) \, \hOp_{r,h}(a)] (\vecy,\vecy') \\
& = r^{2(n-\ell+1)d} \hspace{-0.2cm} \sum_{\vecm_\ell,\ldots,\vecm_n\in\ZZ^d} \sum_{\vecm_\ell',\ldots,\vecm_n'\in\ZZ^d} \\
& \times \e(\tfrac12 \, s_\ell\|\vecy + \vecm_n'\|^2) \hat W(-r\vecm_\ell') 
\overline{\scrT}_{\ell,n-1}(\vecy + \vecm_n') \e(-\tfrac12 \, s_n \|\vecy \|^2) \\
& \times \e(-\tfrac12 \, s_\ell\|\vecy + \vecm_n' \|^2) \hat W(r\vecm_\ell) 
\scrT_{\ell,n-1}(\vecy + \vecm_n') \e(\tfrac12 \, s_n \|\vecy + \vecm_n'-\vecm_n \|^2) \\
& \times \tilde a(r^{1-d} (\vecy - \vecm_n + \vecm_n' -\vecy'),\tfrac h2 (\vecy - \vecm_n + \vecm_n'+\vecy')),
\end{split}\end{equation}
and similarly
\begin{equation} \begin{split}
& [\hat K_{\ell,n}(\vecs) \hat K_{\ell,n}(\vecs)^\dagger \,\hOp_{r,h}(a)] (\vecy,\vecy') \\
& = r^{2(n-\ell+1)d} \hspace{-0.2cm} \sum_{\vecm_\ell,\ldots,\vecm_n\in\ZZ^d} \sum_{\vecm_\ell',\ldots,\vecm_n'\in\ZZ^d} \\
& \times \e(-\tfrac12 \, s_\ell\|\vecy\|^2) \hat W(r\vecm_\ell) 
\scrT_{\ell,n-1}(\vecy) \e(\tfrac12 \, s_n \|\vecy-\vecm_n \|^2) \\
& \times \e(\tfrac12 \, s_\ell\|\vecy-\vecm_n + \vecm_n'\|^2) \hat W(-r\vecm_\ell') 
\overline{\scrT}_{\ell,n-1}(\vecy-\vecm_n + \vecm_n') \e(-\tfrac12 \, s_n \|\vecy-\vecm_n \|^2) \\
& \times  \tilde a(r^{1-d} (\vecy + \vecm_n'-\vecm_n-\vecy'),\tfrac h2 (\vecy + \vecm_n'-\vecm_n+\vecy')).
\end{split}\end{equation}

Combining these yields explicitly 
\begin{equation}
\begin{split}
\Tr_\vecalf &  [K_{1,\ell}(\vecs)^\dagger K_{1,\ell}(\vecs)\,  \Op_{r,h}(a)\, K_{\ell+1,n}(\vecs) K_{\ell+1,n}(\vecs)^\dagger \,\Op_{r,h}(\bar a)]\\
&= r^{2nd-d(d-1)} h^{d}  \sum_{\substack{\vecm_0,\vecm_1,\ldots,\vecm_n\in\ZZ^d\\ \vecm_1',\ldots,\vecm_n'\in\ZZ^d}} \\
& \times \int_{\RR^d}  \hat W(r\vecm_1) 
\scrT_{1,\ell-1}(\vecm_0+\vecm_\ell'+\vecalf) \e(\tfrac12 \, s_\ell \|\vecm_0+\vecm_\ell'-\vecm_\ell+\vecalf \|^2)    \\
& \times  \hat W(-r\vecm_1') 
\overline{\scrT}_{1,\ell-1}(\vecm_0+\vecm_\ell'+\vecalf) \e(-\tfrac12 \, s_\ell \|\vecm_0+\vecalf \|^2) \\
& \times \tilde a(r^{1-d} (\vecm_0+\vecm_\ell'-\vecm_\ell+\vecalf -\vecy),\tfrac{h}2(\vecm_0+\vecm_\ell'-\vecm_\ell +\vecalf+\vecy)) \\
& \times\e(-\tfrac12 \, s_{\ell+1}\|\vecy\|^2) \hat W(r\vecm_{\ell+1}) 
\scrT_{\ell+1,n-1}(\vecy)   \\
& \times \e(\tfrac12 \, s_{\ell+1}\|\vecy+\vecm_n'-\vecm_n\|^2) \hat W(-r\vecm_{\ell+1}') 
\overline{\scrT}_{\ell+1,n-1}(\vecy+\vecm_n'-\vecm_n)   \\
& \times  \tilde {\overline{a}}(r^{1-d} (\vecy+\vecm_n'-\vecm_n-\vecm_0-\vecalf),\tfrac{h}2(\vecy+\vecm_n'-\vecm_n+\vecm_0+\vecalf)) \, \d \vecy .
\end{split}
\end{equation} 
Now we make the substitution $\vecy = r^{d-1} \veceta + \vecm_0 + \vecalf + \vecm_\ell'-\vecm_\ell $ so that the first argument of $\tilde a$ becomes $-\veceta$.  Now $\tilde{\overline{a}}$ has first argument $\veceta + r^{1-d}( \vecm_n'-\vecm_n + \vecm_\ell'-\vecm_\ell)$, and hence (using the rapid decay of $\tilde a$) we have that $\vecm_n' - \vecm_n + \vecm_\ell' - \vecm_\ell =0$. This yields the expression
\begin{equation}
\begin{split}
\Tr_\vecalf &[K_{1,\ell}(\vecs)^\dagger K_{1,\ell}(\vecs)\, \Op_{r,h}(a)\, K_{\ell+1,n}(\vecs) K_{\ell+1,n}(\vecs)^\dagger \,\Op_{r,h}(\bar a)] \\
&= r^{2nd} h^{d}  \sum_{\substack{\vecm_0,\vecm_1,\ldots,\vecm_n\in\ZZ^d\\ \vecm_1',\ldots,\vecm_n'\in\ZZ^d}} \mathbb{1}[ \vecm_n' - \vecm_n + \vecm_\ell' - \vecm_\ell =0] \\
& \times \int_{\RR^d} \hat W(r\vecm_1) 
\scrT_{1,\ell-1}(\vecm_0+\vecm_\ell'+\vecalf) \e(\tfrac12 \, s_\ell \|\vecm_0+\vecm_\ell'-\vecm_\ell+\vecalf \|^2)    \\
& \times  \hat W(-r\vecm_1') 
\overline{\scrT}_{1,\ell-1}(\vecm_0+\vecm_\ell'+\vecalf) \e(-\tfrac12 \, s_\ell \|\vecm_0+\vecalf \|^2) \\
& \times \tilde a(-\veceta,h(\vecm_0+\vecm_\ell'-\vecm_\ell +\vecalf+ \tfrac12 r^{d-1} \veceta)) \\
& \times\e(-\tfrac12 \, s_{\ell+1}\| r^{d-1} \veceta + \vecm_0 + \vecalf + \vecm_\ell'-\vecm_\ell \|^2) \hat W(r\vecm_{\ell+1}) \\
& \times \scrT_{\ell+1,n-1}( r^{d-1} \veceta + \vecm_0 + \vecalf + \vecm_\ell'-\vecm_\ell )   \\
& \times \e(\tfrac12 \, s_{\ell+1}\| r^{d-1} \veceta + \vecm_0 + \vecalf \|^2) \hat W(-r\vecm_{\ell+1}') 
\overline{\scrT}_{\ell+1,n-1}( r^{d-1} \veceta + \vecm_0 + \vecalf )   \\
& \times  \tilde {\overline{a}}(\veceta,h(  \vecm_0 + \vecalf + \tfrac12 r^{d-1} \veceta )) \, \d \veceta + O(r^\infty).
\end{split}
\end{equation}
We then make the substitution $\vecm_0 \to \vecm_0 - \vecm_\ell'$, followed by the substitutions $\vecm_j \to \vecm_0 - \vecm_j$ for $j=1,\dots,\ell$ and $\vecm_j \to \vecm_\ell - \vecm_j$ for $j=\ell+1, \dots, n$ as well as the analogous substitutions for the $\vecm_j'$. This yields the simpler expression

\begin{equation}
\begin{split}
\Tr_\vecalf & [K_{1,\ell}(\vecs)^\dagger K_{1,\ell}(\vecs)\, \Op_{r,h}(a)\, K_{\ell+1,n}(\vecs) K_{\ell+1,n}(\vecs)^\dagger \, \Op_{r,h}(\bar a)] \\&= r^{2nd} h^{d}  \sum_{\substack{\vecm_0,\vecm_1,\ldots,\vecm_n\in\ZZ^d\\ \vecm_1',\ldots,\vecm_n'\in\ZZ^d}} \mathbb{1}[ \vecm_n' - \vecm_n + \vecm_\ell -\vecm_\ell'= 0]  \\
& \times  \int_{\RR^d} \hat W(r(\vecm_0-\vecm_1)) \scrT_{1,\ell-1}^-(\vecalf)   \e(\tfrac12 \, s_\ell (\|\vecm_\ell+\vecalf \|^2-\|\vecm_\ell'+\vecalf \|^2))     \\
& \times \hat W(r(\vecm_1'-\vecm_0)) \overline{\scrT}_{1,\ell-1}^-(\vecalf)  \tilde a(-\veceta,h(\vecm_\ell +\vecalf+ \tfrac12 r^{d-1} \veceta)) \\
& \times \hat W(r(\vecm_\ell-\vecm_{\ell+1})) \scrT_{\ell+1,n-1}^-(r^{d-1} \veceta + \vecalf )    \\
& \times \e(\tfrac12 \, s_{\ell+1}(\| r^{d-1} \veceta + \vecm_\ell' + \vecalf \|^2- \| r^{d-1} \veceta + \vecalf + \vecm_\ell \|^2))  \\
& \times \hat W(r(\vecm_{\ell+1}'-\vecm_\ell')) \overline{\scrT}_{\ell+1,n-1}^-(r^{d-1} \veceta+ \vecalf )    \tilde {\overline{a}}(\veceta,h(  \vecm_\ell'+ \vecalf + \tfrac12 r^{d-1} \veceta )) \, \d \veceta \\
& + O(r^\infty).
\end{split}
\end{equation}
This yields \eqref{eq:compli} after substituting $\veceta \to - \veceta$.

\end{document}